%% file: main.tex
\documentclass[letterpaper,onecolumn,draftclsnofoot]{IEEEtran}
\usepackage[utf8]{inputenc} 
\usepackage[numbers,compress]{natbib}
\usepackage[cmex10]{amsmath}
\usepackage{amsfonts,mathrsfs,mathdots,mathtools,amsthm,amssymb,centernot}
\usepackage[bb=ams]{mathalfa}
\usepackage{algorithmic}
\usepackage{graphicx}
\usepackage{textcomp}
\usepackage[dvipsnames]{xcolor}
\usepackage{tikz}
\usepackage{pgfplots}
\usepackage{subcaption}
\usepackage{dsfont}
\usepackage{pifont}
\usepackage{siunitx}
\usepackage{comment}
\usepackage{pgfplotstable}
\usepackage[ruled]{algorithm2e}
\usepackage{color,soul}
\usepackage{bm}	
\usepackage[inline]{enumitem}
\setlist[enumerate,1]{label = \arabic*.,ref = \arabic*}
\usepackage{nicefrac}
\usepackage{subcaption}
\usepackage{colortbl}
\usepackage{booktabs}
\usepackage{diagbox}
\usepackage{url}
\usepackage{ifthen}
\usepackage{balance}
\usepackage{xfrac}
\usepackage{dirtytalk}

\usepackage{hyperref}  
\hypersetup{colorlinks,
            linkcolor=blue,
            citecolor=blue,
            urlcolor=blue,
            anchorcolor=blue,
            filecolor=blue,
            linktocpage,
            plainpages=false,
            breaklinks=true}

\input{macros}
\begin{document}
\title{Contraction and Statistical Inference under Privacy for Uniformly Bounded Distributions} 

\author{Leonhard~Grosse,~\IEEEmembership{Member,~IEEE,}
        Sara~Saeidian,~\IEEEmembership{Member,~IEEE,}
        Tobias~J.~Oechtering,~\IEEEmembership{Senior~Member,~IEEE}
        Mikael~Skoglund,~\IEEEmembership{Fellow,~IEEE,}
\thanks{Parts of this article have been presented at the 2026 IEEE International Symposium on Information Theory.}        
}

\maketitle

\begin{abstract}
We investigate $c$-interior \emph{pointwise maximal leakage} (PML) as a tool for contraction analyses and disclosure control. Based on the strong adversarial threat models from maximal leakage, $c$-interior PML generalizes \emph{local differential privacy} (LDP) to data-generating distributions with densities uniformly bounded away from zero by $c>0$. Viewing $c$-interior PML as an algebraic constraint on a kernel yields more flexible (and often tighter) contraction analyses than standard LDP. We provide tight bounds on the Dobrushin coefficient, and bound the contraction coefficient of the Hockeystick-divergence. We further derive strong data processing inequalities on  $f$-divergences under $c$-interior PML constraints when the input distributions to the divergence are restricted to be in the $c$-interior. These results extend beyond the regime of pure LDP to cover a larger class of kernels, including, e.g., arbitrary stochastic matrices. We apply the results to minimax theory and provide asymptotically optimal strategies under $c$-interior PML constraints for binary hypothesis testing and mean estimation. The results show that disclosure control with PML allows analysts to reason about systems in a more differentiated manner: For example, it allows us to quantify the privacy leakage of deterministic systems, and can give precise adversarial guarantees with respect to arbitrary distributional assumptions. Interestingly, a recurring theme in the disclosure analyses is that if the privacy problem is relatively regular (if the density bound $c$ is large), private inference can be possible without incurring any additional cost in terms of sample complexity.   
\end{abstract}


\setcounter{tocdepth}{2}
\tableofcontents

\newpage

\section{Introduction}
\input{sections/intro}

\section{Preliminaries}
\label{sec:background}
\input{sections/background}

\section{Dobrushin Coefficients and PML Privacy Constraints}
\label{sec:dobrushin}
\input{sections/dobrushin}
\section{Contraction of $E_\gamma$-Divergence under PML Constraints}
\label{sec:E_gamma}
\input{sections/E_gamma}

\section{Strong Data Processing Inequalities with Privacy for Uniformly Bounded Distributions}
\label{sec:sdpi}
\input{sections/sdpi}

\section{Applications}
\label{sec:minimax}
\input{sections/minimax}

\subsection{Other Applications}
\label{sec:otherapplications}
\input{sections/otherapplications}

\section{Conclusions}
\input{sections/conclusions}

\appendices
\input{sections/appendix}

\bibliographystyle{IEEEtranN}
\footnotesize
\balance
\bibliography{main}

\end{document}

%% file: macros.tex
\DeclareMathOperator*{\esssup}{ess\,sup}

\def\supp{\textnormal{supp}}

\theoremstyle{definition}
\newtheorem{definition}{Definition}
\newtheorem{lemma}{Lemma}
\newtheorem{theorem}{Theorem}

\newtheorem{corollary}{Corollary}
\newtheorem{proposition}{Proposition}
\newtheorem{example}{Example}
\newtheorem{remark}{Remark}
\newtheorem{claim}{Claim}

\newcommand{\bR}{\ensuremath{\mathbb{R}}}

\newcommand{\cX}{\ensuremath{\mathcal{X}}}
\newcommand{\cY}{\ensuremath{\mathcal{Y}}}

\newcommand{\ind}{\ensuremath{\mathbf{1}}}

\mathtoolsset{showonlyrefs}
\allowdisplaybreaks

%% file: sections/intro.tex
Performing statistical tasks subject to privacy constraints usually results in increased sample complexity \cite{duchi2013local,asoodeh2024contraction,duchi2024right,dinur2003revealing,dwork2010difficulties}. Commonly, privacy in these settings is measured by \emph{differential privacy} (DP) \cite{DPoriginalpaper} or \emph{local differential privacy} (LDP) \cite{duchi2013LDPminmaxDEF}. In the local setting a (discrete) kernel $\mathsf K$ satisfies $(\alpha,\delta)$-local differential privacy for private data $X$ and privatized data $Y$, if,
\begin{equation}
\label{eq:LDPdef}
    \mathsf K(y|x) \leq e^\alpha \mathsf K(y|x')+\delta \quad \forall x,x'\in\mathcal X,\,\forall y\in \mathcal Y.
\end{equation}
The measure has found extensive application in practical implementations in recent years (see, e.g., \cite{googleRAPPOR,appleDP,ghazi_et_al:LIPIcs.ITCS.2025.53,abowd2018us,chen2025scalableprivatepartitionselection}). However, especially in the local setting, various works show that the effect of privacy constraints on estimation tasks can be quite drastic \cite{duchi2013local,duchi2024right,asoodeh2024contraction}, e.g., reducing the effective sample complexity from $n$ to $\alpha^2 n$ when $\alpha$ is small. Further, and likely due to this drastic decrease in sample efficiency, many practical implementations resort to choosing large privacy parameters \cite{desfontainesblog20211001}, resulting in weak and difficult to interpret privacy protection for users \cite{dwork2019exposeyourepsilons}. This issue of balancing good privacy and utility seems to be fundamental in some important practical settings \cite{ertan2026fundamentallimitationsfavorableprivacyutility}. Together with some issues with, e.g., correlated data \cite{saeidian2024evaluating,saeidian2023inferential}, this motivates us to move beyond the scope of the definition of differential privacy to explore the properties of alternative (but related) privacy measures.

It has recently been observed that LDP can be derived in the framework of \emph{pointwise maximal leakage} (PML) \cite{saeidian2023pointwise}. PML is an operational privacy measure based on adversarial threat models, and measures the inference risk incurred by releasing (possibly privatized) data with respect to the wide range of adversaries considered in its formulation. It was shown in \cite[Theorem 14]{IssaMaxL} that guaranteeing LDP is equivalent to protecting data from the adversaries in the PML framework for \emph{any} possible data-generating distribution, where the worst-case adversarial gain is attained when the data is distributed according to a point density (a distribution containing zero-probability outcomes). While this is good news for robustness, it also provides an underlying reason for the conservative nature of LDP. In fact, taking into account knowledge about the data-generating distribution can often significantly improve the privacy-utility tradeoff in the PML framework by allowing for more bespoke mechanism design \cite{grosse2025privacy,10646583}.

From a more theoretical standpoint, the definition of $\alpha$-LDP in \eqref{eq:LDPdef} also exhibits a somewhat converse issue when used as a mathematical tool to describe stochastic systems: As an example, consider a kernel in the form of a stochastic matrix that contains zero elements. From \eqref{eq:LDPdef}, it becomes clear that any such kernel can only be described by $(\alpha,\delta)$-LDP with nonzero $\delta$ (that is, \emph{pure} LDP cannot hold). This is intentional: if we want to make deterministic disclosure of an input realization entirely impossible under \emph{any} input distribution, even a single zero element in the stochastic matrix is to be prohibited. However, this can lead to somewhat unintuitive behaviors in the broader perspective. Consider the two kernels,
\begin{equation}
   \mathsf K_1= \begin{bmatrix}
        \frac{1}{3} & \frac{1}{3} & \frac{1}{3} & 0 \\
        \frac{1}{3}& \frac{1}{3} &0 & \frac{1}{3} \\
        \frac{1}{3} & 0 & \frac{1}{3} & \frac{1}{3} \\
        0 & \frac{1}{3} & \frac{1}{3} & \frac{1}{3}
    \end{bmatrix},
    \quad \mathsf K_2 = \begin{bmatrix}
        1 & 0 & 0 & 0 \\
        0 & 1 & 0 & 0 \\
        0 & 0 & 1 & 0 \\
        0 & 0 & 0 & 1
    \end{bmatrix}.
\end{equation}
Since there exists $y,x,x'$ for which $\mathsf K_1(y|x) = 0$ and $\mathsf K_1(y|x')=\nicefrac{1}{3}$, $\mathsf K_1$ cannot satisfy any finite pure LDP guarantee. This  renders the kernels $\mathsf K_1$ and $\mathsf K_2$ \emph{equivalent} through the lens of pure LDP. While this is consistent with the privacy quantification of LDP that prohibits deterministic disclosure of any input, it is reasonable to argue the $\mathsf K_1$ and $\mathsf K_2$ constitute significantly different stochastic systems. As an algebraic analysis tool, LDP therefore falls short of characterizing the systems in this example.\footnote{One might assume that this is remedied by using $(\varepsilon,\delta)$-LDP (approximate LDP) instead. However, we show in Section \ref{sec:otherapplications} that even $(\varepsilon,\delta)$-LDP has significant difficulties in distinguishing systems as in the above example. More specifically, the characterization via $(\varepsilon,\delta)$-LDP for, e.g., stochastic matrices with zero entries will often collapse to a $(0,\delta)$-LDP guarantee. Further, it has been shown in, e.g., \cite{bun2019heavy,zamanlooy2024mathrm}, that $(\varepsilon,\delta)$-LDP can always be reduced to $\varepsilon$-LDP. Hence, contrary to the centralized setting, the introduction of an additional slack parameter in approximate LDP offers no advantage in terms of utility in the local model.}
Further, while $\mathsf K_2$ offers no privacy compared to simply releasing the realization of the secret $X$ directly, it is entirely reasonable to argue that $\mathsf K_1$ offers a level of privacy protection by significantly randomizing the input realizations of $X$. Still, both $\mathsf K_1$ and $\mathsf K_2$ are equivalently \say{non-private} in the framework of LDP. 

 In this work, we present the framework of PML as a flexible tool for the analysis of stochastic systems. From a privacy standpoint, we move away from the worst-case distributional assumptions in LDP, and examine the effect of privacy constraints on estimation and testing when the densities $p_X$ have a degree of regularity, that is, when there exists a positive constant $c>0$ such that $\inf_x p_X(x) \geq c$. From an analytical standpoint, we present refined contraction analyses of kernels for various divergence measures via the PML privacy measure. In a contraction analysis, we aim to answer the following question: For any two distributions, what can we say about the relation between input and output divergences when these two distributions are passed through a stochastic system? This question has been extensively investigated in the information theory literature (see \cite[Chapter 33]{Polyanskiy_Wu_2025} for a somewhat detailed introduction). In the privacy domain, recent works show that contraction analyses are highly useful tools when deriving privacy guarantees for iterative algorithms \cite{asoodeh2020privacy,asoodeh2024privacy,feldman2018privacy}, as well as for establishing fundamental bounds on statistical estimation tasks \cite{duchi2013local,duchi2024right,asoodeh2024contraction}. See Section \ref{sec:background:subsec:privacyandcontraction} on related works for more details.
 
\subsection{Our Contributions}
The goal of this work is to highlight two different aspects of PML privacy. In the first half of this work, we argue that a notion of privacy which we call \emph{$c$-interior PML} constitutes a powerful analytical tool for contraction analyses. In short, $c$-interior PML limits the leakage in the PML sense by a constant $\varepsilon> 0$ for all distributions in the $c$-interior of the probability simplex $\mathcal P_c = \{P_X\in\mathcal P(\mathcal X): P_X \geq c\}$. In the latter half, we apply these contraction results to privacy assessment and disclosure control of hypothesis testing and estimation procedures. The results show that statistical inference can be done privately at no or very low cost if a reasonable degree of regularity in form of the parameter $c$ is guaranteed. That is, if the minimum mass/density is bounded away from zero by a sufficiently large constant, relatively strong protection against the adversarial attacks considered in the PML framework is achieved by estimation procedures that are optimal even in the non-private setting. Even when $c$ is not large, or the privacy guarantee is strict, assuming a lower bound on the densities can enable improved estimation compared to the worst-case distributional assumptions in LDP. We now provide a more detailed summary of the results in this work.

Section \ref{sec:background} lays out important definitions and previous results, as well as a summary of related works most relevant to this one. We begin Section \ref{sec:dobrushin} by providing conversion rules between $c$-interior PML and pure LDP in Lemma \ref{lem:PMLimpliesLDP}. Importantly, while we show that conversions exists, we highlight that these conversions are loose. Further, they require more than just a parameter transformation; there are many kernels that do not satisfy \emph{any} finite LDP guarantee, while at the same time satisfying an infinite range of $c$-interior PML guarantees with finite $\varepsilon$ and $c$. Therefore, the presented conversions require a more careful technical setup, and often work on transformed kernels. These results allow us to transfer existing results from the LDP literature into the PML framework. 

Theorem \ref{thm:TVcontraction} constitutes the main technical result of this work: We show that if the PML of a kernel $\mathsf K$ acting on the space $\mathcal X$ is bounded from above by $\varepsilon$ for any distribution in the $c$-interior of the simplex, then the total variation distance between the output marginals $P_X\mathsf K$ and $Q_X\mathsf K$ induced by the input distributions $P_X$ and $Q_X$ satisfies,
\begin{equation}
    \text{TV}(P_X\mathsf K||Q_X\mathsf K) \leq \underbrace{\min\left\{1\,,\frac{e^\varepsilon-1}{e^\varepsilon(1-c\mu_X(\mathcal X))+1}\right\}}_{\coloneqq \eta_\text{TV}(\varepsilon)}\text{TV}(P_X||Q_X),
\end{equation}
where $\mu_X$ denotes the reference measure of the underlying space (e.g. the counting measure in the discrete setting). We also show that this bound is tight (and improves upon existing results) both when $c$-interior PML is only used as an algebraic property ($P_X$, $Q_X$ arbitrary distributions), and when $P_X$ and $Q_X$ are selected only from $\mathcal P_c$. We further extend the results and proof technique to the Hockeystick-divergence $E_\gamma$ in Theorem \ref{thm:E_gamma_contraction} of Section \ref{sec:E_gamma}. Finally, in Section \ref{sec:sdpi}, we derive strong data processing inequalities on the relative entropy and Hellinger divergence when a kernel satisfies a $c$-interior PML constraint, and input distributions are picked from the $c$-interior only. These results quantify how divergence measures change when a kernel with PML constraints is applied, and can therefore be used to reason about the effects of privacy for statistical inference, etc., or to characterize the stochastic behavior of a kernel, e.g., in the analysis of mixing times \cite{zamanlooy2024mathrm, daeijavad2026local}. 

During the exposition of the technical results, we highlight the key differences between the two different approaches for obtaining contraction bounds presented (via LDP conversion vs. directly via PML). As we will see, there are many places in which the latter provides added flexibility that results in tighter characterizations. Further, we traverse through two conceptually different approaches along another axis: For most of the analyisis of contraction coefficients, we treat $c$-interior PML as a purely algebraic property. That is, the input distributions to the divergences are free to be picked from the entire simplex. As we progress through the work, we will then use the results to obtain \emph{set-restricted} SDPIs, where now the choice of input distributions is also limited to the $c$-interior $\mathcal P_c$. There are key differences between these two approaches. As a particularly extreme example, when we pick $P_X$ and $Q_X$ arbitrarily from the simplex, we might often encounter unbounded relative entropy, i.e.,
\begin{equation} 
    D(P_X\mathsf K||Q_X\mathsf K) = \infty,
\end{equation}
even when the kernel $\mathsf K$ satisfies a non-trivial privacy constraint. On the other hand, when the choice of $P_X$ and $Q_X$ is restricted to the $c$-interior, we are able to show that any $f$-divergence (including relative entropy) between output marginals will remain bounded. Specifically, Theorem \ref{thm:binettefdivbound} shows that for some constant $C$ depending only on $f$, $\varepsilon$, and $c$, we have,
\begin{equation}
    D_f(P_X\mathsf K||Q_X\mathsf K) \leq C(f,\varepsilon,c)\,\text{TV}(P_X||Q_X),
\end{equation}
whenever $P_X,Q_X$ are picked from $\mathcal P_c$. Further, in Theorem \ref{thm:chi^2_TV^2_nr2}, we show that for relative entropy, we also have,
\begin{equation}
    D(P_X\mathsf K||Q_X\mathsf K) \leq \Xi(\varepsilon,c)\,\text{TV}^2(P_X||Q_X),
\end{equation}
where $\Xi(\varepsilon,c) = \mathcal O(\min\{e^\varepsilon,\varepsilon^2\})$. This shows that when $P_X$ and $Q_X$ are picked arbitrarily from $\mathcal P_c$, then the application of a kernel $\mathsf K$ transforms the relative entropy into a squared-variation-type distance, and hence recovers a property held by LDP on the entire simplex \cite{duchi2013local}.

In the second half of this work, we use the technical results from Part I to analyze the effect of PML privacy on two important statistical tasks: binary hypothesis testing and mean estimation. For the hypothesis testing problem, Theorem~\ref{thm:hypothesistestingH^2asymp} shows that the number of samples needed to privately distinguish between $P_X$ and $Q_X$ in $\mathcal P_c$, denoted by $n^*_{\varepsilon,c}(P_X,Q_X,\zeta)$, scales according to,
\begin{equation}
    n^*_{\varepsilon,c}(P_X,Q_X,\zeta) \asymp \frac{1}{\eta_\text{TV}^2(\varepsilon)\text{TV}^2(P_X||Q_X)}.
\end{equation}
Further, we show that in certain settings, this is achieved by a simple and efficiently implementable privatized version of the standard optimal test. Note that $\eta_\text{TV}(\varepsilon) = 1$ for many mechanisms with non-trivial adversarial guarantees. The results therefore show that we privately achieve the non-private risk in these cases.\footnote{This improves upon an observation for LDP-private hypothesis testing in \cite{pensia2023simple,asoodeh2024contraction}, see Section \ref{sec:minimax:subsec:nhypo} for details.} We further analyze the one-shot Bayesian testing setting via the Hockeystick-divergence in Proposition \ref{prop:1hypotestEgamma}.

For mean estimation of distributions supported on the $d$-dimensional unit ball of norm $p\in[1,2]$, we show in Propositions~\ref{prop:meanestimation}-\ref{prop:highd-estimator} that the minimax risk under $c$-interior PML for small enough $\varepsilon$ follows,
\begin{equation}
    \mathfrak R^{\varepsilon,c}_n(\theta(\mathcal P),||\cdot||_2^2) \asymp \frac{d}{n\,\eta^2_\text{TV}(\varepsilon)}.
\end{equation}
Under some technical conditions, we show that if $\varepsilon$ is large, then we can in fact (privately) achieve the non-private asymptotic risk,
\begin{equation}
    \mathfrak R^{\varepsilon,c}_n(\theta(\mathcal P),||\cdot||_2^2) \lesssim \frac{1}{n}.
\end{equation}
In both cases, the upper bound follows from constructing corresponding order-optimal estimators. 
 
Finally, we use the technical results from the first half to provide two applications where $c$-interior PML can be used as a tool for LDP analysis: Conversion to approximate LDP and asymptotic amplification: In the former setting, we are interested in the $(\alpha,\delta)$-LDP guarantees that can be provided by a kernel satisfying $c$-interior PML. We show that in cases where the kernel does not satisfy any finite (pure) LDP guarantee, the $(\alpha,\delta)$-LDP guarantees that can be given collapse to $(0,\delta)$-guarantees. In the latter setting, we show that as long as a kernel satisfies $\eta_\text{TV}(\varepsilon)<1$, we can use it to repeatedly post-process an $\alpha_0$-LDP kernel $\mathsf K_0$ to obtain a system satisfying $\alpha_n$-LDP, where,
\begin{equation}
    \alpha_n \lesssim \alpha\,\eta_\text{TV}(\varepsilon)^n \stackrel{n\to \infty}\longrightarrow 0.
\end{equation}
Hence, we provide an example setup where $c$-interior PML and pure LDP can be used synergistically. This should be seen as a motivating example: in certain settings, we are able to leverage the increased flexibility and broader scope of $c$-interior PML to improve the analysis of LDP systems.  

%% file: sections/background.tex
In this section, we introduce the necessary notation, assumptions and background on probability theory, contraction of measure and pointwise maximal leakage. To stay concise, the definitions and background for estimation and testing are introduced as needed in Section \ref{sec:minimax}. 
\subsection{Notation and Assumptions}
\label{ssec:notation}
Throughout this work, $X$ and $Y$ denote random variables defined on the standard Borel measurable spaces $(\mathcal X,\Sigma_\mathcal X)$ and $(\mathcal Y,\Sigma_\mathcal Y)$, respectively, with distributions $P_X$ and $P_Y$, and transition kernel $\mathsf K$. We assume that there exists some $\sigma$-finite reference measures $\mu_X$, $\mu_Y$, such that $P_X \ll \mu_X$ and $\mathsf K(\cdot |x)\ll \mu_Y$ for all $x\in\mathcal X$. We use $\mathcal  P(\mathcal X)$ to denote the set of all distributions $P_X\ll \mu_X$ on $\mathcal X$. to denote the set of all such transition kernels. The transition kernel $\mathsf K: \Sigma_\mathcal Y \times \mathcal X \to [0,1]$ assigns a probability measure on $\mathcal Y$ to any input symbol from $\mathcal X$. We use $p = \frac{dP_X}{d\mu_X}$ and $k = \frac{d \mathsf K(\cdot | x)}{d\mu_Y}$ to denote the corresponding (conditional) densities with respect to $\mu_X$ and $\mu_Y$. For any measure $\nu$, $\nu \mathsf K$ denotes the push-forward of $\nu$ under $\mathsf K$, that is, $(\nu\mathsf K)(B) = \int_\mathcal X \mathsf K(B|x)\nu(dx)$ for all $B \in \Sigma_\mathcal Y$.

We put the following additional assumptions on the kernels consider in this paper: For any $\mathsf K\in\mathcal P(\mathcal Y|\mathcal X)$ and any $B\in\Sigma_\mathcal Y$, let $\esssup_{\mu_X} \mathsf K(B|\cdot) \coloneqq \inf \left\{a \in \bR : \mu_X \big(\{x : \mathsf K(B|x) > a \} \big) = 0 \right\}$ denote the essential supremum with respect to $\mu_X$. For simplicity, we assume that $\esssup_{\mu_X} \mathsf K(B|\cdot) = \sup_{x \in \cX} \mathsf K(B|x)$ for all measurable $B$ and all kernels considered in this paper. This assumption holds whenever, 
\begin{equation*}
    \mu_X \Big(\big\{x : \mathsf K(B|x) > \sup_{x'} \mathsf K(B|x') - \zeta \big\} \Big) > 0 
\end{equation*}
for all $\zeta >0$. That is, every set of points where $\mathsf K(B|\cdot)$ gets arbitrarily close to its supremum must have positive measure under $\mu_X$. We impose this assumption to avoid lengthy discussions about null sets and supports of measure that are tangential to the main focus of the paper. Similarly, we also assume that the infimum and essential infimum of kernels are equal. By the definition of $\mathcal P_c$, we have $c\mu_X \leq P_X \ll \mu_X$ for all $P_X\in\mathcal P_c$, so we can similarly replace the essential supremum/infimum w.r.t $P_X$ with supremum and infimum over $\mathcal X$ for all $P_X$ in $\mathcal P_c$ as well. 

\subsection{Contraction Coefficients and Strong Data Processing Inequalities}
An important concept in statistics and information theory are measures of \emph{divergence} between probability distributions. Perhaps the most commonly used class of divergences is that of $f$-divergences, which is defined as follows.
\begin{definition}
    Let $f:(0,\infty)\to \mathbb R$ be a convex function with $f(1)=0$. Let $P\ll Q$ be two probability measures defined on the measurable space $(\Omega,\Sigma)$. Then the \emph{$f$-divergence between $P_X$ and $Q_X$} is defined as,
    \begin{equation}
        D_f(P||Q) \coloneqq \mathbb E_Q\left[f\left(\frac{dP}{dQ}\right)\right] = \int_{\Omega} f\left(\frac{dP}{dQ}\right)dQ.
    \end{equation}
\end{definition}
It is a well established fact \cite{Polyanskiy_Wu_2025} that any $f$-divergence satisfies the \emph{data processing inequality}, which states that any kernel $\mathsf K\in \mathcal P(\mathcal Y\mid \mathcal X)$ cannot increase the divergence between distribution. In other words, it holds that for any kernel $\mathsf K$ and any two distributions $P_X,Q_X\in\mathcal P(\mathcal X)$, 
\begin{equation}
    D_f(P_X\mathsf K||Q_X\mathsf K) \leq D_f(P_X||Q_X).
\end{equation}
Initiated by \citet{ahlswede1976spreading}, it has long been of interest to find the precise constants with which the DPI holds under various conditions on the kernel, i.e., to find the smallest number $0\leq\eta_f(\mathsf K)\leq1$ such that $D_f(P_X\mathsf K||Q_X\mathsf K) \leq \eta_f(\mathsf K) D_f(P_X||Q_X)$. For any $f$-divergence, this constant is referred to as the \emph{contraction coefficient} of the kernel under the $f$-divergence. Various lines of work develop frameworks to bound this coefficient under different assumptions about the kernels, etc., see e.g. \cite{raginsky2016strong,polyanskiy2015dissipation,polyanskiy2017strong,asoodeh2020contraction,du2017strong}. In this work, we will introduce a (slightly) more general definition of this constant that includes flexibility with respect to the set of input distributions.
\begin{definition}
    Let $\mathcal Q\subseteq \mathcal P(\mathcal X)$ be a set of probability distributions defined on a measurable space $(\mathcal X,\Sigma_\mathcal X)$. Let $\mathsf K\in\mathcal P(\mathcal Y|\mathcal X)$ be any transition kernel. The contraction coefficient of $\mathsf K$ under a $f$-divergence $D_f$ restricted to $\mathcal Q$ is defined as,
    \begin{equation}
    \label{eq:contractiondefinition}
        \eta_f^{\mathcal Q}(\mathsf K) \coloneqq \sup_{P_X,Q_X\in\mathcal Q: P_X\neq Q_X} \frac{D_f(P_X\mathsf K||Q_X\mathsf K)}{D_f(P_X||Q_X)}.
    \end{equation}
    Whenever $\mathcal Q=\mathcal P(\mathcal X)$, we recover the standard definition of the (distribution-independent) contraction coefficient, for which we omit the superscript, i.e., $\eta_f^{\mathcal P(\mathcal X)}(\mathsf K) \equiv \eta_f(\mathsf K)$.
\end{definition}
Note that in general, $\eta_f^\mathcal Q(\mathsf K)\leq \eta_f(\mathsf K)$. A few common instances of $f$-divergences and the notation for their contraction coefficients are listed below.
\begin{itemize}
    \item \emph{Total variation distance} $\text{TV}(P||Q)$: $f(t) = 0.5|t-1|$; we write $\eta_f(\mathsf K) = \eta_\text{TV}(\mathsf K)$.

    \item \emph{Hockeystick-divergence} $E_\gamma(P||Q)$: $f_\gamma(t) = (t-\gamma)_+-(1-\gamma)_+$; we write $\eta_{f_\gamma}(\mathsf K) = \eta_\gamma(\mathsf K)$.

    \item \emph{Relative entropy} $D(P||Q)$: $f(t) = t\log t$; we write $\eta_f(\mathsf K) = \eta_\text{KL}(\mathsf K)$.

    \item \emph{$\chi^2$-divergence} $\chi^2(P||Q)$: $f(t) = t^2-1$; we write $\eta_f(\mathsf K) = \eta_{\chi^2}(\mathsf K)$.

    \item \emph{Squared Hellinger divergence} $H^2(P||Q)$: $f(t) = (\sqrt{t}-1)^2$; we write $\eta_f(\mathsf K) = \eta_{H^2}(\mathsf K)$.
\end{itemize}
The contraction coefficient of the total variation distance is also known as the \emph{Dobrushin coefficient}, and can be computed by a simple two-point characterization \cite{dobrushin1956central}, that is,
\begin{equation}
\label{eq:dobrushintwopoint}
    \eta_\text{TV}(\mathsf K) = \sup_{x\neq x'} \text{TV}\left(\mathsf K(\cdot|x)||\mathsf K(\cdot|x')\right).
\end{equation}
The Dobrushin coefficient is a crucial quantity in contraction analysis, as it poses an upper bound to the contraction coefficient of almost all $f$-divergence \cite{cohen1998comparisons,del2003contraction,raginsky2016strong}. More specifically, for any $f$, we have,
\begin{equation}
    \eta_f(\mathsf K) \leq \eta_\text{TV}(\mathsf K), \quad \forall \mathsf K\in\mathcal P(\mathcal Y|\mathcal X).
\end{equation}
If $f$ is operator-convex and continuously twice-differentiability, we further have \cite{choi1994equivalence,cohen1998comparisons},
\begin{equation}
    \eta_f(\mathsf K) = \eta_{\chi^2}(\mathsf K), \quad \forall \mathsf K\in\mathcal P(\mathcal Y|\mathcal X).
\end{equation}
In general, the contraction coefficient of the $\chi^2$-divergence is, however, only a lower bound on all other contraction coefficients, that is, $\eta_{\chi^2}(\mathsf K)\leq \eta_f(\mathsf K)$. 

Contraction coefficients have found many applications in information theory and statistics. For a detailed review of their properties and an overview of applications, we refer the interested reader to \citet{Polyanskiy_Wu_2025}. Central for this work is the application of contraction coefficients (or more generally strong data processing inequalities) to minimax risk with local privacy, as pioneered by \citet{duchi2013local} and further developed by \citet{asoodeh2024contraction,asoodeh2020contraction}, the basics of which we will detail in Section \ref{sec:background:subsec:privacyandcontraction} after introducing the fundamentals of local privacy in what follows.

\subsection{Pointwise Maximal Leakage and Local Differential Privacy}
In this section, we introduce the two privacy measures most relevant to this work, pointwise maximal leakage and local differential privacy. As we will see, there is an interesting connection between the two measures, and pointwise maximal leakage can reasonably be seen as a generalization of local differential privacy. This connection is the main motivation of the results in this paper. 

\subsubsection{Local differential privacy}
Perhaps most standard, local privacy is quantified by \emph{local differential privacy} (LDP). For any kernel $\mathsf K$, LDP quantifies the worst-case likelihood ratio of output distributions caused by two differing inputs, modulo a small slack parameter that can be chosen to slightly loosen the guarantee.
\begin{definition}[\emph{Local differential privacy} (LDP) {{\cite{duchi2013LDPminmaxDEF}}}]
    For $\alpha\geq 0$, $\delta\in[0,1]$, a kernel $\mathsf K\in \mathcal P(\mathcal Y|\mathcal X)$ \emph{satisfies $(\alpha,\delta)$-LDP}, if,
    \begin{equation}
        \mathsf K(B|x) \leq e^\alpha \,\mathsf K(B|x') + \delta \quad \forall x\neq x', B\in\Sigma_\mathcal Y.
    \end{equation}
    If a kernel satisfies $(\alpha,0)$-LDP, we say it satifies \emph{pure} $\alpha$-LDP, or simply $\alpha$-LDP.
\end{definition}
The definition of LDP admits a strong operational interpretation in terms of hypothesis testing adversaries \cite{kairouz2016extremal}. Further, \citet{9517999} show that $(\varepsilon,\delta)$-LDP is in fact \emph{equivalent} to a bound on the contraction coefficient of the Hockeystick-divergence $E_\gamma$. More precisely, the authors show that a kernel $\mathsf K$ satisfies $(\alpha,\delta)$-LDP if and only if $\eta_{e^\alpha}(\mathsf K) \leq \delta$. For pure $\alpha$-LDP, this implies that applying the kernel to any two distributions results in a zero-valued Hockeystick-divergence $E_\gamma(P_X\mathsf K||Q_X\mathsf K)$ whenever $\gamma \geq e^\alpha$.

\subsubsection{Pointwise maximal leakage}
The framework of \emph{pointwise maximal leakage} (PML) has been recently proposed by \citet{saeidian2023pointwise} as a generalization of maximal leakage \cite{alvim2012measuring,IssaMaxL}. It can be seen as a measure in the framework of \emph{quantitative information flow} (QIF) \cite{alvim2020science}, where privacy leakage measures are derived from adversarial threat-model formulations. Specifically, PML quantifies privacy by comparing the maximum gain an adversary might obtain after an observation of the output of a privacy mechanism (kernel) to the average gain of a blind guess without access to the output of the mechanism. This gain-function view for a specific gain function $g$ is then maximized over all possible non-negative gain functions to obtain a robust measure of privacy leakage.
\begin{definition}[Gain function view of PML {{\cite{saeidian2023pointwisegeneral}}}]
\label{def:PMLgainfunc}
    Let $X$ be a random variable defined on the set $\mathcal X$ distributed according to $P_X$. Let $(X,Y)$ be induced by $P_X$ together with the mechanism $P_{Y|X}$. Then the \emph{pointwise maximal leakage in the gain function view from $X$ to an outcome $y$} is defined as 
    \begin{equation}
    \label{eq:PMLdefGainFunc}
        \ell_{P_{XY}}(X\to y) \coloneqq \log \, \sup_{g\in\mathfrak G} \sup_{(\mathcal W,\Sigma_\mathcal W)}\frac{\sup_{P_{W|Y}}\mathbb E[g(X,W)\mid Y=y]}{\max_{w\in\mathcal W}\mathbb E[g(X,w)]},
    \end{equation}
    where $\mathcal W$ is some arbitrary \say{guessing space}, the supremum in the numerator is over arbitrary kernels from $(\mathcal Y,\Sigma_\mathcal Y)$ into $(\mathcal W,\Sigma_\mathcal W)$ and $\mathfrak G$ is the set of all non-negative and bounded gain functions, i.e.,
    \begin{equation}
        \mathfrak G \coloneqq \left\{g\in (\Sigma_\mathcal X \otimes \Sigma_\mathcal W)_+\mid \sup_{w\in \mathcal W}\mathbb E[g(X,w)]<\infty\right\}.
    \end{equation}
\end{definition}
This definition is akin to the definition of min-capacity in the QIF framework \cite{alvim2012measuring}. By allowing for arbitrary nonnegative gain functions, this formulation includes many real-world adversarial attacks, e.g., membership inference attacks and attribute inference attacks (see \cite{saeidian2023pointwise} for examples of corresponding gain functions). If the private random variable $X$ is at most countably infinite, an equivalent formulation of this quantity can be made akin to the randomized function model of the maximal leakage framework \cite{IssaMaxL}.
\begin{definition}[Randomized function view of PML {{\cite{saeidian2023pointwise}}}]
\label{def:PMLrandfunc}
    Let $P_{XY}$ be the joint distribution of two random variables defined on the discrete set $\mathcal{X} \times \mathcal{Y}$. Suppose the Markov chain $U - X - Y - \hat{U}$ holds. Then the \emph{pointwise maximal leakage from $X$ to an outcome $y \in \mathcal{Y}$} is defined as
    \begin{equation}
    \label{eq:PMLdefRandFunc}
            \ell(X \to y) \coloneqq 
    \log \sup_{P_{U \mid X}}\frac{\sup\limits_{P_{\hat U \mid Y=y}} \mathbb{P} \left[U=\hat U \mid Y=y \right]}{\max\limits_{u\in \mathcal{U}} P_U(u)}.
    \end{equation}
\end{definition}
This formulation can be interpreted as follows: Assume an adversary is interested in the realization of $U$, a (possibly randomized) function of the private data $X$ (e.g. a persons gender or the town they live in). After observing the output $Y=y$ of a privacy mechanism, the adversary forms a guess $\hat U$ of $U$. The randomized function view of PML then measures leakage by comparing the porbability of a correct guess \emph{with} the observation $Y=y$ to the probability of correctly guessing $U=u$ blindly, that is, \emph{without} access to the output of the privacy mechanism. Maximizing this quantity over \emph{all} possible randomized functions $U$ of the secret yields pointwise maximal leakage.

As shown in \cite{saeidian2023pointwise}, both the gain function view and the randomized function view of PML are equivalent, that is, they lead to the same leakage quantification for any system. As a result, they simplify to the same expression.
\begin{theorem}[{{\cite[Theorem 1]{saeidian2023pointwise}}}]
     Let $X,Y$ be random variables distributed according to the joint distribution $P_{XY}$. Let $P_X \in \mathcal P(\mathcal X)$ be the marginal distribution of $X$. The pointwise maximal leakage according to both Definition~\ref{def:PMLgainfunc} and~\ref{def:PMLrandfunc} simplify to
    \begin{equation}
        \ell_{P_{XY}}(X\to y) = D_\infty(P_{X|Y=y}||P_X),
    \end{equation}
    where $P_{X|Y=y}$ denotes the conditional distribution of $X$ given an outcome $Y=y$, and $D_\infty$ denotes the Rényi-divergence of order infinity \cite{renyi1961entropy}.
\end{theorem}

The core difference between PML and min-capacity/maximal leakage---and the reason for why we call this measure \say{pointwise}---is that PML quantifies leakage to \emph{each outcome $y$ separately}. That is, since the input distribution $P_X$ together with the kernel $\mathsf K$ induce an output random variable $Y$, the pointwise maximal leakage of the system is also a random variable, denoted by $\ell_{P_X\mathsf K}(X\to Y)$. This enables a much more refined privacy analysis compared to the on-average measure of min-capacity; similar to (local) differential privacy, this random variable view of privacy leakage can for example be used to evaluate the probability of tail events of large leakage. Further, it enables \emph{stricter} privacy guarantees by uniform bounds. Central to this work is an observation made by \citet{IssaMaxL}: There, the authors show that in fact, guaranteeing local differential privacy is equivalent to guaranteeing $\varepsilon$-PML for \emph{all} distributions in the probability simplex on $\mathcal X$. More precisely, we have the following theorem.
\begin{theorem}[{{\cite[Theorem 14]{IssaMaxL}}}]
\label{thm:LDPisPMLwithc=0}
    Let $\mathsf K$ be kernel mapping $\mathcal X\to \mathcal Y$. Then $\mathsf K$ satisfies $\varepsilon$-LDP for some $\varepsilon\geq 0$ if and only if it satisfies,
    \begin{equation}
        \sup_{P_X\in\mathcal P(\mathcal X)} \sup_{y\in\mathcal Y} \ell_{P_X\times \mathsf K}(X\to y) \leq \varepsilon.
    \end{equation}
\end{theorem}
This observation motivates us to investigate more context-aware privacy guarantees that restrict the set of distributions over which $\ell(X\to y)$ is supremized to be a strict subset of the probability simplex. As we will see later, such restrictions allow us to give non-trivial privacy guarantees that significantly increase utility by moving away from the overly conservative assumption of arbitraily skewed distributions. 

\subsection{Related Works} 
\label{sec:background:subsec:privacyandcontraction}
We summarize related results on the connection between privacy and contraction, as well as sample complexities of private mean estimation and hypothesis testing. For an overview of contraction results and strong data processing inequalities more generally, we refer the reader to, e.g, \cite[Chapter 33]{Polyanskiy_Wu_2025}, \cite{polyanskiy2017strong}, \cite{polyanskiy2015dissipation}, \cite{raginsky2016strong}, \cite{ahlswede1976spreading}, \cite{cohen1998comparisons}. 

\emph{Strong data processing inequalities} (SDPIs), that is, inequalities of the form,
\begin{equation}
    D_{f}(P_X\mathsf K||Q_X\mathsf K) \leq T\left(D_g(P_X||Q_X)\right),
\end{equation}
where $f$ and $g$ are two (possibly different) functions corresponding to $f$-divergences and $T$ is some (possibly nonlinear) transformation, have been proven useful for determining asymptotic risk of private estimation procedures. Initially, \citet{duchi2013local} showed that, if $\mathsf K$ satisfies $\varepsilon$-LDP, then,
\begin{equation}
    D(P_X\mathsf K||Q_X\mathsf K) \leq \min\{4,e^{2\varepsilon}\}(e^\varepsilon-1)^2 \text{TV}(P_X||Q_X)^2.
\end{equation}
The authors then use this SDPI to show how an additional LDP constraint can significantly increase the effective sample complexity of various statistical problems. Related, it was shown by \citet{kairouz2016extremal} that under the same $\varepsilon$-LDP constraint,
\begin{equation}
    \eta_\text{TV}(\mathsf K) \leq \frac{e^\varepsilon-1}{e^\varepsilon+1}. 
\end{equation}
Further, the Dobrushin coefficient has been shown to exactly characterize a version of maximal leakage where adversaries are restricted to maximize over binary randomized functions of the secret \cite{cung2024binary}. Improving on the work of Duchi, \citet{asoodeh2024contraction} show that,
\begin{equation}
    \eta_\text{KL}(\mathsf K) = \eta_{\chi^2}(\mathsf K) = \eta_{H^2}(\mathsf K) \leq \left[\frac{e^\varepsilon-1}{e^\varepsilon+1}\right]^2. 
\end{equation}
The authors then use this bound on the contraction coefficients to improve upon many of the statistical problems already examined by Duchi in \cite{duchi2013local}, as well as others. 

For mean estimation, it is shown in \cite{duchi2013local} that whenever $\alpha\in [0,1]$, then the $\alpha$-locally differentially private risk of mean estimation in $d$ dimensions scales according to $ d/(\alpha^2n)$. More recent works show that for general $\alpha\geq 0$, the risk of estimating $d$-dimensional means under $\alpha$-LDP scales according to $d/\min\{\alpha, \alpha^2,d\}$ \cite{chen2020breaking,isik2023exact,asi2022optimal}, resolving the discrepancy between the non-private risk (which scales according to $1/n$) and the risk under $\alpha$-LDP for $\alpha \to \infty$. 

For simple binary hypothesis testing, recent works \cite{pensia2023simple,asoodeh2024contraction} show that testing between binary distributions $P_X$ and $Q_X$ at a fixed error requires at least $n^*_\alpha$ samples, where,
    \begin{equation}
    \label{eq:LDPprivnhypo}
        n^*_{\alpha}(P_X,Q_X) \asymp \begin{cases}
            (\alpha^2\cdot\text{TV}^2(P_X||Q_X))^{-1}, &\text{if }\alpha\in(0,1),\\
            (e^\alpha \cdot \text{TV}^2(P_X||Q_X))^{-1}, &\text{if }\alpha\in[1,\log \frac{H^2(P_X||Q_X)}{\text{TV}^2(P_X||Q_X)}], \\
            H^{2}(P_X||Q_X)^{-1}, &\text{otherwise. }
        \end{cases}
    \end{equation}
Additional results for estimation and testing are derived in \cite{acharya2023unified,acharya2020inference,rohde2020geometrizing,barnes2020fisher,duchi2024right}. Contraction properties of kernels satisfying $(\alpha,\delta)$-LDP are investigated in \cite{nuradha2025non}. Other applications of the contraction bounds in privacy include privacy amplification \cite{feldman2018privacy,asoodeh2020privacy,grosse2025bounds}, analysis of the privacy of iterative algorithms \cite{asoodeh2024privacy} and Markov-chain mixing times \cite{zamanlooy2024mathrm}.

%% file: sections/dobrushin.tex
In this section, we will provide the first main contribution of this paper: A bound on the Dobrushin coefficient of kernels that satisfy $\varepsilon$-PML for a set of distributions with minimum mass or density bounded away from zero by some constant $c$. This is motivated by the observation in Theorem \ref{thm:LDPisPMLwithc=0}: Local differential privacy offers robust privacy by protecting agains any adversary in the threat-models of PML for \emph{any} distribution in the simples $\mathcal P(\mathcal X)$. To retain the same robustness with respect to the considered adversaries, but move away from the worst-case distributional assumptions, we propose the following relaxation of LDP, which we name \emph{$c$-interor $\varepsilon$-PML}.
\begin{definition}[$c$-interior $\varepsilon$-PML] 
\label{def:eps_c_PML}
    Let $(\mathcal X,\Sigma_\mathcal X)$ be a measurable space, and let $\mathcal P(\mathcal X)$ be the space of probability distributions on $(\mathcal X,\Sigma_\mathcal X)$ that are absolutely continuous with respect to a reference measure $\mu_X$. For $c\in[0,\nicefrac{1}{\mu_X(\mathcal X)}]$, define the \emph{$c$-interior} $\mathcal P_c(\mathcal X)$ of the simplex $\mathcal P(\mathcal X)$ as,\footnotemark 
    \begin{equation}
        \mathcal P_c(\mathcal X) \coloneqq \{P_X\in\mathcal P(\mathcal X): \frac{dP_X}{d\mu_X} =f_X \geq c\}.
    \end{equation}
    \footnotetext{Whenever the space $\mathcal X$ is clear from context, we write $\mathcal P_c(\mathcal X) \coloneqq \mathcal P_c$. Note that throughout this paper, we assume $c$ to be a fixed problem constant. We therefore omit the dependence on $c$ in some of the notation in the following to streamline notation.}
    A transition kernel $\mathsf K$ satisfies \emph{$c$-interior $\varepsilon$-pointwise maximal leakage} ($\varepsilon$-PML on $\mathcal P_c$), if,
    \begin{equation}
    \label{eq:c_setPML}
       \log  \frac{\mathsf K(B | x)}{(P_X \mathsf K)(B)} \leq \varepsilon \quad \forall x\in\mathcal X,\,B\in \Sigma_\mathcal Y\, \text{ and }\, \forall P_X \in\mathcal P_c. 
    \end{equation}
    Further, we denote the set of all kernels satisfying $\varepsilon$-PML on $\mathcal P_c$ by $\mathcal M(\varepsilon,c)$.
\end{definition}
The definition of the $c$-interior allows us to restrict our analyses to distributions that are relatively regular. Specifically, any choice of $c>0$ ensures that the set $\mathcal P_c(\mathcal X)$ only contains distributions on $\mathcal X$ with full support, and the exact value of $c$ determines the minimum probability of rare events, e.g. outliers appearing in a data sequence. The parameter $c$ can therefore be seen as determining the \say{degree of uniformity} of the set of distributions. The definition of the $c$-interior makes sense  for any $c\in[0,\nicefrac{1}{\mu_X(\mathcal X)}]$. However, it is worth pointing out that by it's definition, $\mathcal P_0(\mathcal X)= \mathcal P(\mathcal X)$, and when $c=\nicefrac{1}{\mu_X(\mathcal X)}$, then $\mathcal P_c$ only contains the uniform distribution on $\mathcal X$. In what follows, we will therefore state all results for values $0<c<\nicefrac{1}{\mu_X(\mathcal X)}$. Another important observations is that the definition of $c$-interior PML limits our analyses to \emph{finite} measure spaces.
\begin{remark}
\label{rem:spaceneedstobebounded}
    In order to make useful statements for any value of $c>0$ in Definition \ref{def:eps_c_PML}, the measurable space $(\mathcal X,\Sigma_\mathcal X)$ needs to be bounded in the sense that,
    \begin{equation}
        \mu_X(\mathcal X) = \int_\mathcal X \mu_X(dx) < \infty. 
    \end{equation}
    If this is not the case, then the set $\mathcal P_c(\mathcal X)$ is empty for any $c>0$, since no distribution on an unbounded space can have uniformly lower-bounded densities. Heuristically, the definition of the $c$-interior is therefore only useful on spaces on which we can define a uniform distribution, e.g, norm-balls with finite radii or discrete spaces. 
\end{remark}

The definition of $c$-interior $\varepsilon$-PML is a relaxation (or generalization) of $\alpha$-LDP in the sense that the specific choice of $c=0$ recovers the definition of $\alpha$-LDP.
\begin{remark}
    If a kernel $\mathsf K$ satisfies $\varepsilon$-PML on $\mathcal P_c$ for all $c>0$, then it satisfies $\varepsilon$-LDP due to Theorem \ref{thm:LDPisPMLwithc=0}. That is,
    \begin{equation}
        \mathsf K \in \mathcal M(\varepsilon,0) \iff \mathsf K\text{ satisfies } \varepsilon\text{-LDP.}
    \end{equation}
\end{remark}

To facilitate the analyses below, we provide the following lemma, which expresses inequality \eqref{eq:c_setPML} for the kernels satisfying $\varepsilon$-PML on $\mathcal P_c$ in a more useful form. The lemma is proved in Appendix~\ref{app:proof_pml_contraints}.
\begin{lemma}
\label{lemma:pml_constraints}
If the kernel $\mathsf K$ satisfies $\varepsilon$-PML on $\mathcal P_c$, then 
\begin{equation}
    \mathsf K(B| x) \leq e^\varepsilon \Big(c(\mu_X \mathsf K)(B)+ d_c \cdot \mathsf K(B| x') \Big), 
\end{equation}
for all $x,x' \in \cX$ and $B \in \Sigma_{\cY}$, where $d_c \coloneqq 1-c\mu_X(\mathcal X)$. 
\end{lemma}
The constant $d_c = 1-c\mu_X(\mathcal X)$ will be used frequently below, and can be interpreted as a sort of \say{dual} to the constant $c$: Whenever $c$ is small, the constant $d_c$ is large, and vice-versa. In particular, $c=0 \implies d_c=1$, and $d_c = 0 \implies c = \nicefrac{1}{\mu_X(\mathcal X)}$.

Throughout this work, we will highlight two different approaches of solving contraction problems for classes of private kernels. The first approach works by transferring results from the LDP literature to our $c$-interior PML framework: We prove a connection between $c$-interior PML and LDP shortly, which will enable this transfer. However, as we will also show below, $c$-interior PML and LDP are by no means \emph{equivalent} measures, and the two privacy definitions capture fundamentally different aspects of the kernels. This leads us to our second approach, where we highlight the application of $c$-interior PML as an \emph{analytical tool} for deriving contraction results. Due to the additional assumptions in $c$-interior PML about, e.g., the size of the underlying space, using PML constraints directly will enable us to show tighter bounds in some scenarios below. Further, even when both approaches yield equivalent results, the proof techniques with the latter approach via $c$-interior PML are often more direct, and reveal interesting structures of the problem. Finally, and perhaps most importantly, the results shown by using $c$-interior PML extend beyond the LDP regime, and enable us to make statements about kernels that are out of scope for the LDP measure, e.g., discrete kernels with zero-valued entries.  

\subsection{Relationship between $\varepsilon$-PML on $\mathcal P_c$ and LDP}
We begin this study by investigating the relationship between $c$-interior PML and LDP. Interestingly, this study reveals that the correspondence between PML and LDP happens in two distinct regimes: If $\varepsilon$ is small enough, any kernel satisfying $\varepsilon$-PML on $\mathcal P_c$ also satisfies some $\alpha$-LDP guarantee. This, however, cannot be the case for larger values of $\varepsilon$. In particular, there are many kernels that satisfy finite $\varepsilon$-PML guarantees on $\mathcal P_c$ for some $c>0$, but that do not satisfy \emph{any} finite $\alpha$-LDP guarantee (consider, e.g., our running example of discrete kernels represented by stochastic matrices with zero entries). For these kernels, we instead work on a \emph{transformed} kernel $\mathsf M$ defined from the original kernel $\mathsf K$: We construct $\mathsf M$ by adding a density floor to $\mathsf K$, which ensures that $\mathsf M$ satisfies a finite $\alpha$-LDP guarantee. By choosing the level of this floor conveniently, we can ensure that $P_X\mathsf K = \bar P_X\mathsf M$, where $\bar P_X$ is a one-to-one transformation of $P_X$ such that $P_X \in \mathcal P_c\implies \bar P_X\in\mathcal P(\mathcal X)$. This will later enables us to write, e.g.,  $D_f(P_X\mathsf K||Q_X\mathsf K)= D_f(\bar P_X\mathsf M||\bar Q_X\mathsf M)$. Since we show that $\mathsf M$ satisfies $\alpha$-LDP, we will then be able to bound the RHS of this equation by LDP techniques. The proof of the Lemma is provided in Appendix~\ref{app:proofLDPlemma}.
\begin{lemma}[$c$-interior PML to LDP]
\label{lem:PMLimpliesLDP}
    Given a kernel $\mathsf K$, define the new kernel $\mathsf M$ by,
    \begin{equation}
         \mathsf M(\cdot |x) = c(\mu_X \mathsf K )(\cdot) + d_c\mathsf K(\cdot|x) \quad \forall x\in\mathcal X.
    \end{equation}
    Note that with this definition, $P_X \mathsf K = \bar P_X \mathsf M$, where $\bar P_X \coloneqq (P_X - c\mu_X)/d_c$ for all $P_X\in\mathcal P_c$. If $\mathsf K$ satisfies $\varepsilon$-PML on $\mathcal P_c$, then $\mathsf M$ satisfies $\log(e^\varepsilon d_c+1)$-LDP. Further, if $\varepsilon<-\log(c\mu_X(\mathcal X))$, then $\mathsf K$ itself satisfies $\alpha$-LDP, where,
    \begin{equation}
        \alpha = \log \,\frac{e^\varepsilon d_c}{1-e^\varepsilon c \mu_X(\mathcal X)}.
    \end{equation}
\end{lemma}
Lemma \ref{lem:PMLimpliesLDP} shows that whenever $\varepsilon<-\log(c\mu_X(\mathcal X))$, then a kernel $\mathsf K$ that satisfies $\varepsilon$-PML on $\mathcal P_c$ also satisfies some finite $\alpha$-LDP guarantee. Whenever $\mathsf K$ satisfies such an $\varepsilon$-PML guarantee on $\mathcal P_c$, we say that the kernel is in the \emph{high-privacy regime}. The fact that such a threshold exists is no surprise: A similar behavior has been observed for $\varepsilon$-PML with respect to a fixed data-generating distribution in \cite{10646583}.

\subsection{Dobrushin Coefficients under PML constraints}
We now present the main result of this work: A bound on the Dobrushin coefficient of kernels satisfying $\varepsilon$-PML on $\mathcal P_c$ for some $c>0$. The bound follows from the two-point characterization of the Dobrushin coefficient in \eqref{eq:dobrushintwopoint}, together with the formulation of $c$-interior PML constraints in Lemma \ref{lemma:pml_constraints}. We also show that the bound is tight by constructing a kernel that achieves it for any choice of privacy parameter $\varepsilon\geq 0$, and density floor bound $c>0$.

\begin{theorem}
\label{thm:TVcontraction}
    Let $(\mathcal X,\Sigma_\mathcal X)$ be a measurable space and let $\mu_X$ be a finite dominating measure on $\mathcal X$. If $\mathsf K$ satisfies $\varepsilon$-PML on $\mathcal P_c$ for $\varepsilon>0$ and $0<c < \mu_X(\mathcal X)^{-1}$, then it holds that,
    \begin{equation}
        \eta_\text{TV}(\mathsf K) \leq \min\Bigg\{\frac{e^\varepsilon-1}{1+e^\varepsilon(1-c\cdot\mathbb \mu_X(\mathcal X))},\,1\Bigg\}\eqqcolon \eta_\text{TV}(\varepsilon).
    \end{equation}
    Further, there exists a kernel $\mathsf K_{\varepsilon,c}^\star \in \mathcal M(\varepsilon,c)$ that achieves the bound.
\end{theorem}
\begin{proof}
    We have according to \citet{dobrushin1956central},
\begin{equation}
\label{eq:dobrushin}
    \eta_\text{TV}(\mathsf K) = \sup_{x\neq x'}\text{TV}\Big(\mathsf K(\cdot| x)||\mathsf K(\cdot | x')\Big) =  \sup_{x\neq x'} \sup_{B\in\Sigma_{\mathcal Y}}\Big(\mathsf K(B| x)-\mathsf K(B| x')\Big). 
\end{equation}

Fix any $B\in\Sigma_\mathcal Y$. We have $\mathsf K(B| x) = 1-\mathsf K(B^c| x)$. Applying Lemma \ref{lemma:pml_constraints} to both $B$ and $B^c$ yields the following two constraints on $\mathsf K$:
\begin{equation}
\label{eq:Zpmlconstraintz0}
    \mathsf K(B| x) \leq e^\varepsilon c (\mu_X\mathsf K)(B) + e^\varepsilon d_c \mathsf K(B| x'),
\end{equation}
and,
\begin{align}
    \label{eq:Zpmlconstraintz1}
    1-\mathsf K(B| x') &\leq e^\varepsilon c \int_\mathcal X (1-\mathsf K(B| x))\mu_X(dx) + e^\varepsilon d_c (1-\mathsf K(B|x)) \\[.5em] &= e^\varepsilon c (\mu_X(\mathcal X)- (\mu_X\mathsf K)(B))+ e^\varepsilon d_c (1-\mathsf K(B| x)).
\end{align}
Summing up \eqref{eq:Zpmlconstraintz0} and \eqref{eq:Zpmlconstraintz1} yields,
\begin{equation}
    \mathsf K(B|x)-\mathsf K(B|x') \leq \frac{e^\varepsilon-1}{1+e^\varepsilon d_c}.
\end{equation}
Since this holds for \emph{any} measurable set $B$, we obtain the bound,
\begin{equation}
    \sup_{x,x'} \sup_{B\in\Sigma_\mathcal Y} \Big\{\mathsf K(B| x) - \mathsf K(B|x')\Big\} \leq \frac{e^\varepsilon-1}{1+e^\varepsilon d_c}.  
\end{equation}
By the definition of $\eta_\text{TV}(\mathsf K)$, we directly obtain,
\begin{equation}
    \eta_\text{TV}(\mathsf K) \leq \frac{e^\varepsilon-1}{1+e^\varepsilon d_c}.
\end{equation}
Finally, if $e^\varepsilon -1 > 1+e^\varepsilon d_c$, then the bound is vacuous, and we can instead take the bound implied by the standard DPI, that is, $\eta_\text{TV}(\mathsf K)\leq 1$. This proves the upper bound. 

The following kernel constructions achieve the bound: Let $\varepsilon^* \coloneqq -\log(\nicefrac{c\mu_X(\mathcal X)}{2}) $If $\varepsilon<\varepsilon^*$, pick any set $\mathcal A \in \Sigma_\mathcal X$ such that $\max\{\mu_X(\mathcal A), \mu_X(\mathcal A^c)\}\leq (e^\varepsilon c)^{-1}$.\footnote{A straightforward choice of $\mathcal A$ here and below is a measurable set such that $\mu_X(\mathcal A) = \mu_X(\mathcal A^c)=\frac{1}{2}\mu_X(\mathcal X)$, if such a set exists in $\Sigma_\mathcal X$. If it does not, this can be resolved by a single bit of common randomness, see Appendix \ref{app:mechanismrandomization}.} For a binary output space $\mathcal Y =\{0,1\}$, define the kernel $\mathsf K^\star$ by,
    \begin{equation}
        \mathsf K_{\varepsilon,c}^\star(0|x) = \frac{1}{1+e^\varepsilon d_c}\begin{cases}
            e^\varepsilon(1-c\mu_X(\mathcal A)), &\text{if }x\in \mathcal A,\\ 
            1-e^\varepsilon c\mu_X(\mathcal A), &\text{if }x\in\mathcal A^c. 
        \end{cases}\,, \qquad \mathsf K^\star_{\varepsilon,c}(1|x) = 1- \mathsf K_{\varepsilon,c}^\star(0|x).
    \end{equation}
    If $\varepsilon\geq \varepsilon^*$, pick $\mathcal A\in\Sigma_\mathcal X$ such that $\min\{\mu_X(\mathcal A),\mu_X(\mathcal A^c)\}\geq (e^\varepsilon c)^{-1}$ and define,
    \begin{equation}
        \mathsf K_{\varepsilon,c}^\star(0|x) = \mathbf 1\{x\in \mathcal A\}, \quad \mathsf K^\star_{\varepsilon,c}(1|x) = \mathbf 1\{x\in\mathcal A^c\}.
    \end{equation}
    It can be checked that both of these constructions satisfy $\varepsilon$-PML on $\mathcal P_c$ in the respective ranges of $\varepsilon$, see Appendix \ref{app:mechanismrandomization}. Further, from the two-point characterization of $\eta_\text{TV}$ in \eqref{eq:dobrushin}, we to see that,
    \begin{equation}
        \eta_\text{TV}(\mathsf K_{\varepsilon,c}^\star) = \frac{1}{1+e^\varepsilon d_c}(e^\varepsilon(1-c\mu_X(\mathcal A)-1+e^\varepsilon c \mu_X(\mathcal A))= \frac{e^\varepsilon-1}{1+e^\varepsilon d_c},
    \end{equation} 
    if $\varepsilon <\varepsilon^*$ and $\eta_\text{TV}(\mathsf K^\star_{\varepsilon,c}) = 1$ if $\varepsilon \geq \varepsilon^*$.
\end{proof}
If $\mathsf K\in\mathcal M(\varepsilon,c)$ with $\varepsilon\geq \varepsilon^* \coloneqq -\log(c\mu_X(\mathcal X)/2)$, we say that it is in the \emph{low-privacy regime}. (Recall the definition of the \emph{high-privacy regime} $\varepsilon<-\log(c\mu_X(\mathcal X))$ in Lemma \ref{lem:PMLimpliesLDP}.) In the low-privacy regime, we have $\eta_\text{TV}(\varepsilon)=1$. It is important to point out that the bound in Theorem \ref{thm:TVcontraction} is on the \emph{input-distribution independent} contraction coefficient, while the privacy guarantee is made only for distributions in the $c$-interior. That is, we have used $c$-interior PML only as an algebraic constraint. However, the following Lemma shows that for the contraction of total variation distance specifically, this distinction is irrelevant. The lemma is proved in Appendix \ref{app:QdontchangeTVcontrproof}. 

\begin{lemma}
    \label{lem:QcdontchangeTVcontr}
        For any $0< c <\mu_X(\mathcal X)^{-1}$ and any $\varepsilon> 0$, we have
        \begin{equation}
             \sup_{\mathsf K \in\mathcal M(\varepsilon,c)} \eta^{\mathcal P_c}_\text{TV}(\mathsf K) = \eta_\text{TV}(\varepsilon). 
        \end{equation}
\end{lemma}

Since $\varepsilon$-PML on $\mathcal P_0$, that is, $\varepsilon$-PML on $\mathcal P(\mathcal X)$ is equivalent to $\varepsilon$-LDP, Theorem \ref{thm:TVcontraction} yields a bound for kernels satisfying $\varepsilon$-LDP as a special case. In fact, as the following remark shows, it recovers the known result for LDP, and generalizes the results in \cite{kairouz2016extremal} to arbitrary probability spaces. Further, we highlight that the result in Theorem \ref{thm:TVcontraction} is in fact sharper than what can be shown via the correspondence between PML and LDP in Lemma \ref{lem:PMLimpliesLDP}.
\begin{remark}
\label{rem:ourtechniquebetterTV}
    If $c\to0$, that is, if a kernel $\mathsf K$ guarantees $\alpha$-LDP, Theorem \ref{thm:TVcontraction} shows that,
    \begin{equation}
    \label{eq:etaTVLDP}
        \eta_\text{TV}(\mathsf K) \leq \frac{e^\alpha-1}{e^\alpha+1}. 
    \end{equation}
    This is exactly the result shown in \cite[Corollary 11]{kairouz2016extremal} for discrete spaces. Using the correspondence in Lemma \ref{lem:PMLimpliesLDP}, \eqref{eq:etaTVLDP} also yields that if $\mathsf K$ satisfies $\varepsilon$-PML on $\mathcal P_c$, then, by realizing that $\eta_\text{TV}(\mathsf M) = d_c\eta_\text{TV}(\mathsf K)$, we find,
    \begin{equation}
        \eta_\text{TV}(\mathsf K) \leq \begin{cases}
            \frac{e^\varepsilon-1}{e^\varepsilon(1-2c\mu_X(\mathcal X))+1} , &\text{if }0\leq \varepsilon <-\log(c\mu_X(\mathcal X)),\\
            \frac{e^\varepsilon }{e^\varepsilon d_c+2}, &\text{otherwise.}
        \end{cases} 
    \end{equation}
    Clearly, both of these bounds are looser than what we derived in Theorem \ref{thm:TVcontraction} via $c$-interior PML contraints.
\end{remark}
We conclude this section by providing two examples of the general bound by specializing it to two canonical examples of probability spaces: finite (discrete) spaces and $p$-norm balls in $d$-dimensional Euclidean space. Note that, as detailed in Remark \ref{rem:spaceneedstobebounded}, for the latter, the definition of \say{densities uniformly bounded away from zero} requires the balls to have finite radius. Otherwise, the set $\mathcal P_c$ is empty for any $c>0$, as such densities cannot exist. 
\begin{example}[Discrete alphabets]
    Let $\mathcal X = [N]$ for some $N\in\mathbb N$, and let $\mu_X$ denote the counting measure on $\mathcal X$. For any kernel $\mathsf K$ that satisfies $\varepsilon$-PML on $\mathcal P_c$ (i.e., $\varepsilon$-PML w.r.t. all $P_X: \min_{i\in[N]} P_X(i)\geq c$) for some $c< |\mathcal X|^{-1}$,
    \begin{equation}
        \eta_\text{TV}(\mathsf K) \leq \min\left\{1,\,\frac{e^\varepsilon-1}{1+e^\varepsilon(1-c\cdot N)}\right\}.
    \end{equation}
\end{example}
\begin{example}[$d$-dimensional Euclidean space]
    For some $d\in\mathbb N$, $p\in[1,\infty)$, let $\mathcal X = \mathcal B_d^p \subset \mathbb R^d$ be the $d$-dimensional (closed) unit ball in the $p$-norm, that is, let,
    \begin{equation}
        \mathcal B_d^p \coloneqq \left\{x \in \mathbb R^d: ||x||_p = \bigg(\sum_{i=1}^d |x_i|^p\bigg)^{\nicefrac{1}{p}} \leq 1\right\}.
    \end{equation}
    Let $\mu$ denote the Lebesgue measure restricted to $\mathcal B_d^p$. A kernel $\mathsf K$ satisfies $\varepsilon$-PML on $\mathcal P_c$, if,
    \begin{equation}
        \log \left[k(y|x)\Big(\int_{\mathcal B^p_d}k(y|x')q(x')\mu_X\Big)^{-1}\right] \leq \varepsilon \quad \forall x \in\mathcal X, \,\mu_Y\text{-a.e.}, 
    \end{equation}
    for any $q \in \{q'=dP_X/d\mu_X: P_X\in\mathcal P_c,\; q'(x) \geq c \quad \forall x\in\mathcal X\}$. Under these conditions,
    \begin{equation}
        \eta_\text{TV}(\mathsf K) \leq \min\left\{1,\,\frac{e^\varepsilon-1}{1+e^\varepsilon(1-c\cdot \mathbb V(\mathcal B_d^p))}\right\},
    \end{equation}
    where $\mathbb V(\mathcal B_d^p) = \int_{\mathcal B_d^p}1\mu_X$ denotes the volume of the ball $\mathcal B_d^p$.
\end{example}

%% file: sections/E_gamma.tex
The result in Theorem \ref{thm:TVcontraction} extends to the $E_\gamma$-divergence, also known as \emph{Hockeystick-divergence}. First proposed by \citet{cohen1998comparisons}, the $E_\gamma$-divergence has found many applications in contemporary information theory, for example for channel resolvability \cite{liu2016e_} and channel coding converses \cite{polyanskiy2010arimoto,polyanskiy2010channel,sharma2012strong}. It further has a direct operational interpretation in terms of the hypothesis testing trade-off function via the convex conjugate \cite{dong2022gaussian,elkayam2016variational}. Because of the close connection between the $E_\gamma$-divergence and differential privacy in many of its variants \cite{barthe2013beyond,9517999,asoodeh2020contraction}, the measure has also become central in newer developments in the privacy literature, see, e.g., \cite{dong2022gaussian,balle2019privacy,gomez2025gaussian,dvijotham2020framework}.
\begin{definition}
    For two probability measures $P \ll Q$ and $\gamma \geq 0$, the \emph{Hockeystick-divergence $E_\gamma(P||Q)$} is the $f$-divergence corresponding to $f(t) = (t-\gamma)_+ - (1-\gamma)_+$, where $(x)_+ = \max\{0,x\}$. It can be equivalently expressed as, 
    \begin{equation}
    \label{eq:Egammadef}
        E_\gamma (P||Q) = \sup_{B}\Big[ \,P(B)-\gamma Q(B)\Big] - (1-\gamma)_+,
    \end{equation}
    where the supremum is over all measurable sets $B\in\Sigma_\mathcal Y$.
\end{definition}
Clearly, we have $E_1(P||Q) = \text{TV}(P||Q)$. Hence, the $E_\gamma$-divergence generalizes the total variation distance. As observed in \cite{liu2016e_}, perhaps the most simple operational interpretation of the $E_\gamma$-divergence is the following: In a Bayesian binary hypothesis testing setting, let $\pi_P$ and $\pi_Q$ be the prior probabilities of distributions $P$ and $Q$. Correctly distinguishing between $P$ and $Q$ is possible with probability,
\begin{equation}
    \pi_Q + \pi_P E_{\frac{\pi_Q}{\pi_P}}(P||Q).
\end{equation}
We will examine this application more closely in Section \ref{sec:minimax}.

A particularly valuable property of the $E_\gamma$-divergence is that it can be interpreted as an \say{elementary divergence} for $f$-divergences. In particular, for any $P$, $Q$, any $f$-divergence with twice-differentiable $f$ can be written as \cite{cohen1998comparisons,7552457},
\begin{equation}
\label{eq:fdivRepWithE_gamma}
    D_f(P||Q) = \int_0^\infty E_\gamma(P||Q)f''(\gamma)d\gamma = \int_1^\infty (E_\gamma(P||Q)f''(\gamma)+\gamma^{-3}E_\gamma(Q||P)f''(\gamma^{-1}))d\gamma.
\end{equation}
This property makes the $E_\gamma$-divergence particularily useful for deriving bounds on $f$-divergences, see \cite{7552457,hirche2024quantum,grosse2025bounds,polyanskiy2015dissipation}.
In this section, we investigate the contraction behavior of the $E_\gamma$-divergence when a transition kernel satisfies $\varepsilon$-PML on $\mathcal P_c$. 

\subsection{Contraction of $E_\gamma$-divergence with $\varepsilon$-PML as an Algebraic Property}
We begin by treating the requirement that $\mathsf K$ satisfy $\varepsilon$-PML on $\mathcal P_c$ as a purely algebraic property of the kernel. That is, we bound the \emph{unrestricted} contraction coefficient $\eta_\gamma(\mathsf K)$, where the optimization in \eqref{eq:contractiondefinition} is over input distributions from the entire simplex $\mathcal P(\mathcal X)$. Note that unlike $\eta_\text{TV}(\mathsf K)$, in the case of $E_\gamma$-divergence we often have $\eta_\gamma^{\mathcal P_c}(\mathsf K)<\eta_\gamma(\mathsf K)$ for $\gamma > 1$ and $c>0$. However, the treatment of $\varepsilon$-PML as an algebraic property only is of independent interest, as it can be used to quantify contraction for \emph{arbitrary} kernels, as they appear, e.g., when dealing with Markov-chain mixing. We will devise strategies to bound the set-restricted coefficient $\eta_\gamma^{\mathcal P_c}(\mathsf K)$ in the following sections. Since $\varepsilon$-PML on $\mathcal P_0$ is equivalent to $\varepsilon$-LDP, and the Hockeystick-divergence at $\gamma=1$ is equivalent to the total variation distance, the following statement can be seen as a generalization of both \cite[Theorem 1]{10206578} and Theorem \ref{thm:TVcontraction} above. 

\begin{theorem}
\label{thm:E_gamma_contraction}
    Let $(\mathcal X,\Sigma_\mathcal X)$ be a measurable space and let $\mu_X$ be a finite dominating measure on $\mathcal X$. For some $\varepsilon> 0$ and some $0<c < \mu_X(\mathcal X)^{-1}$, define
     \begin{equation}
        \gamma^* = \frac{e^\varepsilon d_c}{1-e^\varepsilon c \mu_X(\mathcal X)},
    \end{equation}
    If $\mathsf K$ is some transition kernel satisfying $\varepsilon$-PML on $\mathcal P_c$ , then for all $\gamma \geq 1$,
    \begin{equation}
        \eta_\gamma(\mathsf K)\leq \Psi_\gamma(e^\varepsilon,c) = \begin{cases}
            \left(\frac{\gamma^*-\gamma}{\gamma^*-1}\right)_+\eta_\text{TV}(\varepsilon), &\text{if }0\leq \varepsilon< -\log(c\mu_X(\mathcal X)), \\
            \eta_\text{TV}(\varepsilon), &\text{otherwise.}
        \end{cases}
    \end{equation}
    where $d_c \coloneqq 1-c\mu_X(\mathcal X)$ and $\eta_\text{TV}(\varepsilon)$ is defined as in Theorem \ref{thm:TVcontraction}.
\end{theorem}
We prove Theorem \ref{thm:E_gamma_contraction} by an extension of the technique used to prove Theorem \ref{thm:TVcontraction}. First, note that $\eta_\gamma(\mathsf K)$ admits a similar two-point characterization as the Dobrushin coefficient in \eqref{eq:dobrushintwopoint}: As shown in \cite[Theorem 2]{asoodeh2020contraction}, we have,
\begin{equation}
    \eta_\gamma(\mathsf K) = \sup_{x\neq x'} E_\gamma\big(\mathsf K(\cdot |x)||\mathsf K(\cdot|x')\big) = \sup_{x\neq x'} \sup_{B\in\Sigma_\mathcal Y} \left\{\mathsf K(B|x) - \gamma \mathsf K(B|x')\right\}. 
\end{equation}
We therefore bound the difference
\begin{equation}
    \mathsf K(B|x)-\gamma \mathsf K(B|x') \quad \forall B\in\Sigma_\mathcal Y, \, \forall x,x'\in\mathcal X.
\end{equation}
While in the proof of Theorem \ref{thm:TVcontraction}, we obtain a bound on this quantity for $\gamma=1$ by simple algebraic manipulation, this simple algebra does not suffice for $\gamma >1$. We therefore find an upper bound by formulating a corresponding linear program, and constructing an explicit set of feasible dual variables. By standard Lagrangian duality \cite[Section~5.1.3]{boyd2004convex}, this leads us to the presented upper bound. The full proof is presented in Appendix \ref{app:Egammaproof}.

\begin{remark}
\label{rem:Egammarewriting}
    Whenever $\varepsilon$ is in the high-privacy regime, the contraction coefficient is decreasing in $\gamma$. For larger values of $\varepsilon$, the presented bound is the trivial bound via $\eta_\gamma(\mathsf K) \leq \eta_\text{TV}(\mathsf K)$. The result in Theorem \ref{thm:E_gamma_contraction} further recovers the result in \cite[Theorem 1]{10206578} for local differential privacy as $c\to 0$, where it is shown that, if $\mathsf K$ satisfies $\alpha$-LDP, then,
    \begin{equation}
    \label{eq:LDPHScontraction}
        \eta_\gamma(\mathsf K) \leq  \bigg(\frac{e^{\alpha}-\gamma}{e^{\alpha}+1}\bigg)_+.
    \end{equation}
    Analogous to the observation about Theorem \ref{thm:TVcontraction} in Remark \ref{rem:ourtechniquebetterTV}, the presented technique yields a sharper bound than plugging Lemma \ref{lem:PMLimpliesLDP} into \cite[Theorem 1]{10206578}. (However, the point $\gamma^*$ for which $\eta_{\gamma^*}(\mathsf K) =0$ remains identical with both techniques). It is worth pointing out that the behavior outside of the high-privacy regime is significantly different from \eqref{eq:LDPHScontraction}. In fact, whenever $\varepsilon>-\log(c|\mathcal X|)$ in the discrete case, we can construct kernels satisfying $\varepsilon$-PML on $\mathcal P_c$ that contain zero-valued entries in each row, e.g., the kernel $\mathsf K_1$ defined in the introduction, which satisfies $-\log(3c)$-PML on $\mathcal P_c$ and for any $\gamma \geq 1$, $\eta_\gamma(\mathsf K) = \nicefrac{1}{3}$. That is, the contraction is constant in $\gamma$. This shows that contrary to the high-privacy regime $\varepsilon<-\log(c\mu_X(\mathcal X))$, for larger values of $\varepsilon$ there can be no $\gamma^*$ for which $\eta_{\gamma^*}(\mathsf K) = 0$ for all kernels satisfying $\varepsilon$-PML on $\mathcal P_c$. However, as we will see below, when a restriction is made from $\eta_\gamma$ to $\eta_\gamma^{\mathcal P_c}$, this \say{zero-hitting} behavior is recovered for $\varepsilon$-PML guarantees on $\mathcal P_c$ with any $\varepsilon\geq 0$.
\end{remark}

\subsection{Contraction of $E_\gamma$-divergence on $\mathcal P_c$}
Theorem \ref{thm:E_gamma_contraction} provides a bound on the contraction coefficient given that only the channel satisfies an $\varepsilon$-PML on $\mathcal P_c$ constraint, while the input distributions when measuring divergence can be chosen arbitrarily from the simplex. That is, the privacy constraint is treated only as an algebraic property of the kernel. In general, we may be interested in a bound on the quantity $\eta_\gamma^{\mathcal P_c}$, quantifying contraction when the set of input distributions is also restricted to $\mathcal P_c$. The following lemma allows us to obtain a refinement beyond the trivial bound $\eta_\gamma^{\mathcal P_c}(\mathsf K) \leq \eta_\gamma(\mathsf K)$. Specifically, it will allow us to find a characteristic range of parameters $\gamma$ outside of which the $E_\gamma$-divergence between any two input distribution from $\mathcal P_c$ is zero.
\begin{lemma}
    \label{lem:gammamingammamaxbound}
    For $\varepsilon > 0$ and $0<c<\mu_X(\mathcal X)^{-1}$, let,
\begin{equation}
    \Gamma_{\max}(\varepsilon,c) \coloneqq \sup_{P_X,Q_X \in \mathcal P_c} \sup_{\mathsf K\in\mathcal M(\varepsilon,c)}\sup_{B\in\Sigma_\mathcal Y}\frac{(P_X\mathsf K)(B)}{(Q_X\mathsf K)(B)},
\end{equation}
\begin{equation}
    \Gamma_{\min}(\varepsilon,c) \coloneqq \inf_{P_X,Q_X \in \mathcal P_c} \inf_{\mathsf K\in\mathcal M(\varepsilon,c)}\inf_{B\in\Sigma_\mathcal Y}\frac{(P_X\mathsf K)(B)}{(Q_X\mathsf K)(B)}.
\end{equation}
We have $\Gamma_{\max}(\varepsilon,c)\leq e^\varepsilon d_c+1$ and $\Gamma_{\min}(\varepsilon,c)\geq (e^\varepsilon d_c + 1)^{-1}$.
\end{lemma}
Lemma \ref{lem:gammamingammamaxbound} is proved in Appendix \ref{app:lemgammamingammamaxbound} by a simple application of Lemma \ref{lemma:pml_constraints}. The quantities $\Gamma_{\min}$ and $\Gamma_{\max}$ are sometimes referred to \textit{relative information extrema} \cite{binette2019note} or max-divergence due to their close relation to the Rényi-divergence term $\exp(D_\infty(P_X||Q_X))$. Both quantities play an important role in the derivation of divergence inequalities (cf. $\beta_1$ and $\beta_2$ in \cite{7552457}). We will see an application to divergence bounds in Section \ref{sec:sdpi}. Here, the bound in Lemma \ref{lem:gammamingammamaxbound} together with Theorem \ref{thm:E_gamma_contraction} allows us to show the following corollary.
\begin{corollary}
\label{corr:restrictedEgammacontraction}
    Let $\varepsilon>0$ and $0< c<\mu_X(\mathcal X)^{-1}$. If a kernel $\mathsf K$ satisfies $\varepsilon$-PML on $\mathcal P_c$, then for any $\gamma \geq 1$,
    \begin{equation}
        \eta_\gamma^{\mathcal P_c}(\mathsf K) \leq \mathbf 1\{\gamma <e^\varepsilon d_c +1\}\eta_\gamma(\mathsf K).
    \end{equation}
\end{corollary}
\begin{proof}
    Due to the definition of the $E_\gamma$-divergence in \eqref{eq:Egammadef}, we have that, 
    \begin{equation}
         \frac{(P_X\mathsf K)(B)}{(Q_X\mathsf K)(B)} <\gamma \text{ for all } B\in\Sigma_\mathcal Y \implies E_\gamma(P_X\mathsf K||Q_X\mathsf K) = 0. 
    \end{equation}
    The bound on $\Gamma_{\max}(\varepsilon,c)$ in Lemma \ref{lem:gammamingammamaxbound} shows that whenever $\mathsf K$ satisfies $\varepsilon$-PML on $\mathcal P_c$, if $\gamma>e^\varepsilon d_c+1$, then,
    \begin{equation}
        \frac{(P_X\mathsf K)(B)}{(Q_X\mathsf K)(B)} \leq e^\varepsilon d_c+1 < \gamma.
    \end{equation}
    Hence the claim follows.
\end{proof}

The results in Theorem \ref{thm:E_gamma_contraction} and Corollary \ref{corr:restrictedEgammacontraction} together with the integral representation of $f$-divergences in \eqref{eq:fdivRepWithE_gamma} allows us to formulate strong data processing inequality for arbitrary $f$-divergences with twice-differentiable $f$. For the interested reader, we present this application of Theorem \ref{thm:E_gamma_contraction} in Section \ref{subsec:otherapp:nonlinearSPDI}.

%% file: sections/sdpi.tex
We now turn to bounds more specific to the setting of privacy guarantees for distributions in $\mathcal P_c$. We have shown in the preceding section that the restriction from $\mathcal P(\mathcal X)$ to $\mathcal P_c(\mathcal X)$ does not affect the contraction coefficient of the total variation distance. For other divergences, however, the restriction \emph{does} matter, and may often increases contraction in a non-trivial manner, as we have seen in Corollary \ref{corr:restrictedEgammacontraction}. For this reason, in this section, we focus on deriving SDPIs that hold \emph{only} for distributions in $\mathcal P_c$. This will later enable us to derive improved bounds on the minimax risk of estimation problems, given that the estimators satisfy $\varepsilon$-PML on $\mathcal P_c$ for some non-zero $c$.

The results we present are twofold: First, we use Binett's refinement of reverse Pinsker's inequality \cite{binette2019note} to derive a bound on arbitrary $f$-divergences that scales \emph{linearily} in the total variation distance between the input distributions. Second, we obtain SDPIs on the relative entropy in terms of the \emph{squared} total variation distance of the input distributions. This result will serve as a central tool for proving minimax bounds in Section~\ref{sec:minimax}.
\subsection{SDPIs on general $f$-divergences in Terms of Linear Total Variation Distance}
\begin{theorem} 
\label{thm:binettefdivbound}
    Let $f$ be a convex function such that $f(1)=0$. If $\mathsf K$ satisfies $\varepsilon$-PML on $\mathcal P_c$ for some $\varepsilon> 0$ and $0 < c <\mu_X(\mathcal X)^{-1}$, then for any $P_X,Q_X \in\mathcal P_c$, we have,
    \begin{equation}
        D_f(P_X\mathsf K||Q_X\mathsf K) \leq \frac{\eta_\text{TV}(\varepsilon)}{e^\varepsilon d_c}\bigg[f(e^\varepsilon d_c+1)+(e^\varepsilon d_c+1)f\big((e^\varepsilon d_c+1)^{-1}\big)\bigg]\text{TV}(P_X||Q_X).
    \end{equation}
\end{theorem}
\begin{proof}
    It is shown in \cite[Theorem 1]{binette2019note} that, 
    \begin{equation}
        D_f(P||Q) \leq \left(\frac{f(m)}{1-m}+\frac{f(M)}{M-1)}\right)\text{TV}(P||Q),
    \end{equation}
    where $m=\text{ess}\,\inf \nicefrac PQ$ and $M = \text{ess}\,\sup \nicefrac{P}{Q}$. Applying this inequality to $P_X\mathsf K$ and $Q_X\mathsf K$ we obtain, 
    \begin{equation}
    \label{eq:binettestatement}
        D_f(P_X\mathsf K||Q_X \mathsf K) \leq \text{TV}(P_X\mathsf K||Q_X\mathsf K)\Bigg[\frac{f\big(\Gamma_{\max}(\varepsilon,c)\big)}{\Gamma_{\max}(\varepsilon,c)-1} + \frac{f\big(\Gamma_{\min}(\varepsilon,c)\big)}{1-\Gamma_{\min}(\varepsilon,c)}\Bigg],
    \end{equation}
    with $\Gamma_{\max}$ and $\Gamma_{\min}$ defined as in Lemma \ref{lem:gammamingammamaxbound}. Since the multiplicative term on the RHS of \eqref{eq:binettestatement} is increasing in $\Gamma_{\max}$ and decreasing in $\Gamma_{\min}$, we can plug the bound from Lemma \ref{lem:gammamingammamaxbound} into \eqref{eq:binettestatement} to obtain,
    \begin{equation}
        \frac{f(\Gamma_{\max}(\varepsilon,c))}{\Gamma_{\max}(\varepsilon,c)-1} + \frac{f\big(\Gamma_{\min}(\varepsilon,c)\big)}{1-\Gamma_{\min}(\varepsilon,c)} \leq \frac{f(e^\varepsilon d_c+1)}{e^\varepsilon d_c} + (e^\varepsilon d_c+1)\frac{f(1/(e^\varepsilon d_c+1))}{e^\varepsilon d_c}.
    \end{equation}.
    Finally, due to Theorem \ref{thm:TVcontraction}, if $\mathsf K$ satisfies $\varepsilon$-PML on $\mathcal P_c$, then $\text{TV}(P_X\mathsf K||Q_X\mathsf K) \leq \eta_\text{TV}(\varepsilon)\text{TV}(P_X||Q_X)$.
\end{proof}

The result in Theorem \ref{thm:binettefdivbound} holds true for any $f$-divergence. To illustrate the bound, we next specialize it to the two tensorizing divergences, that is, to relative entropy and squared Hellinger divergence \cite{cruz2025tensorization}.
\begin{corollary}[Relative Entropy]
    \label{corr:relativeentropy}
    For any $0< c < \mu_X(\mathcal X)^{-1}$, $\varepsilon> 0$, assume the kernel $\mathsf K$ satisfies $\varepsilon$-PML on $\mathcal P_c$. Then we have for any $P_X,Q_X \in \mathcal P_c$,
    \begin{align}
       &D(P_X\mathsf K||Q_X\mathsf K) \leq \min\bigg\{\frac{e^\varepsilon-1}{e^\varepsilon d_c+1},\, 1\bigg\}\log\Big(e^\varepsilon d_c+1\Big)\text{TV}(P_X||Q_X).
    \end{align}
\end{corollary}

\begin{corollary}[Squared Hellinger Divergence]
\label{corr:hellinger}
    For any $0< c <\mu_X(\mathcal X)^{-1}$, $\varepsilon>0$, assume the kernel $\mathsf K$ satisfies $\varepsilon$-PML on $\mathcal P_c$. Then we have for any $P_X,Q_X \in\mathcal P_c$,
    \begin{equation}
        H^2(P_X\mathsf K||Q_X\mathsf K) \leq \min\bigg\{\frac{e^\varepsilon-1}{e^\varepsilon d_c+1},\, 1\bigg\} \frac{2(\sqrt{e^\varepsilon d_c+1}-1)^2}{e^\varepsilon d_c} \text{TV}(P_X||Q_X).
    \end{equation}
\end{corollary}
\begin{example}
\label{ex:RR}
Let $N=10$ and consider the mechanism $\mathsf K_3 \in \mathcal M(\log\nicefrac{10}{3},0.05)$ defined by
\begin{equation}
    \mathsf K_3 = \begin{bmatrix}
        \nicefrac{15}{16} & \dots & \nicefrac{15}{16} &\nicefrac{1}{16} & \dots & \nicefrac{1}{16} \\
        \nicefrac{1}{16} & \dots & \nicefrac{1}{16} &\nicefrac{15}{16} & \dots & \nicefrac{15}{16}
    \end{bmatrix}^\top,
\end{equation}
further, for $N=5$, consider $\mathsf K_4 \in \mathcal M(\log\nicefrac{10}{3},0.1)$ given by,
\begin{equation}
    \mathsf K_4 = \begin{bmatrix}
        \nicefrac{1}{3} & \nicefrac{1}{3} & \nicefrac{1}{3} & 0 & 0 \\
        0 & \nicefrac{1}{3} & \nicefrac{1}{3} & \nicefrac{1}{3} & 0 \\
        0 & 0 & \nicefrac{1}{3} & \nicefrac{1}{3} & \nicefrac{1}{3} \\
        \nicefrac{1}{3} & 0 & 0 & \nicefrac{1}{3} & \nicefrac{1}{3} \\
        \nicefrac{1}{3} & \nicefrac{1}{3} & 0 & 0 & \nicefrac{1}{3}
    \end{bmatrix}.
\end{equation}
Note that $\mathsf K_3$ also satisfies $\log(15)$-LDP. In Figure \ref{fig:corrbounds}(a), we compare our presented upper bound in Corollary \ref{corr:relativeentropy} to the bound presented by \citet{duchi2013local}, the bound presented by \citet[Theorem 5]{10206578} and the bound obtained by the contraction coefficient of relative entropy presented by \citet{asoodeh2024contraction} with an additional application of reverse Pinsker's inequality in \cite[Theorem 28]{7552457} (with $Q_{\min}=c$) to obtain a bound between $D(P_X\mathsf K||Q_X\mathsf K)$ and $\text{TV}(P_X||Q_X)$. Since $\mathsf K_4$ does not satisfy any finite LDP guarantee, such a comparison is not possible in this case. The bound in Corollary \ref{corr:hellinger} is shown in Figure \ref{fig:corrbounds}(b) for $\mathsf K_4$.
\end{example}

\begin{figure}
    \centering
    \begin{subfigure}[b]{.49\linewidth}
    \centering
    \label{fig:corrbounds:subfig:corr1}
        \includegraphics[scale=0.65]{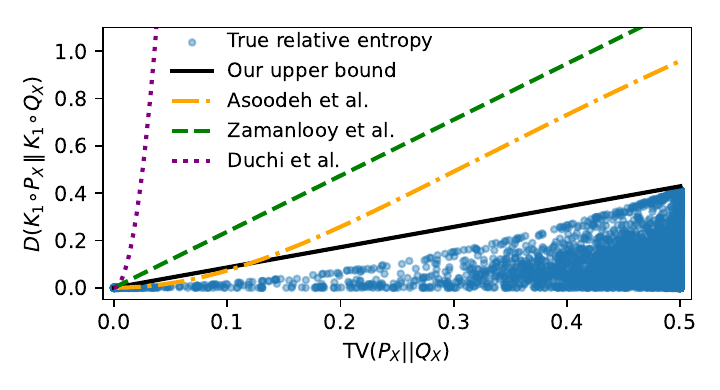}
        \caption{$\mathsf K_3$, $N=10$, $c=0.05$.}
    \end{subfigure}
    \begin{subfigure}[b]{.49\linewidth}
    \centering
    \label{fig:corrbounds:subfig:corr2}
        \includegraphics[scale=0.65]{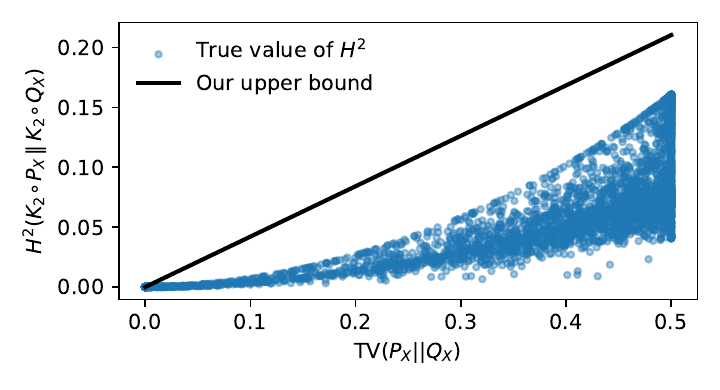}
        \caption{$\mathsf K_4$, $N=5$, $c=0.1$.}
    \end{subfigure}
    \caption{Numerical evaluation of the bounds presented in Corollaries \ref{corr:relativeentropy} and \ref{corr:hellinger} for randomly generated distributions in the set $\mathcal P_c$. $\mathsf K_3$ and $\mathsf K_4$ are defined in Example \ref{ex:RR}. For $\mathsf K_3$, the presented bound is the tightest \emph{linear} bound possible, and is only outperformed by the bound in \cite{asoodeh2024contraction} for samll values of $\varepsilon$. For $\mathsf K_4$, LDP-based bounds fail due to the zero-valued entries in its stochastic matrix.}
    \label{fig:corrbounds}
\end{figure}

\subsection{SDPIs on Tensorizing Divergences in Terms of Squared Total Variation Distance}
\label{sec:tensorizing}
In order to obtain bounds on sample complexities, it is necessary to examine the behavior of the divergences between kernel-outputs given $n$ i.i.d. inputs. That is, for deriving structural minimax bounds, we are interested in the divergences between measures of the form $(P_X\mathsf K)^{\otimes n}$. In this context, we make use of the \emph{tensorization} property of relative entropy, which states that,
\begin{equation}
    D\left((P_X\mathsf K)^{\otimes n}||(Q_X\mathsf K)^{\otimes n}\right) = n D(P_X \mathsf K ||Q_X\mathsf K).
\end{equation}
In this section, we present two SDPIs between relative entropy and the squared total variation distance. Both of the inequalities are proved via a bound on the $\chi^2$-distance, for which we have $D(P||Q)\leq \chi^2(P||Q)$. At first glance, it might seem like the choice of \emph{squared} total variation distance is arbitrary, and that we might as well have simply used the SDPIs in (linear) terms of total variation presented the preceeding section. However, the following relation conceptually shows why SDPIs in terms of squared total variation are particularly useful for for deriving sample complexity bounds: Let $\eta$ be a constant such that $D(P_X\mathsf K||Q_X\mathsf K)\leq \eta \text{TV}^2(P_X||Q_X)$. Applying Pinsker's inequality to the $n$-fold product distribution of outputs, we see that,
\begin{equation}
    \text{TV}((P_X\mathsf K)^{\otimes n}||(Q_X\mathsf K)^{\otimes n}) \leq \sqrt{\frac{n}{2}D(P_X\mathsf K||Q_X\mathsf K)} \leq \sqrt{\frac{\eta \cdot n}{2}}\,\text{TV}(P_X||Q_X).
\end{equation}
That is, an SDPI on relative entropy in terms of \emph{squared} total variation allows us to bound the total variation distance between $n$-fold product output distributions in terms of the total variation of the single-letter input distributions.

\begin{theorem}
\label{thm:chi^2_TV^2_nr2}
    Let $0<c<\mu_X(\mathcal X)^{-1}$ and $\varepsilon> 0$. If a kernel $\mathsf K$ satisfies $\varepsilon$-PML on $\mathcal P_c$, then for all $P_X,Q_X\in\mathcal P_c$,
    \begin{equation}
        D(P_X\mathsf K||Q_X \mathsf K) \leq \,\min\{4\,\Xi(\varepsilon,c)\text{TV}^2(P_X||Q_X),\,\tilde\Xi(\varepsilon,c)\text{TV}(P_X||Q_X)\}.
    \end{equation}
    where we define, $\Xi(\varepsilon,c)\coloneqq\min\{A,B\}$, $\tilde \Xi(\varepsilon,c)\coloneqq\min\{d_cA,B\}$, with,\footnote{To cover the case $1-e^\varepsilon c\mu_X(\mathcal X)=0$, we define $1/0 \equiv \infty$ here, meaning that in this case, the minimum function always chooses $A$.}
    \begin{equation}
        A \coloneqq\frac{e^{2\varepsilon}}{e^\varepsilon d_c +1}, \,\quad B\coloneqq  \frac{(e^\varepsilon-1)^2}{e^\varepsilon d_c(1-e^\varepsilon c\mu_X(\mathcal X))_+}.
    \end{equation}

\end{theorem}

\begin{proof}
    It is shown in \cite[Theorem 2]{asoodeh2024contraction}, that if $\mathsf M \in\mathcal M(\alpha,0)$, then,
    \begin{equation}
        \chi^2(P_X \mathsf M||Q_X \mathsf M) \leq \frac{(e^{\alpha}-1)^2}{e^{\alpha}}\min\{4\text{TV}^2(P_X||Q_X),\,\text{TV}(P_X||Q_X)\}, \quad \forall P_X,Q_X \in \mathcal P(\mathcal X).
    \end{equation}
    Using the implication from $\varepsilon$-PML on $\mathcal P_c$ to $(1+e^\varepsilon d_c)$-LDP in Lemma \ref{lem:PMLimpliesLDP}, we can therefore write,
    \begin{equation}
        \chi^2(P_X\mathsf K||Q_X\mathsf K) = \chi^2(\bar P_X\mathsf M||\bar Q_X\mathsf M) \leq \frac{4e^{2\varepsilon} d_c^2}{e^\varepsilon d_c +1}\text{TV}^2(\bar P_X||\bar Q_X). 
    \end{equation}
    What remains is to take note of the fact that,
    \begin{equation}
        \text{TV}(\bar P_X||\bar Q_X) = \text{TV}\left(\frac{P_X - c\mu_X}{d_c}\bigg|\bigg|\frac{Q_X-c\mu_X}{d_c}\right) = \frac{1}{d_c}\text{TV}(P_X||Q_X).
    \end{equation}
    Hence we have, 
    \begin{equation}
        \chi^2(P_X\mathsf K||Q_X\mathsf K) \leq \min\left\{\frac{4e^{2\varepsilon}}{e^\varepsilon d_c +1} \text{TV}^2(P_X||Q_X),\,\frac{e^{2\varepsilon }d_c}{e^\varepsilon d_c+1}\text{TV}(P_X||Q_X)\right\} \quad \forall P_X,Q_X \in \mathcal P_c.
    \end{equation}
    The second part of the minium is shown by the second part of Lemma \ref{lem:PMLimpliesLDP}: If $\varepsilon <-\log(c\mu_X(\mathcal X))$, then $\mathsf K$ itself satisfies $\alpha$-LDP with parameter $\alpha = (e^\varepsilon d_c)/(e^\varepsilon d_c(1-e^\varepsilon c \mu_X(\mathcal X)))$. Plugging this into \cite[Theorem 2]{asoodeh2024contraction} yields,
    \begin{equation}
        \chi^2(P_X\mathsf K||Q_X\mathsf K) \leq \frac{(e^\varepsilon-1)^2}{e^\varepsilon d_c (1-e^\varepsilon c \mu_X(\mathcal X))}\text{TV}^2(P_X||Q_X), \quad \varepsilon<-\log(c\mu_X(\mathcal X)).
    \end{equation}
     The desired result follows, since $D(P||Q)\leq \chi^2(P||Q)$ \cite{Polyanskiy_Wu_2025}.
\end{proof}

%% file: sections/minimax.tex
In this section, we will apply the technical results derived above to two statistical problems: Binary hypothesis testing and mean estimation. Even though both of these problems are somewhat canonical textbook problems, they remain central in modern data processing systems. For example, the gradient aggregation phase in federated learning essentially boils down to a mean estimation problem \cite{alistarh2017qsgd}, and many modern applications like watermark-testing in LLMs \cite{kirchenbauer2023watermark} or best-arm identification in multi-armed bandits \cite{garivier2016optimal} are closely connected to hypothesis testing. 

We will examine the binary hypothesis testing problem with $n$ samples and the Bayesian case with one sample under PML constraints first. After, we will analyze the mean estimation problem with PML privacy. We will first present the simple one-dimensional case as a warm-up. Then, we devote a large part of this section to the treatment of the mean estimation problem in general $d$-dimensions. The results show that using $c$-interior PML in estimation task can be beneficial in the sense that it often decreases the sample complexity of problems compared to LDP. This means that the incorporation of distributional assumptions into the privacy guarantees in $c$-interior PML can indeed lead to better estimation procedures. 
\subsection{Binary Hypothesis Testing}
\label{sec:minimax:subsec:nhypo}
Given $n$ i.i.d. samples $X^n= \{X_i\}_{i\in[n]}$, the goal of binary hypothesis testing is to distinguish between two hypothesis $H_0 = \{X^n\text{ was generated by }P_X^{\otimes n}\}$ and $H_1 =\{X^n\text{ was generated by }Q_X^{\otimes n}\}$. To do so privately, each $X_i$ is passed through a kernel $\mathsf K$ satisfying $\varepsilon$-PML on $\mathcal P_c$ to produce privatized samples $Y^n \sim (R_X\mathsf K)^{\otimes n}$, where $R_X \in \{P_X,Q_X\}$. Below, we investigate the effect of this privatization on the sample complexity assuming that $P_X,Q_X \in \mathcal P_c$ for some $c\geq 0$. 

Let $\phi_n:\mathcal X^n \to \{0,1\}$ be a (possibly randomized) test function, where we interpret,
\begin{equation}
    \phi_n(X^n) = \begin{cases}
        0, \text{ accept }H_0, \\
        1, \text{ accept }H_1.
    \end{cases}
\end{equation} 
Let $\alpha_n(\phi_n)$ and $\beta_n(\phi_n)$ denote the type-I and type-II error probabilities associated with $\phi_n$, respectively, that is,
\begin{equation}
    \alpha_n(\phi_n) = \mathbb P_{X^n\sim P_X^{\otimes n}}[\phi_n(X^n)=1], 
\end{equation}
\begin{equation}
    \beta_n(\phi_n) = \mathbb P_{X^n\sim Q_X^{\otimes n}}[\phi_n(X^n)=0].
\end{equation}
The total minimax error probability of test $\phi_n$ is given by the probability of error for equal priors \cite{Polyanskiy_Wu_2025}, that is,
\begin{equation}
    P_e^{(n)} = \frac{1}{2}\Big(\alpha_n(\phi_n) + \beta_n(\phi_n)\Big).
\end{equation}
With these quantities at hand, we define the sample complexity of distinguishing between two distributions $P_X$ and $Q_X$ with error at most $\zeta\in(0,1)$ as,
\begin{equation}
    n^*(P_X,Q_X,\zeta) = \inf\left\{n\in\mathbb N: \inf_{\phi_n}P_e^{(n)} \leq \zeta \right\}
\end{equation}
In what follows, we are interested in the sample complexity of the \emph{privatized} problem defined by,
\begin{equation}
    n^*_{\varepsilon,c}(P_X,Q_X,\zeta) \coloneqq \inf_{\mathsf K\in\mathcal M(\varepsilon,c)}\; n^*\left((P_X \mathsf K)^{\otimes n},(Q_X\mathsf K)^{\otimes n},\zeta\right).
\end{equation}
The main result of this section is summarized in the following theorem. Recall that we define, $\varepsilon^* \coloneqq -\log (\nicefrac{c\mu_X(\mathcal X)}{2})$ as the threshold to the low-privacy regime. By this definition, for any $\varepsilon$ in the low-privacy regime ($\varepsilon\geq \varepsilon^*$), we have $\eta_\text{TV}(\varepsilon)=1$. The remainder of this section will be devoted to proving and discussing the following theorem.

\begin{theorem}
\label{thm:hypothesistestingH^2asymp}
      Let $0< c< \mu_X(\mathcal X)^{-1}$, and $\zeta \in(0,\nicefrac{1}{2})$ be arbitrary but fixed. For any $P_X,Q_X \in\mathcal P_c$ and any $\varepsilon > 0$,
    \begin{equation}
          n^*_{\varepsilon,c}\left(P_X,Q_X,\zeta\right) \asymp \frac{1}{ \eta_\text{TV}^2(\varepsilon)\text{TV}^2(P_X||Q_X)} \stackrel{\mathcal X\text{ discrete}}{\asymp} \frac{1}{ \eta_\text{TV}^2(\varepsilon)H^2(P_X||Q_X)},
    \end{equation}
    where the implicit constants depend only on $c$ and $\zeta$.
\end{theorem}

The theorem shows that the introduction of $c$-interior PML privacy into a hypothesis testing problem turns out to result in a degradation of the problem's sample complexity by a factor of $\eta_\text{TV}^{-2}(\varepsilon)$. In particular, in the low-privacy regime $\varepsilon>\varepsilon^*$, we see that the $\varepsilon$-PML on $\mathcal P_c$ comes at no cost to the sample complexity. We will compare to the non-private case (and the LDP case) in more detail below. The following non-asymptotic bound constitute the core of the proof of Theorem \ref{thm:hypothesistestingH^2asymp}. 

\begin{proposition}
\label{prop:nhypothesistesting}
    Let $\varepsilon > 0$ and $0< c< \mu_X(\mathcal X)^{-1}$. For any $P_X,Q_X \in\mathcal P_c$ and any $\zeta \in(0,1)$,
    \begin{equation}
      \frac{(1-2\zeta)^2}{2\,\Xi(\varepsilon,c)\,\text{TV}^2(P_X||Q_X)}\leq n^*_{\varepsilon,c}\left(P_X,Q_X,\zeta\right) \leq \frac{8\log(1/\zeta)}{\eta_\text{TV}^2(\varepsilon)\text{TV}^2(P_X||Q_X)}.
    \end{equation}
\end{proposition}
\begin{proof}
    We begin by showing the lower bound. It is known that the error probability of binary hypothesis testing is proportional to $1-\text{TV}(P||Q)$ \cite{Polyanskiy_Wu_2025}, specifically, 
    \begin{align}
        P_e^{(n)} &= \frac{1}{2}\left[1-\text{TV}((P_X\mathsf K)^{\otimes n}||(Q_X\mathsf K)^{\otimes n})\right].
\end{align}
To show the lower bound, we write
\begin{align}
        \frac12-\frac{1}{2}\text{TV}((P_X\mathsf K)^{\otimes n}||(Q_X\mathsf K)^{\otimes n})&\stackrel{(i)}{\geq} \frac{1}{2}\left[1-\sqrt{\frac{1}{2}D((P_X\mathsf K)^{\otimes n}||(Q_X\mathsf K)^{\otimes n}})\right] \\[.5em]
        &\stackrel{(ii)}{\geq} \frac{1}{2} \left[1-\sqrt{\frac{n}{2}D(P_X\mathsf K||Q_X\mathsf K)}\right] \\[.5em]
        &\stackrel{\text{Thm. \ref{thm:chi^2_TV^2_nr2}}}{\geq} \frac{1}{2}\left[1-\sqrt{2n\,\Xi(\varepsilon,c)}\,\text{TV}(P_X||Q_X)\right], 
    \end{align}
    where $(i)$ follows from Pinsker's inequality, and $(ii)$ follows from the tensorization property of relative entropy. To achieve $P_e^{(n)} \leq \zeta$ in this term, we must choose $n$ large enough such that,
    \begin{equation}
        \frac{1}{2}\left[1-\sqrt{2n\,\Xi(\varepsilon,c)}\,\text{TV}(P_X||Q_X)\right] \leq \zeta,
    \end{equation}
    or, equivalently, we must choose some $n$ for which, 
    \begin{equation}
        n \geq \frac{(1-2\zeta)^2}{2\,\Xi(\varepsilon,c)\,\text{TV}^2(P_X||Q_X)}.
    \end{equation}

To show the upper bound, we proceed in a similar fashion as \cite[Lemma 2]{asoodeh2024contraction}. Throughout, define $\varepsilon^\star \coloneqq -\log\big(c\,\mu_X(\mathcal X)/2\big)$ and $\bar\varepsilon\coloneqq\min\{\varepsilon,\varepsilon^\star\}$, so that $\eta_\text{TV}(\varepsilon)=1$ for all $\varepsilon\geq\varepsilon^\star$. In particular,
$\eta_\text{TV}(\bar\varepsilon)=\eta_\text{TV}(\varepsilon)$ for every $\varepsilon>0$. We will need the following lemma, proved in Appendix~\ref{app:proofTVseparatinglemma}.
\begin{lemma}
    \label{lem:TVseparatinglemma}
    Let $(\mathcal X,\Sigma_\mathcal X,\mu_X)$ be a finite measure space, let $0<\varepsilon\leq\varepsilon^\star$ and
    $0<c<\mu_X(\mathcal X)^{-1}$. Then for all $P_X,Q_X\ll\mu_X$, there is a measurable function
    $a:\mathcal X\to[0,1]$ such that with $m_a\coloneqq\int a\,d\mu_X$,
    \begin{equation}
        \max\{m_a,\ \mu_X(\mathcal X)-m_a\}\leq(e^\varepsilon c)^{-1},
        \quad\text{and}\quad
        \left|\int a\,dP_X-\int a\, dQ_X\right|\geq\frac12\,\text{TV}(P_X||Q_X).
    \end{equation}
\end{lemma}
Heuristically, this lemma states that as long as $\varepsilon\leq\varepsilon^\star$, we can always find a test that separates $P_X$ and $Q_X$ \say{well} (within a factor $\nicefrac12$ of the optimal likelihood-ratio test attaining $\mathrm{TV}(P_X\|Q_X)$), while satisfying the mass condition required for the construction of the optimal mechanism in Theorem~\ref{thm:TVcontraction}.\footnote{If $\mu_X$ is non-atomic, $a$ may be taken to be the indicator of a set $\mathcal A$ with $\max\{\mu_X(\mathcal A),\mu_X(\mathcal A^c)\}\leq(e^\varepsilon c)^{-1}$.}

Let $\mathsf K^\star_{\bar\varepsilon,c}$ be the kernel construction in Theorem \ref{thm:TVcontraction} for $\bar\varepsilon$, built from the function $a$ obtained by applying Lemma \ref{lem:TVseparatinglemma}
with $\bar\varepsilon$, and let $Y^n\sim(R_X\mathsf K^\star_{\bar\varepsilon,c})^{\otimes n}$
for $R_X\in\{P_X,Q_X\}$. Since $\bar\varepsilon\leq\varepsilon$, this mechanism is satisifies $\varepsilon$-PML on $\mathcal P_c$. It achieves,
\begin{equation}
\label{eq:approxoptimalTVseparation}
    \text{TV}\big(P_X\mathsf K^\star_{\bar\varepsilon,c}||Q_X\mathsf K^\star_{\bar\varepsilon,c}\big)
    = \eta_\text{TV}(\bar\varepsilon)\left|\int a\,dP_X-\int a\, dQ_X\right|
    \geq \frac{\eta_\text{TV}(\varepsilon)}{2}\,\text{TV}(P_X||Q_X),
\end{equation}
where the equality follows from Theorem \ref{thm:TVcontraction} and the inequality follows
from Lemma \ref{lem:TVseparatinglemma} together with $\eta_\text{TV}(\bar\varepsilon)=\eta_\text{TV}(\varepsilon)$. Note that the case $\varepsilon>\varepsilon^\star$, in which no admissible test exists at level $\varepsilon$
itself, is handled by the definition of $\bar\varepsilon$.

According to \cite[Theorem 2]{canonne2022short}, we can distinguish $H_0$ and $H_1$ up to error $\zeta$ whenever
    \begin{equation}
    \label{eq:hypoproof_YdomainUpperbound}
        n \geq \frac{2\log(1/\zeta)}{H^2\big(P_X\mathsf K^\star_{\bar\varepsilon,c}||Q_X\mathsf K^\star_{\bar\varepsilon,c}\big)}.
    \end{equation}
By the standard inequality $\text{TV}(P||Q)\leq H(P||Q)$, and \eqref{eq:approxoptimalTVseparation},
   \begin{equation}
       H^2(P_X\mathsf K^\star_{\bar\varepsilon,c}||Q_X\mathsf K^\star_{\bar\varepsilon,c})
       \geq \text{TV}^2(P_X\mathsf K^\star_{\bar\varepsilon,c}||Q_X\mathsf K^\star_{\bar\varepsilon,c})
       \geq \frac{\eta_\text{TV}^2(\varepsilon)}{4}\,\text{TV}^2(P_X||Q_X).
   \end{equation}
Combining this with \eqref{eq:hypoproof_YdomainUpperbound} shows that error $\zeta$ is achievable as soon as, 
   \begin{equation}
       n \geq \frac{8\log(1/\zeta)}{\eta_\text{TV}^2(\varepsilon)\,\text{TV}^2(P_X||Q_X)},
   \end{equation}
which finishes the proof.
\end{proof}

That the test statistic constructed in Proposition \ref{prop:nhypothesistesting} via $\mathsf K^\star_{\varepsilon,c}$ from Theorem \ref{thm:TVcontraction} is order-optimal is shown by the following Lemma, which we prove in Appendix \ref{app:etaTVsq=Xi}. This allows us to finish the proof of Theorem \ref{thm:hypothesistestingH^2asymp}.

\begin{lemma}
    \label{lem:etaTVsq=Xi}
    Let $0<c<\mu_X(\mathcal X)^{-1}$ be a fixed constant and define $\varepsilon^*\coloneqq \log2-\log(c\mu_X(\mathcal X))$. Then there exists $a(c),A(c)\in\mathbb R$ such that for any $\varepsilon\in[0,\varepsilon^*)$,
    \begin{equation}
        a(c) \eta_\text{TV}^2(\varepsilon) \leq \Xi(\varepsilon,c)\leq A(c) \eta_\text{TV}^2(\varepsilon).
    \end{equation}
\end{lemma}

\begin{proof}[Proof of Theorem \ref{thm:hypothesistestingH^2asymp}]
   Lemma \ref{lem:etaTVsq=Xi} shows that $\eta^2_\text{TV}(\varepsilon) \asymp \Xi(\varepsilon,c)$ as long as $0< \varepsilon<\varepsilon^*$. Further, since $\varepsilon\geq\varepsilon^*$ implies $\eta_\text{TV}(\varepsilon)=1$, and the private sample complexity by set inclusion is necessarily lower-bounded by the non-private sample complexity \cite{canonne2022short}, we have asymptotic equivalence in the remaining region. The equivalence between $\text{TV}^2$ and $H^2$ on discrete spaces is proved in Appendix \ref{app:lem:TV^2_H^2}. 
\end{proof}

The upper bound in Proposition \ref{prop:nhypothesistesting} can be strengthened by a factor of $4$ under additional assumptions. With these assumptions, we can go beyond order-optimality, and show that the test-statistic constructed above reduces to a version of the standard likelihood-ratio test (see Remark \ref{rem:privVSnonprivhypo}).
\begin{corollary}
\label{cor:optimalupperPQsep}
    For $\varepsilon>0$ and $0<c<\mu_X(\mathcal X)^{-1}$, define,
    \begin{equation}
    \mathcal {PQ}_c^{\varepsilon}(\mathcal X) \coloneqq \left\{(P,Q)\mid P,Q \in\mathcal P_c(\mathcal X): \max\{\mu_X(S_{P,Q}),\,\mu_X(S_{P,Q}^c)\}\leq (e^\varepsilon c)^{-1}\right\},
\end{equation}
where $S_{P,Q}\in \Sigma_\mathcal X$ is a set such that $P(A)-Q(A)=\text{TV}(P||Q)$. We have for each pair $(P_X,Q_X)\in\mathcal{PQ}^\varepsilon_c$,
\begin{equation}
    n_{\varepsilon,c}^*(P_X,Q_X,\zeta) \leq \frac{2\log(1/\zeta)}{\eta_\text{TV}^2(\varepsilon)\text{TV}^2(P_X||Q_X)}.
\end{equation}
\end{corollary}
\begin{proof}
 By definition, whenever $P_X, Q_X$ are from $\mathcal {PQ}^\varepsilon_c(\mathcal X)$, we can pick $\mathcal A = S_{P,Q}$ in the optimal kernel constructed in the proof of Proposition \ref{prop:nhypothesistesting}. Hence $\text{TV}^2(P_X\mathsf K_{\varepsilon,c}^\star||Q_X\mathsf K_{\varepsilon,c}^\star) = \eta^2_\text{TV}(\varepsilon)\text{TV}^2(P_X||Q_X)$, saving a factor of $4$. 
 \end{proof}

\begin{remark}
\label{rem:privVSnonprivhypo}
Theorem \ref{thm:hypothesistestingH^2asymp} allows us to quantify the \say{cost of privacy} in binary hypothesis testing with $\varepsilon$-PML on $\mathcal P_c$. In the non-private case, it is shown in \cite{canonne2022short,bar2002complexity} that, 
\begin{equation}
    n^*(P_X,Q_X,\zeta) \asymp \frac{1}{H^2(P_X||Q_X)}.
\end{equation}
Comparing this result with Theorem \ref{thm:hypothesistestingH^2asymp}, we see that the decrease in effective sample size introduced by privatization with a kernel satisfying $\varepsilon$-PML on $\mathcal P_c$ scales with $\eta_\text{TV}^{-2}(\varepsilon)$ on discrete spaces. Whenever $\eta_\text{TV}(\varepsilon)=1$, the sample complexity of the private problem is the same as that of the non-private problem up to a factor of at most $4$. This happens whenever $\varepsilon \geq \varepsilon^*$. That is, in the low-privacy regime, privacy in hypothesis testing comes \emph{for free} in terms of sample complexity. Further, Corollary \ref{cor:optimalupperPQsep} implies that if the distributions are picked from $\mathcal{PQ}_c^{\varepsilon}$, the standard likelihood-ratio test together with a binary mechanism according to Theorem \ref{thm:TVcontraction} is order-optimal.
\end{remark}

Theorem \ref{thm:hypothesistestingH^2asymp} partially recovers the result in \cite{pensia2023simple,asoodeh2024contraction} in \eqref{eq:LDPprivnhypo} as $c\to 0$. Further, the restriction from $\mathcal P(\mathcal X)$ to $\mathcal P_c$ for some $c>0$ leads to the observation that $\text{TV}^2(P_X||Q_X)\geq cH^2(P_X||Q_X)$ on discrete spaces (see Appendix \ref{app:lem:TV^2_H^2}). The result in \eqref{eq:LDPprivnhypo} therefore implies that on $\mathcal P_c$, the non-private sample complexity can be achieved with $\alpha$-LDP whenever $\alpha>-\log(c)$. Theorem \ref{thm:hypothesistestingH^2asymp} shows that the non-private complexity under $\varepsilon$-PML on $\mathcal P_c$ can be achieved for,
\begin{equation}
    \varepsilon\ \geq \varepsilon^\star = -\log\left(\frac{c|\mathcal X|}{2}\right)
     = -\log(c) - \log\left(\frac{|\mathcal X|}{2}\right).
\end{equation}
Since \eqref{eq:LDPprivnhypo} is stated for $|\mathcal X|=2$, the statement in \cite{pensia2023simple} implies ours for binary distributions. For $|\mathcal X|>2$, \cite{pensia2023simple} shows that the free-privacy threshold $\alpha^*$ in LDP scales as $\alpha^* \asymp H^{-2}(P_X||Q_X)$. In contrast, assuming regularity of the problem in form of the parameter $c$ reintroduces a \say{free-privacy} threshold that is independent of the specific distributions $P_X$, $Q_X$, while retaining full adversarial protection for the specific pair $P_X,Q_X\in\mathcal P_c$ in the adversarial model of PML in \eqref{eq:PMLdefGainFunc}, and only discarding protection for overly pessimistic data distributions (those not in $\mathcal P_c$).

\subsection{Bayesian testing with $n=1$}
In what follows, we assume that $n=1$, and that the probability of hypothesis $H_0$ being true is $\pi_0$, while the probability of $H_1$ being true is $\pi_1=1-\pi_0\leq \pi_0$. The Bayesian error probability of test $\phi_n$ is given by,
\begin{equation}
    P_e^{\text{Bayes}}(\pi_1) = (1-\pi_1) \alpha_n(\phi_n) + \pi_1\beta_n(\phi_n).
\end{equation}
It can be shown that the optimal test given sample $X=x$ is to choose $H_0$ whenever $\pi_0 p(x) > \pi_1 q(x)$ ($p$ and $q$ denoting densities w.r.t $\mu_X$), which yields the (minimal) Bayesian error probability,
\begin{equation}
    P_e^*(\pi_1) = \int_\mathcal X \min\{\pi_0 p(x),\pi_1q(x)\}d\mu_X.
\end{equation}
Using the identity $\min\{a,b\} = a-\max\{0,a-b\}$, we can rewrite this error probability as,
\begin{align}
    P_e^*(\pi_1) &= \int_\mathcal X \pi_1q(x) - \left[\pi_1q(x)-\pi_0p(x)\right]_+d\mu_X =\pi_1 - \pi_1 \int_\mathcal X \left[q(x) - \frac{\pi_0}{\pi_1}p(x)\right]_+d\mu_X \\[.5em]
    &= \pi_1 - \pi_1 E_{\frac{\pi_0}{\pi_1}}(Q_X||P_X).
\end{align}
That is, the error probability of the optimal Bayesian test is determined by the $E_\gamma$-divergence of $Q_X$ from $P_X$ with $\gamma = \nicefrac{\pi_0}{\pi_1}$.\footnote{We remark that this relation was, to the best of our knowledge, first observed in \cite{liu2016e_}.} The result in Theorem \ref{thm:E_gamma_contraction} therefore enables us to obtain a private bound for the Bayesian hypothesis testing problem.
\begin{proposition}
\label{prop:1hypotestEgamma}
    Let $P_X,Q_X\in\mathcal P_c$ and assume that $\mathsf K \in\mathcal M(\varepsilon,c)$ for some $\varepsilon \geq 0$ and $0< c < \mu_X(\mathcal X)^{-1}$. Let $\Psi_\gamma$ be defined according to Theorem \ref{thm:E_gamma_contraction}. The Bayesian error probability of any test $\phi$ operating on a privatized sample $Y\sim R_X\mathsf K$ with $R_X\in\{P_X,Q_X\}$ is bounded from below by,
    \begin{equation}
        P_e^*(\pi_1) \geq \pi_1 - \pi_1\Psi_{\frac{\pi_0}{\pi_1}}(e^\varepsilon,c) \bm 1\left\{\frac{\pi_0}{\pi_1}<e^\varepsilon d_c+1\right\} E_{\frac{\pi_0}{\pi_1}}(Q_X||P_X).
    \end{equation}
\end{proposition}
We point at two observations relating $\varepsilon$-PML on $\mathcal P_c$ to the impossibility of distinguishing $P_X$ and $Q_X$ beyond a blind guess: First, whenever $\varepsilon \leq \log\left(\frac{1}{d_c}\left[\frac{\pi_0}{\pi_1}-1\right]\right)$, we have $P_e^*(\pi_1) = \pi_1$, that is, the privacy is so strict that the optimal error probability is obtained by blindly guessing $H_0$ and disregarding the observation $X=x$. Second, for $c\to0$, we recover the result of \citet{10206578}, which states that privately testing between $P_X$ and $Q_X$ is impossible beyond a blind guess whenever $E_{\frac{\pi_0}{\pi_1}}(Q_X||P_X) \leq \frac{(\pi_0-\pi_1)}{\pi_1(e^\varepsilon+1)}$. 

\subsection{Mean Estimation}
 In this section, we consider the problem of estimating the mean of a distribution $P_X \in\mathcal P_c$. We begin by focusing on the one-dimensional problem, for which a standard Le\,Cam argument suffices. For the multi-dimensional case, we will need to apply Assouad's method, which is why we separate the case $d=1$ from $d\geq 2$. For any $P_X \in \mathcal P(\mathcal X)$, let $\theta(P_X)$ denote a parameter of a random variable distributed according to $P_X\in \mathcal P$, where $\mathcal P$ is some pre-specified class of distribution (e.g. normal with fixed variance). Let $\hat \theta: \mathcal X^n\to \mathbb R$ be any estimator which forms an estimate of $\theta (P_X)$ from observations $X^n = \{X_i\}_{i=1}^n \sim P_X^{\otimes n}$. The non-private minimax estimation risk of the problem is defined as,
\begin{equation}
    \mathfrak R_n\left(\theta(\mathcal P),\Phi\circ \rho\right) = \inf_{\hat \theta} \sup_{P_X\in\mathcal P} \mathbb E_{X^n \sim P_X^{\otimes n}}\left[\Phi\left(\rho\left(\hat \theta(X^n),\theta(P_X)\right)\right)\right].
\end{equation}
We are interested in the private counterpart of this risk. To this end, let again $Y^n \sim (P_X \mathsf K)^{\otimes n}$, where $\mathsf K \in \mathcal M(\varepsilon,c)$ satisfies $\varepsilon$-PML on $\mathcal P_c$. We define the $(\varepsilon,c)$-private minimax estimation risk for some $\varepsilon>0$, $0 < c< \mu_X(\mathcal X)^{-1}$ as,
\begin{equation}
\label{eq:privateminimaxrisk}
    \mathfrak R_n^{\varepsilon,c}(\theta(\mathcal P),\Phi\circ \rho) = \inf_{\mathsf K \in \mathcal M(\varepsilon,c)}\inf_{\hat \theta} \sup_{P_X \in\mathcal P_c \cap \mathcal P} \mathbb E_{Y^n \sim (P_X \mathsf K)^{\otimes n}}\left[\Phi\left(\rho\left(\hat \theta(Y^n),\theta(P_X)\right)\right)\right].
\end{equation}
Note that $\mathfrak R_n^{\infty,0}(\theta(\mathcal P),\Phi\circ \rho) = \mathfrak R_n(\theta(\mathcal P),\Phi\circ \rho)$. 
In what follows, we will obtain bounds on $\mathfrak R_n^{\varepsilon,c}(\theta(\mathcal P),\Phi\circ \rho)$ for the mean estimation problem. 
\subsubsection{One dimensional mean estimation}
Consider $\mathcal P = \mathcal P_c$ and $\theta(P_X) = \mathbb E_{X\sim P_X}[X]$ with $\mathcal X =[-b,b]$ for some $b\in(0,\infty)$. We take the loss function to be the squared Euclidean distance, that is, we fix $\Phi(t) = t^2$ and $\rho(x,x') = |x-x'|$. The following result characterizes the asymptotic behavior of the minimax-risk for this problem.

\begin{theorem}
\label{thm:onedmean}
    For $b\in(0,\infty)$, let $\mathcal X=[-b,b]$. For $\varepsilon> 0$ and $0< c <\mu_X(\mathcal X)^{-1}$, we have,
    \begin{equation}
        \mathfrak R^{\varepsilon,c}_n(\theta(\mathcal P),|\cdot|^2) \asymp \frac{b^2}{n\,\eta_{\rm TV}^2(\varepsilon)}.
    \end{equation}
    where the implicit constants only depend on $c$ and $\mu_X(\mathcal X)$.
\end{theorem}

\begin{remark}
    Theorem \ref{thm:onedmean} reveals a similar structure as already shown in the hypothesis testing problem: Whenever $\varepsilon \geq \varepsilon^* = -\log(bc)$, the standard non-private risk can be achieved while satisfying $\varepsilon$-PML on $\mathcal P_c$. In other words, in the low-privacy regime, estimating the population mean of the data can be done privately \emph{for free}, without incurring any additional risk. Further, the estimator constructed below for the mean estimation problem is \emph{order-optimal}.
\end{remark}
We will now prove Theorem \ref{thm:onedmean}. As previously, we begin by proving a non-asymptotic bound.

\begin{proposition}
\label{prop:meanestimation}
    Let $\varepsilon> 0$ and $0< c <\mu_X(\mathcal X)^{-1}$. Then for $\Phi = (\cdot )^2$ and $\rho$ the Euclidean distance,
    \begin{equation}
        \min\left\{\frac{b^2d_c^2}{4},\,\frac{b^2}{32n\Xi(\varepsilon,c)}\right\} \leq \mathfrak R^{\varepsilon,c}_n(\theta(\mathcal P),\Phi\circ \rho) \leq  \frac{b^2}{n\,\eta_\text{TV}^2(\varepsilon)}.
    \end{equation}
\end{proposition}
\begin{proof}
    We follow Le\,Cam's method (see, e.g., \cite[Chapter 15]{wainwright2019high}). Since $\mu_X$ denotes the Lebsegue measure on the real numbers restricted to the interval $[-b,b]$, we have $d_c = 1- c\mu_X(\mathcal X) = 1-2bc$. Fix $\delta \in (0,1)$ to be picked later, and with a smoothing parameter $\kappa\in(0,b)$, define,
    \begin{equation}
        I_+ = [b-\kappa,b],\quad I_- = [-b,-b+\kappa], \quad b_\kappa = b-\frac{\kappa}{2}.
    \end{equation}
    Let $\bar P_X^\kappa$, $\bar Q_X^\kappa$ denote distributions that admit the $\mu_X$-densities $\bar p_\kappa$, $\bar q_\kappa$, defined as,
    \begin{equation}
        \bar p_\kappa = \frac{1+\delta}{2\kappa}\mathbf 1\{x\in I_+\} +\frac{1-\delta}{2\kappa}\mathbf 1\{x\in I_-\}, \quad \bar q_\kappa = \frac{1+\delta}{2\kappa}\mathbf 1\{x\in I_-\} +\frac{1-\delta}{2\kappa}\mathbf 1\{x\in I_+\},
    \end{equation}
    that is, $\bar P_X^\kappa$ puts mass $(1+\delta)/(2\kappa)$ uniformly on the interval $I_+$ and mass $(1-\delta)/(2\kappa)$ uniformely on $I_-$, and vice-versa for $\bar Q_X^\kappa$. Define the mixtures, 
    \begin{equation}
        P^\kappa_X \coloneqq c\mu_X(\mathcal X) \cdot U_X  + (1-c\mu_X(\mathcal X))\cdot \bar P^\kappa_X,\quad Q^\kappa_X \coloneqq c\mu_X(\mathcal X) \cdot U_X  + (1-c\mu_X(\mathcal X))\cdot \bar Q^\kappa_X,
    \end{equation}
    where $U_X$ denotes the uniform distribution on $[-b,b]$.  Both distributions are absolutely continuous with respect to $\mu_X$ and admit densities lower-bounded by $c$, so $P^\kappa_X,Q^\kappa_X \in \mathcal P_c$. This construction yields,
    \begin{equation}
        \rho(\theta(P_X),\theta(Q_X)) = |\mathbb E_{X\sim P^\kappa_X}[X]-\mathbb E_{X\sim Q^\kappa_X}[X]| = 2 b_\kappa \delta d_c.
    \end{equation}
    Now, by Le\,Cam's method, a lower bound on the estimation risk from $Y^n\sim (P_X\mathsf K)$ is,
    \begin{align}
        \mathfrak R_n^{\varepsilon,c}(\theta(\mathcal P),\Phi\circ \rho) &\geq \frac{b_\kappa^2 \delta^2d_c^2}{2} \left(1-\text{TV}\left((P^\kappa_X\mathsf K)^{\otimes n}||(Q^\kappa_X\mathsf K)^{\otimes n}\right)\right) \\[.5em]
        &\geq \frac{b_\kappa^2 \delta^2d_c^2}{2} \left(1-\sqrt{\frac{1}{2}D\left((P^\kappa_X\mathsf K)^{\otimes n}||(Q^\kappa_X\mathsf K)^{\otimes n}\right)}\right) \\[.5em]
        &\geq \frac{b_\kappa^2 \delta^2d_c^2}{2} \left(1-\sqrt{\frac{n}{2} D\left(P^\kappa_X\mathsf K||Q^\kappa_X\mathsf K\right)}\right) \\[.5em]
        &\geq \frac{b_\kappa^2 \delta^2d_c^2}{2} \left(1-\sqrt{2n\Xi(\varepsilon,c)\text{TV}^2\left(P^\kappa_X||Q^\kappa_X\right)}\right).
    \end{align}
    Observe that from the definition of $P_X$ and $Q_X$, we have $P^\kappa_X-Q^\kappa_X = d_c(\bar P^\kappa_X-Q_X^\kappa)$ since the uniform terms cancel, and further $\bar p_\kappa - \bar q_\kappa = \nicefrac{\delta}{\kappa}\mathbf 1_{\{I_{+}\cup I_-\}}$. Therefore,
    \begin{equation}
        \text{TV}(P_X||Q_X) = d_c \text{TV}(\bar P_X||\bar Q_X) = \frac{d_c}{2} \int |\bar p_\kappa - \bar q_\kappa|d\mu_X = d_c\delta.
    \end{equation}
    Hence, we have, 
    \begin{equation}
        \mathfrak R_n^{\varepsilon,c}(\theta(\mathcal P),\Phi\circ \rho) \geq \frac{b_\kappa^2 \delta^2d_c^2}{2} \left(1-\sqrt{2n\Xi(\varepsilon,c)d_c^2\delta^2}\right)
    \end{equation}
    Finally, we choose $\delta$ such that the divergence term is constant equal to $\nicefrac{1}{2}$ as, 
    \begin{equation}
        \delta^2 = \min\left\{1,\,\frac{1}{8n\Xi(\varepsilon,c)d_c^2}\right\}. 
    \end{equation}
    With this choice, we have for any $\kappa \in (0,b)$,
    \begin{equation}
        \mathfrak R_n^{\varepsilon,c}(\theta(\mathcal P),\Phi\circ \rho) \geq \min\left\{\frac{b_\kappa^2d_c^2}{4},\,\frac{b_\kappa^2}{32n \Xi(\varepsilon,c)}\right\} = \left(1-\frac{\kappa}{2b}\right)^2\min\left\{\frac{b^2d_c^2}{4},\,\frac{b^2}{32n \Xi(\varepsilon,c)}\right\}.
    \end{equation}
    Taking the limit $\kappa \to 0$ shows that,
    \begin{equation}
        \mathfrak R_n^{\varepsilon,c}(\theta(\mathcal P),\Phi\circ \rho) \geq \min\left\{\frac{b^2d_c^2}{4},\,\frac{b^2}{32n \Xi(\varepsilon,c)}\right\}.
    \end{equation}
    This proves the lower bound. We show the upper bound by constructing a variant of the \emph{Duchi-Jordan-Wainwright (DJW) estimator} \cite{duchi2013local}. For each private sample $X_i \sim P_X$, let $Z_i$ be a binary random variable taking values $\{-b,b\}$ with,
    \begin{equation}
        \mathbb P[Z_i=b\mid X_i] = \frac{1}{2} + \frac{X_i}{2b}, \quad \mathbb P[Z_i=-b\mid X_i] = \frac{1}{2}-\frac{X_i}{2b}.
    \end{equation}
    Note that this ensures that $\mathbb E[Z_i\mid X_i] = X_i$. Further, note that by this definition, we have,
    \begin{align}
        \inf_{P_X\in\mathcal P_c}\mathbb P[Z=b] &= \inf_{P_X\in\mathcal P_c}\int_{[-b,b]}\mathbb P[Z=b|X=x]dP_X(x) \\[.5em]
        &= \frac{1}{2} - \sup_{P_X\in\mathcal P_c}\frac{|\mathbb E_{P_X}[X]|}{2b} \\[.5em] 
        &\geq \frac{1}{2}- \frac{b(1-2bc)}{2b} = bc,
    \end{align}
    and since the argument is symmteric, the same holds for $\mathbb P[Z=-b]$. Hence, if we privatize $Z_i$ to outputs $Y_i$, then the Markov chain $X_i-Z_i-Y_i$ holds, and if the mechanism mapping each $Z_i$ to $Y_i$ satisfies $\varepsilon$-PML for any binary distribution $P_Z$ on $\{-b,b\}$ with $\min_z P_Z(z) \geq bc = \nicefrac{c\mu_X(\mathcal X)}{2}$, then the overall procedure mapping from $X_i$ to $Y_i$ satisfies $\varepsilon$-PML for any distribution in $\mathcal P_c$. We differentiate between the following two cases.   
    \begin{enumerate}
        \item $\varepsilon < -\log(bc)$: For each $Z_i$, let the binary random variable $Y_i$ be defined by $Z_i \circ \mathsf K_{\varepsilon,c}^\star$, where according to Theorem \ref{thm:TVcontraction}, 
    \begin{equation}
        \mathsf K_{\varepsilon,c}^\star =\frac{1}{e^\varepsilon (1-2bc)+1}\begin{bmatrix}
            e^\varepsilon (1-bc) & 1-e^\varepsilon bc \\
            1-e^\varepsilon bc & e^\varepsilon (1-bc)
        \end{bmatrix}
    \end{equation}
    From $Y_i$, let the estimator $\hat \theta:\{-b,b\}^n \to \mathbb R$ be defined as,
    \begin{equation}
        \hat \theta(Y^n) = \frac{1}{n} \left(\frac{e^\varepsilon(1-2bc)+1}{e^\varepsilon-1}\right)\sum_{i=1}^n Y_i.
    \end{equation}
    We have $\mathbb E[\hat \theta(Y^n)] = \mathbb E[X_i] = \theta$, hence the squared error of the estimator is it's variance, given by,
    \begin{equation}
        \mathsf{Var}(\hat \theta(Y^n))= \frac{b^2(e^\varepsilon(1-2bc)+1)^2}{n(e^\varepsilon-1)^2}.
    \end{equation}

    \item $\varepsilon \geq -\log(bc)$: Follows by the same steps as above, only now with the deterministic kernel,
    \begin{equation}
        \mathsf K_{\varepsilon,c}^\star = \begin{bmatrix}
            1 & 0 \\
            0 & 1
            \end{bmatrix},
    \end{equation}
    and consequently the estimator,
    \begin{equation}
        \hat \theta(Y^n) = \frac{1}{n}\sum_{i=1}^n Y_i.
    \end{equation}
    This estimator is clearly unbiased, and we have,
    \begin{equation}
        \mathsf{Var}(\hat \theta(Y^n)) = \frac{b^2}{n}.
    \end{equation}
    \end{enumerate}
    Combining both of the above cases, we obtain,
    \begin{equation}
        \mathbb E[(\hat \theta - \theta )^2] = \frac{b^2}{n} \max\left\{1,\left(\frac{e^\varepsilon(1-2bc)+1}{e^\varepsilon-1}\right)^2\right\} = \frac{b^2}{n\,\eta_\text{TV}^2(\varepsilon)}.
    \end{equation}
    This proves the statement of Proposition \ref{prop:meanestimation}.
\end{proof}

To finish the proof of Theorem \ref{thm:onedmean}, we show that the upper and lower bounds in Proposition \ref{prop:meanestimation} coincide asymptotically. 

\begin{proof}[Proof of Theorem \ref{thm:onedmean}]
    For $\varepsilon<\varepsilon^*$ we have $\eta_\text{TV}^2(\varepsilon) \asymp \Xi(\varepsilon,c)$ according to Lemma \ref{lem:etaTVsq=Xi}, hence the bounds in Proposition \ref{prop:meanestimation} coincide here. Whenever $\varepsilon\geq \varepsilon^*$ we have $\eta_\text{TV}(\varepsilon)=1$, and hence the constructed estimator achieves the non-private sample complexity, which is also a lower bound to the private problem, showing order-optimality.
\end{proof}

\subsubsection{Multi-dimensional mean estimation}
We now turn to the case of $d\geq 2$, that is, we assume that $\mathcal X = \mathcal B_d^p$ is the $p$-norm unit ball $\mathcal B_d^2 = \{x\in\mathbb R^d: ||x||_p\leq 1\}$ and examine the problem of finding a lower bound on \eqref{eq:privateminimaxrisk} with $\Phi(t)=t^2$ and $\rho(x,x')=||x-x'||_2$. In contrast to the one dimensional case, in this setting, the reduction to simple binary hypothesis testing results in vacuous bounds. Instead, we employ Assouad's method \cite{Polyanskiy_Wu_2025}, alongside a packing construction based on a hypercube packing inspired by \cite{duchi2013local} and the SDPI in Theorem \ref{thm:chi^2_TV^2_nr2}.

\begin{proposition}
\label{prop:highdmeanestimation}
    Let $p\in[1,2]$ and assume that $d\geq 2$. Let $\mathcal X=\mathcal B_d^p$. For $\Phi(t)=t^2$ and $\rho(x,x')=||x-x'||_2$, $n\geq d^2/(2\Xi(\varepsilon,c)d_c^2)$, we have,
    \begin{equation}
        \mathfrak R^{\varepsilon,c}_n(\theta(\mathcal P),||\cdot ||_2^2) \gtrsim  \max\left\{\frac{1}{n},\frac{d}{n\Xi(\varepsilon,c)}\right\}.
    \end{equation}
\end{proposition}

\begin{proof}
    We employ Assouad's Lemma as formulated in \cite[Theorem 31.2]{Polyanskiy_Wu_2025}. Note that this version of Assouad's Lemma deviates from the standard formulation by introducing a $\alpha$-triangle inequality for the loss that enables us to restrict the maximization over Hamming-neighbors $v,v': d_H(v,v')=1$. The squared Eucledean distance satisfies this $\alpha$-triangle inequality with $\alpha =2$. We will need the following Lemma, which follows from standard packing constructions and is proved in Appendix \ref{app:packingproof}.
\begin{lemma}
\label{lem:multidimpacking}
    For any $p\in[1,2]$, $d\geq 2$ and $\delta\leq \nicefrac{1}{2}$, let $\mathcal X=\mathcal B_d^p$. There exists a set of distributions indexed by the hypercube $\{\pm1\}^d$, such that $\{P_\nu\}_{\nu\in\{\pm 1\}^d} \subset \mathcal P_c(\mathcal X)$ and for any $v,v'\in\{\pm 1\}^d$,
    \begin{equation}
        ||\theta(P_\nu)-\theta(P_{\nu'})||_2^2 \geq \frac{4\delta^2d_c^2}{d^2}d_H(v,v'),\quad \text{and} \quad \text{TV}(P_\nu||P_{\nu'}) = \frac{\delta d_c}{d}d_H(v,v'),
    \end{equation}
    where $d_H(\cdot,\cdot)$ denotes the Hamming distance function.
\end{lemma}
Let $\mathsf K\in\mathcal M(\varepsilon,c)$. Given the packing distributions $\{P_v\}_v$ for which $||\theta(P_v)-\theta(P_{v'})||_2^2\geq \beta d_H(v,v')$ for some $\beta>0$, Assouad's Lemma in \cite[Theorem 31.2]{Polyanskiy_Wu_2025} for our private setting and the squared loss yields,
\begin{align}
    \mathfrak R_n^{\varepsilon,c}(\theta(\mathcal P),||\cdot||^2_2) &\geq \frac{\beta d}{8}\left(1-\max_{v,v'\in\mathcal V:d_H(v,v')=1}\text{TV}((P_v\mathsf K)^{\otimes n}||(P_{v'}\mathsf K)^{\otimes n})\right) \\[.5em]
    &\stackrel{(a)}{\geq} \frac{\beta d}{8}\left(1-\max_{d_H(v,v')=1}\sqrt{\frac{n}{2}D(P_v\mathsf K||P_{v'}\mathsf K)}\right) \\[.5em]
    &\stackrel{(\text{Thm }\ref{thm:chi^2_TV^2_nr2})}{\geq} \frac{\beta d}{8}\left(1-\sqrt{2n\Xi(\varepsilon,c)\text{TV}^2(P_v||P_{v'})}\right),
\end{align}
where $(a)$ follows from Pinsker's inequality and the tensorization property of relative entropy. With the packing in Lemma \ref{lem:multidimpacking} for some $\delta \leq \nicefrac{1}{2}$,
\begin{equation}
     \mathfrak R_n^{\varepsilon,c}(\theta(\mathcal P),||\cdot||^2_2) \geq \frac{\delta^2d_c^2}{2d}\left(1-\sqrt{\frac{2n\Xi(\varepsilon,c)d_c^2 \delta^2}{d^2}}\right).
\end{equation}
By choosing,
\begin{equation}
     \delta^2 = \min\left\{\frac{1}{4},\,\frac{d^2}{8n\Xi(\varepsilon,c)d_c^2}\right\},
\end{equation}
the term inside the bracket is bounded by  $\nicefrac{1}{2}$ from below, and therefore we have for $n \geq d^2/(2\Xi(\varepsilon,c)d_c^2)$,
\begin{equation}
    \mathfrak R_n^{\varepsilon,c}(\theta(\mathcal P),||\cdot||^2_2) \geq \frac{d}{32n\Xi(\varepsilon,c)} \asymp \frac{d}{n\Xi(\varepsilon,c)},
\end{equation}
which shows the first part of the bound. To show the remaining part of the rate, note that, necessarily by set inclusion, the private risk is lower-bounded by the non-private risk, so we have
\begin{equation}
    \mathfrak R_n^{\varepsilon,c}(\theta(\mathcal P),||\cdot||^2_2) \geq \mathfrak R_n(\theta(\mathcal P),||\cdot||^2_2) \asymp\frac{1}{n}.
\end{equation}
Putting the two bounds together, we obtain the desired result.
\end{proof}

We provide a partially matching upper bound for the lower bound in Proposition \ref{prop:highdmeanestimation} in the following proposition. The statement is based on the construction of a variant of the multi-dimensional DJW estimator \cite{duchi2013local} on the Euclidean unit-ball. It is proved in Appendix \ref{app:highd-estimator-proof}.

\begin{proposition}
\label{prop:highd-estimator}
    Let $d\geq 2$ and $p\in[1,2]$. For any $\varepsilon> 0$, we have,
    \begin{equation}
        \mathfrak R^{\varepsilon,c}_n(\theta(\mathcal P),||\cdot||_2^2) \lesssim \frac{d}{n\,\eta^2_\text{TV}(\varepsilon)},
    \end{equation}
    where $\eta_\text{TV}(\varepsilon)$ is defined with $d_c = 1-c\mathbb V(\mathcal B_d^p)$.
\end{proposition}

Interestingly, the bound in Proposition \ref{prop:highdmeanestimation} shows that under $c$-interior $\varepsilon$-PML, even high-dimensional mean estimation may be non-restrictive when the uniform floor $c$ is relatively large and the privacy requirement $\varepsilon$ is not to strict. Specifically, the presented bound, if tight, would imply that whenever $\Xi(\varepsilon,c)\geq d$, then the asymptotic behavior of high dimensional mean estimation with $\varepsilon$-PML on $\mathcal P_c$ is equivalent to the non-private setting. To confirm this insight, we now present a \emph{private} estimator for large values of $\varepsilon$ that is asymptotically optimal up to universal constants. The \say{privacy-part} of this estimator is achieved via a discretization of the unit sphere, together with the fact that for discrete alphabets, the identity mapping satisfies a finite $\varepsilon$-PML guarantee on $\mathcal P_c$ for any $c>0$. The proof of the privacy guarantee of this construction relies on the existence of a vertex-transitive covering of the unit-sphere. This is a strongly simplifying assumption. However, it is important to point out that the presented construction is in no way meant to constitute a practical scheme. The main point of stating it is to show that the lower bound in Proposition \ref{prop:highdmeanestimation} is not vacuous.

\begin{remark}
\label{rem:vertextransEst}
    Assume that $\rho<\nicefrac{\pi}{2}$ and $d\geq 2$ are such that there exists a vertex-transitive $\rho$-covering of the unit-sphere $\mathbb S^{d-1}$ with cardinality $M$.\footnote{A discrete set $\mathcal Z= \{z_1,\dots,z_M\}$ is called vertex-transitive if there exists a finite orthogonal group $G$ that acts transitively on $\mathcal Z$, i.e., $\mathcal Z= Gz=\{gz\mid g\in G\}$ for any $z\in\mathcal Z$. This in particular implies that $\exists g:z_i=gz_k\,\forall i,k$ and that any $g\in G$ is a bijection.} Then there exists a mean estimator with rate equal to the non-private rate up to universal constants that satisfies $\varepsilon^*$-PML on $\mathcal P_c$ with,
    \begin{equation}
        \varepsilon^* = \log(M)-\log(c\mathbb V(\mathcal B_d^2)).
    \end{equation}
    The proof of this existence result relies on the symmetry of the covering and an application of Strassen's theorem of martingale couplings \cite{strassen1965existence}. We present it in detail in Appendix \ref{app:proofVertextTransEst}. Note that the requirement of a symmetric (vertex-transitive) codebook is akin to the construction of optimal estimation strategies for $\varepsilon$-LDP in \cite{isik2023exact}. Even further, we are tempted to conjecture that there exist more explicit estimator constructions that match the lower bound in Proposition~\ref{prop:highdmeanestimation} via, e.g., Kashin's representations \cite{chen2020breaking,lyubarskii2010uncertainty} or similar representations of points on a sphere. However, we remark that there are two significant challenges in designing such estimators: First, proving PML privacy for any such estimator requires proving a lower bound on the marginal probabilities $\mathbb P[Z=z]$ of the discrete symbols, which requires significant structural assumptions on the representations. Second, the privatization of the discrete symbols with mechanisms without full output support for every input significantly complicates the debiasing procedure. Due to these difficulties, we leave the exploration of concrete constructions for future work.
\end{remark}

%% file: sections/otherapplications.tex
\subsubsection{Conversion to Approximate LDP Guarantees}
 It is shown in \cite{9517999} that a mechanism $\mathsf K$ satisfies $(\alpha,\delta)$-LDP if and only if $\eta_{e^\alpha}(\mathsf K)\leq \delta$. This directly leads to the following result relating $\varepsilon$-PML on $\mathcal P_c$ to $(\alpha,\delta)$-LDP.
\begin{corollary}
\label{cor:PMLtoALDP}
    If $\mathsf K$ satisfies $\varepsilon$-PML on $\mathcal P_c$, then it also satisfies $(\alpha, \delta(\alpha))$-LDP for any $\alpha \geq 0$, where $\gamma^*=e^\varepsilon d_c/(1-e^\varepsilon c\mu_X(\mathcal X))$, and,
    \begin{equation}
        \delta(\alpha) = \eta_\text{TV}(\varepsilon) \cdot \begin{cases}
            \left(\frac{\gamma^*-e^\alpha}{\gamma^*-1}\right)_+ &\text{if }\varepsilon <-\log(c\mu_X(\mathcal X)) \\
            1, &\text{otherwise}.
        \end{cases}
    \end{equation}
\end{corollary}

A few interesting insights can follow: Firstly, since $\varepsilon$-PML on $\mathcal P_0=\mathcal P(\mathcal X)$ is equivalent to $\varepsilon$-LDP, Corollary~\ref{cor:PMLtoALDP} can be used obtain a conversion from pure LDP to approximate LDP. To this end, notice that if $c=0$, then $\gamma^* = e^\varepsilon$. Therefore, if $\mathsf K$ satisfies $\varepsilon$-LDP, then it also satisfies $(\alpha,\delta(\alpha ))$-LDP for any $\alpha \geq 0$ and,
    \begin{equation}
        \delta(\alpha) = \left(\frac{e^\varepsilon-e^\alpha}{e^\varepsilon+1}\right)_+.
    \end{equation}
Second, and in line with many of the illustrative remarks above, we highlight the behavior of $(\alpha,\delta)$-LDP for discrete kernels with zero entries. The bound in Corollary \ref{cor:PMLtoALDP} states that whenever a kernel does not satisfy any pure LDP guarantee (i.e., $\varepsilon \geq -\log(c\mu_X(\mathcal X))$, then $\varepsilon$-PML on $\mathcal P_c$ merely implies $(\alpha,\eta_\text{TV}(\varepsilon))$-LDP for any $\alpha \geq 0$, in particular, the kernel satisfies $(0,\eta_\text{TV}(\varepsilon))$-LDP. This bound turns out to be tight for a large class of kernels: Consider any discrete kernel for which $\forall y \in\mathcal Y, \exists x\in\mathcal X: \mathsf K(y|x)=0$ (each column of the representing matrix has at least one zero entry). It becomes clear by inspection that each such kernel can only satisfy $(\alpha,\delta)$-LDP whenever,
\begin{equation}
    \delta \geq \max_{y\in\mathcal Y} \max_{x\in\mathcal X} \mathsf K(y|x),
\end{equation}
and this bound is often independent of $\alpha$. This highlights a core \emph{inflexibility} of $(\alpha,\delta)$-LDP as privacy measure beyond pure LDP: In regimes outside of pure LDP, the slack parameter $\delta$ often fails to characterize properties of kernels that quantify meaningful privacy properties. As an example of one such shortcoming, consider the kernel $\mathsf K_4$ as given in Example \ref{ex:RR}. This kernel satisfies $(\alpha,\delta)$-LDP for $\delta = \nicefrac{2}{3}$ for any $\alpha \geq 0$. The tightest characterization of $\mathsf K_2$ with approximate LDP is therefore $(0,\nicefrac{2}{3})$-LDP (and any other valid guarantee is slack). This aligns with the observations in \cite{bun2019heavy}, where it is shown that $(\alpha,\delta)$-LDP does not offer any structural benefit over pure $\alpha$-LDP, and that any approximate LDP protocol can be converted into a pure LDP protocol with the same accuracy. On the other hand, the adversarial assumptions in $\varepsilon$-PML on $\mathcal P_c$ allow us to precisely quantify the type of privacy protection such a kernel offers, where $\varepsilon$ quantifies a notion of disclosed information \cite[Theorem 3.5]{saeidian2023inferential} and $c$ encodes assumptions about the distribution of the secret data. We therefore argue that for quantifying privacy of systems without any finite LDP guarantee, $c$-interior PML offers a significantly more flexible meaningful toolset than approximate LDP. To further illustrate this point, we return to a setup similar to that given in the introduction: Consider,
\begin{equation}
   \mathsf K_5= \begin{bmatrix}
        \frac{1}{4} & \frac{1}{4} & \frac{1}{4} & \frac{1}{4} \\
        \frac{1}{4}& \frac{1}{4} &\frac{1}{4} & \frac{1}{4} \\
        \frac{1}{3} & 0 & \frac{1}{3} & \frac{1}{3} \\
        0 & 1 & 0 & 0
    \end{bmatrix},
    \quad \mathsf K_6 = \begin{bmatrix}
        1 & 0 & 0 & 0 \\
        0 & 1 & 0 & 0 \\
        0 & 0 & 1 & 0 \\
        0 & 0 & 0 & 1
    \end{bmatrix}.
\end{equation}
Since $\mathsf K_5$ has a column $[\nicefrac{1}{4},\nicefrac{1}{4},0,1]^\top$, it satisfies at best $(0,1)$-LDP. In terms of approximate LDP it is therefore exactly equivalent to the identity mapping $\mathsf K_6$. The quantification with $c$-interior PML, however, yields for any $0<c<0.25$,
\begin{equation}
    \sup_y\sup_{P_X\in\mathcal P_c}\ell_{P_X\mathsf K_1}(X\to y) = -\log\frac{3c}2 < -\log c = \sup_y\sup_{P_X\in\mathcal P_c} \ell_{P_X\mathsf K_2}(X\to y).
\end{equation}

\subsubsection{Contraction coefficient of Rényi-divergence of order $\infty$ and LDP Amplification}
It is shown in \cite[Theorem 4]{vandenbroucque2026contraction} that the distribution-dependent contraction coefficient of the Rényi-divergence of order $\infty$ is bounded as,
\begin{equation}
    \min_{x\in\mathcal X} P_X(x) \eta_\text{TV}(P_X,\mathsf K) \leq \eta_{D_\infty}(P_X,\mathsf K) \leq \frac{\eta_\text{TV}(P_X,\mathsf K)}{\min_{y\in\supp(P_X\mathsf K)}(P_X\mathsf K)(y)}.
\end{equation}
The bound in Theorem \ref{thm:TVcontraction} enables us to bound the set-restricted contraction coefficient of the $\infty$-Rényi-divergence as follows.
\begin{corollary}
\label{cor:Renyicontraction}
    Let $\mathcal X$ and $\mathcal Y$ be discrete sets and assume that $\varepsilon\geq 0$ and $0 < c <|\mathcal X|^{-1}$. If a kernel $\mathsf K$ satisfies $\varepsilon$-PML on $\mathcal P_c$ and    $D_\mathsf K\coloneqq  \min_{y\in \mathcal Y} \sum_{x\in \mathcal X}\mathsf K(y|x) >0$, then,
    \begin{equation}
         \eta^{\mathcal P_c}_{D_\infty}(\mathsf K) \leq \frac{\eta_\text{TV}(\varepsilon)}{c\cdot D_\mathsf K} = \min\left\{1\,,\frac{e^\varepsilon-1}{(e^\varepsilon d_c+1)c D_\mathsf K}\right\}. 
    \end{equation}
\end{corollary}
Under some regularity conditions, this result enables us to provide convergence rates for LDP privacy amplification with channels that satisfy no LDP guarantee, but some $\varepsilon$-PML guarantee. To see this, note that a kernel satisfies $\alpha$-LDP if and only if $D_\infty(\mathsf K(\cdot | x)||\mathsf K(\cdot | x'))\leq \alpha$ for all $x\neq x'$. This leads to the following proposition.

\begin{proposition}[LDP amplification via $\varepsilon$-PML post-processing]
\label{prop:ldp_amplification}
    Let $\mathsf K_0: \mathcal{X} \to \mathcal{Z}$ be a discrete kernel satisfying $\alpha$-LDP. For some $n\in\mathbb N$, let $\{\mathsf K_i\}_{i=1}^n$ be a sequence of discrete kernels mapping $\mathcal Z \to \mathcal Y$ with $|\mathcal Z|=|\mathcal Y|$, each satisfying $\varepsilon$-PML on $\mathcal{P}_c$ for some $0<c<|\mathcal Z|^{-1}$. Define $\mathsf K^{(n)} \coloneqq \mathsf K_0 \mathsf K_1\dots  \mathsf K_n$, and let $Q^{(n)} = U_X \mathsf K^{(n)}$, where $U_X$ denotes the uniform distribution on $\mathcal X$. Let $\mathcal Y^* \coloneqq \supp(Q^{(n)})$. Then $\mathsf K^{(n)}$ satisfies $\alpha_n$-LDP, where
    \begin{equation}
        \alpha_n \leq \alpha \left(\frac{e^\alpha}{\min_{y\in\mathcal Y^*} Q^{(n)}(y)}\right) \, \eta_{\rm TV}(\varepsilon)^n.
    \end{equation}
\end{proposition}

\begin{proof}
    By the definition of $\alpha$-LDP, $\mathsf K_0$ satisfies $D_\infty(\mathsf K_0(\cdot|x) || \mathsf K_0(\cdot|x')) \leq \alpha$ for all $x,x'\in\mathcal X$. We bound the divergence after post-processing by $\mathsf K_{1:n} \coloneqq \mathsf K_1 \dots \mathsf K_n$. First, let $Q_0(z) \coloneqq \frac{1}{|\mathcal X|}\sum_{x'\in\mathcal X}\mathsf K_0(z|x')$ denote the marginal distribution under uniform input. Since $\mathsf K_0$ satisfies $\alpha$-LDP, we have that for all $x, x'\in\mathcal X$ and $z\in\mathcal Z$, $\mathsf K_0(z|x) \geq e^{-\alpha} \mathsf K_0(z|x')$. Fix $x\neq x'$ arbitrary. Averaging over all $x'\in\mathcal X$ yields,
    \begin{equation}
        \mathsf K_0(z|x) \geq e^{-\alpha} Q_0(z). \label{eq:lower_bound_K0}
    \end{equation}
    With this, the conditional probability of the final output $y$ given input $x$ is lower-bounded by
    \begin{align}
        \mathsf K^{(n)}(y|x) &= \sum_{z\in\mathcal Z} \mathsf K_0(z|x) \mathsf K_{1:n}(y|z) \\
        &\geq e^{-\alpha} \sum_{z\in\mathcal Z} Q_0(z)\mathsf K_{1:n}(y|z) \\
        &= e^{-\alpha} Q^{(n)}(y).
    \end{align}
    Because we chose $x\in\mathcal X$ arbitrary, this shows that $\min_{x\in\mathcal X} \mathsf K^{(n)}(y|x) \geq e^{-\alpha} Q^{(n)}(y)$. Next, we plug into \cite[Theorem 4]{vandenbroucque2026contraction}. Notice that $\supp(\mathsf K^{(n)}(\cdot|x)) = \supp(Q^{(n)}) \eqqcolon \mathcal Y^*$. We have,
    \begin{equation}
        \eta_{D_\infty}(\mathsf K_0(\cdot|x), \mathsf K_{1:n}) \leq \frac{\eta_{\rm TV}(\mathsf K_{1:n})}{\min_{y\in\mathcal Y^*} \mathsf K^{(n)}(y|x)} \leq \frac{\eta_{\rm TV}(\mathsf K_{1:n})}{e^{-\alpha}\min_{y\in\mathcal Y^*} Q^{(n)}(y)}.
    \end{equation}
   Further, it is known \cite{cohen1998comparisons} that $\eta_{\rm TV}(\mathsf K_1\mathsf K_2) \leq \eta_{\rm TV}(\mathsf K_1) \eta_{\rm TV}(\mathsf K_2)$ and hence (since all $\mathsf K_i$ satisfy $\varepsilon$-PML on $\mathcal P_c$), we have $\eta_{\rm TV}(\mathsf K_{1:n}) \leq \eta_{\rm TV}(\varepsilon)^n$ according to Theorem \ref{thm:TVcontraction}. This yields,
    \begin{equation}
        D_\infty\left(\mathsf K^{(n)}(\cdot|x) \big|\big| \mathsf K^{(n)}(\cdot|x')\right) \leq \eta_\infty(\mathsf K_0(\cdot|x), \mathsf K_{1:n}) D_\infty(\mathsf K_0(\cdot|x) || \mathsf K_0(\cdot|x')) \leq \frac{e^\alpha \eta_{\rm TV}(\varepsilon)^n D_\infty(\mathsf K_0(\cdot|x) || \mathsf K_0(\cdot|x'))}{\min_{y\in\mathcal Y^*} Q^{(n)}(y)}.
    \end{equation}
    Since $D_\infty(\mathsf K_0(\cdot|x) || \mathsf K_0(\cdot|x')) \leq \alpha$, we obtain,
    \begin{equation}
        \alpha_n \leq \alpha \left(\frac{e^\alpha}{\min_{y\in\mathcal Y^*} Q^{(n)}(y)}\right)  \,\eta_{\rm TV}(\varepsilon)^n,
    \end{equation}
    which is what we wanted to show.
\end{proof}

Proposition \ref{prop:ldp_amplification} shows that if the post-processing kernels are non-absorbing (such that $\min_{y\in\mathcal Y^*} Q^{(n)}(y)$ does not decay to zero as $n\to \infty$), the privacy guarantee $\alpha_n$ decays to zero exponentially fast for all $\varepsilon$ where $\eta_{\rm TV}(\varepsilon) < 1$, that is, whenever $\varepsilon < -\log((c\mu_X{\mathcal X})/2)$. Notably, such a class of non-absorbing kernels is the set of all doubly-stochastic matrices. For this class of kernels, we can further specialize the result.

\begin{corollary}
\label{cor:doubly_stochastic_amplification}
Assume that in the setting of Proposition \ref{prop:ldp_amplification}, all kernels $\{\mathsf K_i\}_{i=0}^n$ are doubly stochastic. Then $\mathsf K^{(n)} = \mathsf K_0\dots \mathsf K_n$ satisfies $\alpha_n^{\text{d-s}}$-LDP, where,
\begin{equation}
    \alpha_n^\text{d-s} \leq \alpha e^\alpha|\mathcal X|\eta_{\rm TV}(\varepsilon)^n.
\end{equation}
\end{corollary}

\begin{proof}
  We find $\min_{y}Q^{(n)}(y)$ in the statement of Proposition \ref{prop:ldp_amplification}. Note that products of doubly stochastic matrices are themselves doubly stochastic \cite[8.7.P1]{horn2012matrix}, so if each $\mathsf K_i$ is doubly stochastic, so is $\mathsf K^{(n)}$. This means that,
  \begin{equation}
      \min_{y\in\mathcal Y} Q^{(n)}(y) =\min_{y\in\mathcal Y} \frac{1}{|\mathcal X|}\sum_{x\in\mathcal X}\mathsf K^{(n)}(y|x) = |\mathcal X|^{-1},
  \end{equation}
  which is what we wanted to show.
\end{proof}

It is worth pointing out that the restriction to doubly stochastic matrices as in Corollary \ref{cor:doubly_stochastic_amplification} is of practical relevance: Due to the symmetry of a local differential privacy constraint, many important discrete mechanisms for $\varepsilon$-LDP satisfy double stochasticity \cite{kairouz2016extremal}. Further, there is a large class of doubly stochastic mechanisms that do not satisfy any finite $\varepsilon$-LDP guarantee, but for which on the other hand $\eta_{\rm TV}(\varepsilon) < 1$. Specifically, any mechanism with zero entries, but for which $\supp(\mathsf K(\cdot | x))\cap \supp(\mathsf K(\cdot |x'))$ is nonempty for any pair $x,x'$ will satisfy $\eta_{\rm TV}(\varepsilon)<1$. $n$-fold post processing with such a doubly stochastic mechanism will hence lead to an exponential decrease in the local differential privacy parameter of the system. 

\subsubsection{Non-linear SDPIs via Integral Representations}
\label{subsec:otherapp:nonlinearSPDI}
Corollary \ref{corr:restrictedEgammacontraction} shows that for any $c>0$, a finite $\varepsilon$-PML guarantee on $\mathcal P_c$ implies that for large enough $\gamma$, we have $\eta_\gamma^{\mathcal P_c}(\mathsf K)=0$. In light of the integral representation of $f$-divergences in \eqref{eq:fdivRepWithE_gamma}, this implies that we can use Theorem \ref{thm:E_gamma_contraction} to obtain non-linear strong data processing inequalities restricted to input distributions from $\mathcal P_c$.

\begin{corollary}
\label{cor:SDPIEgamma_restricted}
    Let $f:(0,\infty)\to \mathbb R$ be a convex twice-differentiable function such that $f(1)=0$. If a kernel $\mathsf K$ satisfies $\varepsilon$-PML on $\mathcal P_c$, then for any $P_X,Q_X \in\mathcal P_c$,
    \begin{equation}
    \label{eq:SDPIInt}
        D_f(P_X\mathsf K||Q_X\mathsf K) \leq \int_1^{e^\varepsilon d_c +1} \Psi_\gamma(e^\varepsilon,c)\Big[E_\gamma(P_X||Q_X)f''(\gamma)+\gamma^{-3}E_\gamma(Q_X||P_X)f''(\gamma^{-1})\Big]d\gamma. 
    \end{equation}
\end{corollary}
\begin{proof}
    Applying the elementary divergence property \eqref{eq:fdivRepWithE_gamma} of $E_\gamma$-divergence to $P_X\mathsf K$ and $Q_X\mathsf K$, we obtain,
    \begin{align}
        D_f(P_X\mathsf K||Q_X\mathsf K) &= \int_1^\infty [E_\gamma(P_X\mathsf K||Q_X\mathsf K)f''(\gamma)+\gamma^{-3}E_\gamma(Q_X\mathsf K||P_X\mathsf K)f''(\gamma^{-1})]\,d\gamma \\[.5em]
        &\stackrel{(a)}{\leq} \int_1^\infty \Psi_\gamma(e^\varepsilon,c)\Big[E_\gamma(P_X||Q_X)f''(\gamma)+\gamma^{-3}E_\gamma(Q_X||P_X)f''(\gamma^{-1})\Big]d\gamma \\[.5em]
        &\stackrel{(b)}{\leq} \int_1^{e^\varepsilon d_c +1} \Psi_\gamma(e^\varepsilon,c)\Big[E_\gamma(P_X||Q_X)f''(\gamma)+\gamma^{-3}E_\gamma(Q_X||P_X)f''(\gamma^{-1})\Big]d\gamma,
    \end{align}
    where $(a)$ follows because the bound on $\eta_\gamma(\mathsf K)$ in  Theorem \ref{thm:E_gamma_contraction} implies that, $E_\gamma(P_X\mathsf K||Q_X\mathsf K) \leq \Psi_\gamma(e^\varepsilon,c)E_\gamma(P_X||Q_X)$ and $E_\gamma(Q_X\mathsf K||P_X\mathsf K) \leq \Psi_\gamma(e^\varepsilon,c)E_\gamma(Q_X||P_X)$ for any $P_X,Q_X$, and $(b)$ follows from Corollary \ref{corr:restrictedEgammacontraction}. 
\end{proof}

This bound is numerically computable. Further, it provides a non-trivial non-linear SDPI on any $f$-divergence if we relax $\Psi_\gamma(e^\varepsilon,c)$ to $\eta_\text{TV}(\varepsilon)$ by noticing that,
\begin{align}
    D_f(P_X\mathsf K||Q_X\mathsf K) &\leq \eta_\text{TV}(\varepsilon) \int_1^{e^\varepsilon d_c+1} E_\gamma(P_X||Q_X)f''(\gamma)+\gamma^{-3}E_\gamma(Q_X||P_X)f''(\gamma^{-1}) \,d\gamma \\[.5em] 
     = \eta_\text{TV}(\varepsilon)&\left[D_f(P_X||Q_X) - \int_{e^\varepsilon d_c+1}^\infty E_\gamma(P_X||Q_X)f''(\gamma)+\gamma^{-3}E_\gamma(Q_X||P_X)f''(\gamma^{-1}) \,d\gamma\right]. \label{eq:SDPIrestrictedDiffInt}
\end{align}
The formulation \eqref{eq:SDPIrestrictedDiffInt} shows an interesting behavior of the nonlinear SDPI in Corollary \ref{cor:SDPIEgamma_restricted}: The \say{improvement} over the trivial bound $D_f(P_X\mathsf K||Q_X\mathsf K)\leq \eta_\text{TV}(\varepsilon) D_f(P_X||Q_X)$ is determined by the value integral term in \eqref{eq:SDPIrestrictedDiffInt}. Problem specific bounds on $E_\gamma(P_X||Q_X)$ and $E_\gamma(Q_X||P_X)$ can be used to obtain sharper non-linear bounds on arbitrary $f$-divergences. We illustrate the bounds that can be numerically obtained by this technique in the following example. 
\begin{example}
\label{ex:IntSDPIrr}
Let $\mathsf K_3$ and $\mathsf K_4$ be defined as in Example \ref{ex:RR}.
Figure \ref{fig:intbounds} shows the SDPI in \eqref{eq:SDPIInt}, with the integral term evaluated numerically for the two kernels defined above. It becomes clear that the bounds with this technique can yield significant improvements over the standard data processing inequality (shown as the dashed line in Figure \ref{fig:intbounds}). 
\end{example}

\begin{figure}
    \centering
    \begin{subfigure}[b]{.49\linewidth}
    \centering
    \label{fig:intbounds:subfig:rr}
        \includegraphics[scale=0.55]{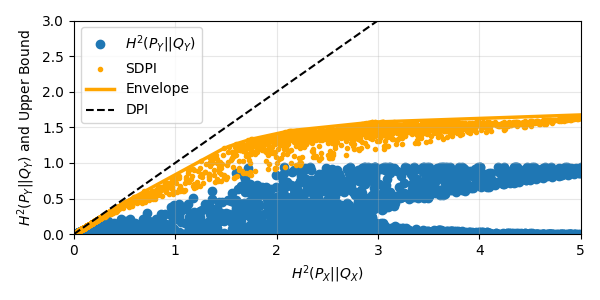}
        \caption{$\mathsf K_3$, $|\mathcal X|=10$, $c=0.05$.}
    \end{subfigure}
    \begin{subfigure}[b]{.49\linewidth}
    \centering
    \label{fig:intbounds:subfig:singular}
        \includegraphics[scale=0.55]{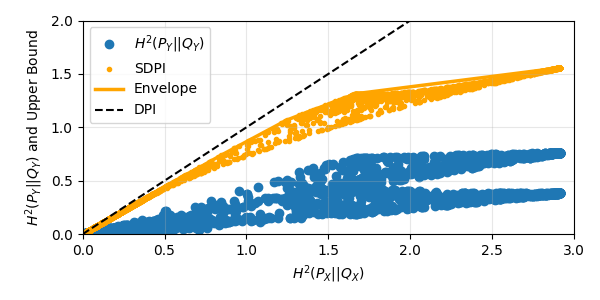}
        \caption{$\mathsf K_4$, $|\mathcal X|=5$, $c=0.1$.}
    \end{subfigure}
    \caption{Numerical evaluation of the bounds presented in Corollary \ref{cor:SDPIEgamma_restricted} for randomly generated distributions in the set $\mathcal P_c$ for the squared Hellinger divergence. $\mathsf K_3$ and $\mathsf K_4$ are defined in Example \ref{ex:RR}. The solid orange line shows the concave upper envelope of the points generated by evaluating the integral \eqref{eq:SDPIInt} for a specific pair of distributions $P_X$ and $Q_X$. The dashed line denotes the standard data processing inequality, that is, $H^2(P_X\mathsf K||Q_X\mathsf K)= H^2(P_X||Q_X)$.}
    \label{fig:intbounds}
\end{figure}

%% file: sections/conclusions.tex
\label{sec:conclusion}
We have shown that $c$-interior PML poses a versatile tool for contraction analyses and disclosure control. The strong data processing inequalities provided in this work can be used to reason about the contraction behavior of classes of kernels beyond those covered by pure LDP. Further, they are often tight and improve over existing results derived in the LDP framework. By devising a theory of contraction on the $c$-interior, we are able to reason about the disclosure risks of a broad class of private systems, including deterministic but reductive mappings. The disclosure risk evaluated in this way is operationally meaningful, and provides the analyst with \emph{concrete} knowledge about the assumptions on considered adversaries. Further, it enables system designers to flexibly incorporate assumptions about the data-generating distributions, and control the inference risk associated with system outputs. The application of the presented theory to hypothesis testing and mean estimation shows that this increased flexibility enables a better and more realistic evaluation of disclosure, and how it can be linked to assumptions about the data-generating distributions. Overall, the results show that PML as an analysis and design tool can be advantageous whenever LDP becomes too restrictive or too rigid.

%% file: sections/appendix.tex
\section{Proof of Lemma \ref{lemma:pml_constraints}}
\label{app:proof_pml_contraints}
First, observe that each $P_X \in \mathcal P_c$ can be expressed as
\begin{equation}
    P_X (dx) = c \mu_X(dx) + d_c \cdot Q_X (dx),
\end{equation}
where $Q_X \in \mathcal P(\mathcal X)$ and $d_c \coloneqq 1-c\mu_X(\mathcal X)$. Hence, \eqref{eq:c_setPML} can be equivalently written as 
\begin{equation}
    \mathsf K(B| x) \leq e^\varepsilon \Big(c\,(\mu_X \mathsf K)(B)+d_c \, (Q_X \mathsf K)(B)\Big),
\end{equation}
for all $x \in \cX$, $B \in \Sigma_\cY$, and $Q_X \in \mathcal P(\mathcal X)$. To obtain the tightest possible bound, we take the infimum of the right-hand side over all admissible $Q_X$. We argue that
\begin{equation}
    \inf_{Q_X : Q_X \ll \mu_X} (Q_X \mathsf K)(B) = \inf_{x \in \cX} \,  \mathsf K(B|x), \quad B \in \Sigma_\cY.
\end{equation}
To see this, note that the inequality 
\begin{equation*}
    \inf_{Q_X : Q_X \ll \mu_X} (Q_X\mathsf K)(B) \geq \inf_{x \in \cX} \,  \mathsf K(B|x), 
\end{equation*}
holds trivially for all $B \in \Sigma_\cY$. To show the opposite direction, fix a set $B \in \Sigma_\cY$ and a small $\zeta >0$, and let 
\begin{equation*}
    A_\zeta = \{x \in \cX : \mathsf K(B|x)\leq \inf_{x'} \,\mathsf K(B|x') + \zeta \}.
\end{equation*}
By the assumption about non-existing maximizing null-sets (see Section \ref{ssec:notation}), $\mu_X(A_\zeta)>0$, so  
\begin{equation*}
    Q_{\zeta} (dx) = \frac{\ind_{A_\zeta} (x) \cdot Q_X(dx)}{Q_X (A_\zeta)}, 
\end{equation*}
is well defined and satisfies $Q_\zeta \ll \mu_X$. ($Q_{\zeta}$ is the conditional distribution of $Q_X$ on the set $A_\zeta$.) We write 
\begin{align*}
     \inf_{Q_X : Q_X \ll \mu_X} (Q_X\mathsf K)(B) &\leq (Q_{\zeta}\mathsf K)(B)\\
     &= \frac{1}{Q_X (A_\zeta)} \int_{A_\zeta} \mathsf K(B|x) \, Q_X(dx)\\
     &\leq \frac{1}{Q_X (A_\zeta)} \int_{A_\zeta} \Big( \inf_{x'}\mathsf K(B|x') + \zeta \Big) \, Q_X(dx)\\
     &= \inf_{x'}\mathsf K(B|x') + \zeta. 
\end{align*}
Then, taking $\zeta \to 0$ yields 
\begin{equation*}
    \inf_{Q_X : Q_X \ll \mu_X} (Q_X \mathsf K)(B) \leq \inf_{x' \in \cX} \,  \mathsf K(B|x'). 
\end{equation*}
Thus, we have established that 
\begin{equation}
    \mathsf K(B|x) \leq e^\varepsilon \Big(c(\mu_X \mathsf K)(B)+ d_c \cdot \inf_{x'} \,\mathsf K(B|x') \Big), \quad x \in \cX, B \in \Sigma_{\cY},
\end{equation}
as desired. \qed

\section{Proof of Lemma \ref{lem:PMLimpliesLDP}}
\label{app:proofLDPlemma}
    Fix $B\in\Sigma_\mathcal Y$ and $x,x'\in\mathcal X$ arbitrary. From Lemma \ref{lemma:pml_constraints}, we can directly see that,
    \begin{equation}
        \mathsf K(B|x) \leq e^\varepsilon \mathsf M(B|x'),
    \end{equation}
    which implies that,
    \begin{equation}
    \label{eq:Mdef1}
        \mathsf M(B|x) \leq c(\mu_X \mathsf K)(B) + e^\varepsilon d_c \mathsf M(B|x').
    \end{equation}
    Now, by definition of $\mathsf M$,
    \begin{equation}
    \label{eq:Mdef2}
        c(\mu_X\mathsf K)(B) = \mathsf M(B|x') - d_c \mathsf K(B|x') \leq \mathsf M(B|x'),
    \end{equation}
    where the inequality follows since $d_c\mathsf K(B|x')\geq 0$. Using \eqref{eq:Mdef2} to bound $c(\mu_X\mathsf K)(B)$ in \eqref{eq:Mdef1}, we obtain,
    \begin{equation}
        \mathsf M(B|x) \leq \mathsf M(B|x') + e^\varepsilon d_c \mathsf M(B|x') = (e^\varepsilon d_c + 1)\mathsf M(B|x').
    \end{equation}
    Since the choice of $x$ and $x'$ was arbitrary, this shows that $\mathsf M$ constructed in this way satisfies $(e^\varepsilon d_c +1)$-LDP, given that $\mathsf K$ satisfies $\varepsilon$-PML on $\mathcal P_c$. In other words, $\mathsf K \in\mathcal M(\varepsilon,c) \implies \mathsf M \in \mathcal M(e^\varepsilon d_c +1,0)$. 

    To proof the second part of the statement, notice that for any $B\in\Sigma_\mathcal Y$,
    \begin{equation}
    \label{eq:Kintbound}
        (\mu_X \mathsf K)(B) = \int_{\mathcal X} \mathsf K(B|x)\mu_X(dx) \leq \sup_x \mathsf K(B|x)\int_\mathcal X \mu_X(dx) = \sup_x\mathsf K(B|x) \cdot \mu_X(\mathcal X). 
    \end{equation}
    Plugging this into Lemma \ref{lemma:pml_constraints}, we get
    \begin{align}
        \sup_x\mathsf K(B|x) &\stackrel{\text{(a)}}{\leq} e^\varepsilon c(\mu_X\mathsf K)(B)+ e^\varepsilon d_c \inf_{x'}\mathsf K(B|x') \\
        &\stackrel{(b)}{\leq} e^\varepsilon c \mu_X(\mathcal X) \sup_x\mathsf K(B|x) + e^\varepsilon d_c \inf_{x'} \mathsf K(B|x'),\label{eq:Kconst1}
    \end{align}
    where $(a)$ is the statement of Lemma \ref{lemma:pml_constraints} and $(b)$ follows from \eqref{eq:Kintbound}. We can rearrange \eqref{eq:Kconst1} to see that, whenever $1-e^\varepsilon c \mu_X(\mathcal X)>0 \iff \varepsilon < -\log(c\mu_X(\mathcal X))$, we have,
    \begin{equation}
    \label{eq:Kisleqzeroatgamma*}
        \sup_x \mathsf K(B|x) - \frac{e^\varepsilon d_c}{1-e^\varepsilon c \mu_X(\mathcal X)}\inf_{x'}\mathsf K(B|x') \leq 0 \quad \forall B\in\Sigma_\mathcal Y.
    \end{equation}
    Let $\gamma^* \coloneqq (e^\varepsilon d_c) / (1-e^\varepsilon c\mu_X(\mathcal X))$. We have,
    \begin{equation}
    \label{eq:gamma*contraction}
        \eta_{\gamma^*}(\mathsf K) \stackrel{(a)}{\leq} \sup_{x\neq x'} \sup_{B\in\Sigma_\mathcal Y} \left\{\mathsf K(B|x)-\gamma^* \mathsf K(B|x')\right\}  \stackrel{(b)}{\leq} \sup_{B\in\Sigma_\mathcal Y} \left\{\sup_{x}\mathsf K(B|x)-\gamma^* \inf_{x'} \mathsf K(B|x')\right\} \leq 0. 
    \end{equation}
    Here $(a)$ follows form the characterization of $\eta_\gamma$ in \cite[Theorem 2]{asoodeh2020contraction}, $(b)$ follows by splitting the joind maximum over $x\neq x'$ into independent maximizations, and $c$ follows from \eqref{eq:Kisleqzeroatgamma*}. Due to the data processing inequality, we also have $\eta_{\gamma^*}(\mathsf K)\geq 0$. Hence, \eqref{eq:gamma*contraction} shows that $\eta_{\gamma^*}(\mathsf K)=0$. This is shown in  \cite{9517999} to be equivalent to $\mathsf K$ satisfying $\log \gamma^*$-LDP, which is what we wanted to show.

\section{Omitted Derivations in Theorem \ref{thm:TVcontraction}}
\label{app:mechanismrandomization}
\subsection*{Privacy Guarantee of the Construction}
To simplify notation, let $c\mu_X(\mathcal A) \eqqcolon a$ and as usual $d_c = (1-c\mu_X(\mathcal X))$. We have,
\begin{equation}
    \ell_{P_X\mathsf K^\star_{\varepsilon,c}}(X\to 0) = \log \frac{\max_{x\in\mathcal X}\mathsf K^\star_{\varepsilon,c}(0|x)}{\int_\mathcal X \mathsf K^\star_{\varepsilon,c}(0|dx)P_X(dx)}.
\end{equation}
We begin by bounding the denominator from below. Recall that, using Lemma \ref{lemma:pml_constraints}, we have, 
\begin{align}
    (1+e^\varepsilon d_c)\int_\mathcal X \mathsf K^\star_{\varepsilon,c}(0|dx)P_X(dx) & \geq (1+e^\varepsilon d_c)\left[c\int_\mathcal X \mathsf K^\star_{\varepsilon,c}(0|dx)\mu_X(dx) + d_c \min_{x\in\mathcal X}\mathsf K^\star_{\varepsilon,c}(0|x)\right] \\[.5em]
    &= c(\mu_X(\mathcal A)e^\varepsilon(1-a)+(\mu_X(\mathcal X)-\mu_X(\mathcal A))(1-e^\varepsilon a))+d_c(1-e^\varepsilon a) \\[.5em]
    &= a(e^\varepsilon-1)+(1-e^\varepsilon a) = 1-c\mu_X(\mathcal A).
\end{align}
The numerator is clearly bounded by $e^\varepsilon(1-c\mu_X(\mathcal A))/(1+e^\varepsilon d_c)$, so we have for any $P_X \in\mathcal P_c$,
\begin{equation}
     \ell_{P_X\mathsf K^\star_{\varepsilon,c}}(X\to 0) \leq \log \frac{e^\varepsilon(1-c\mu_X(\mathcal A))}{1-c\mu_X(\mathcal A)} = \varepsilon.
\end{equation}
To show that $\ell(X\to 1)$ satisfies the same upper bound, note that we may choose $\tilde {\mathcal A}=\mathcal A^c$ to define a new mechanism $\bar {\mathsf K}^\star_{\varepsilon,c}$, which yields,
\begin{equation}
    \bar{\mathsf K}^\star_{\varepsilon,c}(0|x) = \mathsf K^\star_{\varepsilon,c}(1|x).
\end{equation}
Hence, for any $P_X\in\mathcal P_c$,
\begin{equation}
    \ell_{P_X\mathsf K^\star_{\varepsilon,c}}(X\to1) = \ell_{P_X\bar{\mathsf K}^\star_{\varepsilon,c}}(X\to 0) \leq \varepsilon.
\end{equation}

\subsection*{Optimal Mechanism for Odd Atomic Alphabets}
In a setting where there is no $\mathcal A\in\Sigma_\mathcal X$ such that $\mu_X(\mathcal A)=\nicefrac{1}{2}$, we may proceed as follows. First, note that such a case only occurs when the measure has atoms, and we can always decompose $\mu_X$ into nonatomic part $\mu_X^{\text{n-a}}$ and atomic part $\mu_X^{\text{a}}$ \cite{johnson1970atomic}. Let $\mathcal T$ be the largest set of atoms together with the non-atomic part of $\mathcal X$ for which,
\begin{equation}
    \mu_X(\mathcal T) < \frac{\mu_X(\mathcal X)}{2}.  
\end{equation}
Let $x^* \in\Sigma_\mathcal X$ be an atom with $\mu_X(\{x^*\})=p^*$ such that,
\begin{equation}
    \mu_X(\mathcal T \cup \{x^*\}) > \frac{\mu_X(\mathcal X)}{2}.
\end{equation}
Define $\tau = \nicefrac{\mu_X(\mathcal X)}{2}$ and,
\begin{equation}
    \lambda = \frac{\tau-\mu_X(\mathcal T)}{p^*} \in [0,1].
\end{equation}
Define the mechanism $\mathsf K^\star$ whenever $\varepsilon<\log(2/(c\mu_X(\mathcal X)))$ as,
\begin{equation}
        \mathsf K_{\varepsilon,c}^\star(0|x) = \frac{1}{1+e^\varepsilon d_c}\begin{cases}
            e^\varepsilon(1-c\tau), &\text{if }x\in \mathcal T,\\ 
            1+\lambda(e^\varepsilon-1)-e^\varepsilon c\tau, &\text{if }x=x^*\\
            1-e^\varepsilon c\tau, &\text{if }x\in\mathcal T^c\setminus \{x^*\}, 
        \end{cases}\
    \end{equation}
and $\mathsf K^\star_{\varepsilon,c}(1|x) = 1- \mathsf K_{\varepsilon,c}^\star(0|x)$. And similarly, whenever $\varepsilon>\log(2/(c\mu_X(\mathcal X)))$, 
\begin{equation}
    \mathsf K^\star_{\varepsilon,c}(0|x) = \begin{cases}
        1,&\text{if }x\in \mathcal T,\\
        \lambda, &\text{if }x=x^*,\\
        0, &\text{otherwise,}
    \end{cases}
\end{equation}
and $\mathsf K^\star_{\varepsilon,c}(1|x) = 1- \mathsf K_{\varepsilon,c}^\star(0|x)$.  This can be seen as a probabilistic assignment of $x^*$ to either $\mathcal T$ or $\mathcal T^c$, which yields the required condition.

\section{Proof of Lemma \ref{lem:QcdontchangeTVcontr}}
\label{app:QdontchangeTVcontrproof}
    From the extreme-point representation of the set $\mathcal P_c$,
        \begin{equation}
            \mathcal P_c =\{c + (1-c\mu_X(\mathcal X))V_X \mid V_X \in \mathcal P(\mathcal X)\},
        \end{equation}
        where $c + V_X$ is understood as adding a constant function $f(x) = c$ to the density admitted by $V_X$ with respect to $\mu_X$. We have that 
        \begin{align}
            \eta^{\mathcal P_c}_{\text{TV}}(\mathsf K) &= \sup_{P_X,Q_X \in\mathcal P_c} \frac{|| P_X \mathsf K - Q_X \mathsf K ||_1}{||P_X-Q_X||_1} \\[.7em]
            &= \sup_{V_X,W_X \in \mathcal P(\mathcal X)}\frac{||(c + (1-c\mu_X(\mathcal X))V_X \mathsf K) - (c + (1-c\mu_X(\mathcal X))W_X\mathsf K) ||_1}{||c + (1-c\mu_X(\mathcal X))V_X-c - (1-c\mu_X(\mathcal X))W_X||_1}.
        \end{align}
        Now, note that, 
        \begin{equation}
            ||c + (1-c\mu_X(\mathcal X))V_X-c - (1-c\mu_X(\mathcal X))W_X||_1 = (1-c\mu_X(\mathcal X))||V_X-W_X||_1, 
        \end{equation}
        and,
        \begin{equation}
            || (c + (1-c\mu_X(\mathcal X))V_X\mathsf K) - (c + (1-c\mu_X(\mathcal X))W_X\mathsf K) ||_1 = (1-c\mu_X(\mathcal X))||V_X\mathsf K - W_X\mathsf K||_1.
        \end{equation}
        Therefore, we may write,
        \begin{align}
            \eta^{ \mathcal P_c}_{\text{TV}}(K) &= \sup_{V_X,W_X \in \mathcal P(\mathcal X)} \frac{(1-c\mu_X(\mathcal X))||V_X\mathsf K- W_X\mathsf K||_1}{(1-c\mu_X(\mathcal X))||V_X - W_X||_1} \\
            &= \sup_{V_X,W_X \in \mathcal P(\mathcal X)} \frac{||V_X\mathsf K- W_X\mathsf K||_1}{||V_X - W_X||_1} = \eta_{\text{TV}}(\mathsf K).
        \end{align}
        Hence, the restriction from $\mathcal P(\mathcal X)$ to $\mathcal P_c(\mathcal X)$ does not affect $\eta_\text{TV}$. 

\section{Proof of Theorem \ref{thm:E_gamma_contraction}}
\label{app:Egammaproof}
 We exclude the case $\varepsilon=0$, as this implies that $P_X\mathsf K$ and $Q_X\mathsf K$ are independent, which implies that $\eta_\gamma(\mathsf K)=0$. So, suppose $\varepsilon>0$ below. It is shown in \cite[Theorem 2]{asoodeh2020contraction} that, 
    \begin{equation}
        \eta_\gamma(\mathsf K) \leq \sup_{x,x'} E_\gamma(\mathsf K(\cdot|x)||\mathsf K(\cdot |x')) =\sup_{x,x'} \sup_{B\in\Sigma_\mathcal Y} \Big(\mathsf K(B|x) - \gamma \, \mathsf K(B|x')\Big).
    \end{equation}
    We proceed along the same lines as in the proof of Theorem \ref{thm:TVcontraction}, that is, we fix any arbitrary $B \in \Sigma_\mathcal Y$ and bound $\sup_{x\in\mathcal X}\mathsf K(B|x) - \inf_{x\in\mathcal X}\gamma\mathsf K(B|x')$. To this end, fix any $B \in \Sigma_\mathcal Y$ and let $p \coloneqq \sup_x\mathsf K(B|x)$, $q\coloneqq \inf_x\mathsf K(B|x')$, $I =(\mu_X \mathsf K)(B)$. Using Lemma \ref{lemma:pml_constraints} and the the resulting bounds on the kernel in \eqref{eq:Zpmlconstraintz0} \& \eqref{eq:Zpmlconstraintz1}, finding an upper bound on $\mathsf K(B|x)-\gamma \mathsf K(B|x')$ amounts to finding an upper bound on the optimal value of the linear program,
    \begin{align}
    \max_{p,q,I} &\big\{p-\gamma q\big\}, \\
    \text{s.t. } &p \leq e^\varepsilon c I + e^\varepsilon d_c q, \\
    & 1-q \leq e^\varepsilon c(\mu_X(\mathcal X)-I) + e^\varepsilon d_c(1-p), \\
    & 0\leq p,q \leq 1, \quad \mu_X(\mathcal X)q \leq I \leq \mu_X(\mathcal X)p.
\end{align}
In standard form, with $x = [p,q,I]^\top$, $h = [1, -\gamma,0]^\top$, this can be written as
\begin{align}
    \max & \;h^\top x \\
    \text{s.t. }&Ax\leq b,
\end{align}
with the constraint parameters defined by,
\begin{equation}
    A = \begin{bmatrix}
    1 & e^\varepsilon d_c & 1 & -1 & 0 & 0 & 0 & 0 & -\mu_X(\mathcal X)& 0\\
    -e^\varepsilon d_c & -1 & 0 & 0 & 1 & -1 & 0 & 0 & 0 & \mu_X(\mathcal X)\\
    -e^\varepsilon c & e^\varepsilon c & 0 & 0 & 0 & 0 & 1 & -1 & 1 & -1
\end{bmatrix}^\top, 
\end{equation}
and,
\begin{equation}
    b = \begin{bmatrix}
        0, & e^\varepsilon-1, & 1, &0, & 1, & 0, & \mu_X(\mathcal X), & 0, & 0, & 0
    \end{bmatrix}^\top.
\end{equation}
The dual of this LP is given by $\min_\lambda b^\top \lambda \text{ s.t. }A^T\lambda = h$ \cite{boyd2004convex}. By weak duality, any feasible solution $\lambda \geq 0$ provides an upper bound on the primal value, that is, 
\begin{equation}
    \max_{p,q,I} p-\gamma q \leq b^\top \lambda.
\end{equation}
Some algebra reveals that the dual problem is,
\begin{align}
    \min_{\lambda_i \geq 0\,\forall i\in[10]} &\Big\{\lambda_2(e^\varepsilon-1)+\lambda_3+\lambda_5+\lambda_7 \mu_X(\mathcal X)\Big\}. \\[.5em]
    \text{s.t. }\; &\lambda_1 + \lambda_2 e^\varepsilon (1-c\mu_X(\mathcal X))+\lambda_3-\lambda_4-\lambda_9\mu_X(\mathcal X)=1\\
    & \lambda_1 e^\varepsilon(1-c\mu_X(\mathcal X)) + \lambda_2 - \lambda_5 + \lambda_6 -\lambda_{10}\mu_X(\mathcal X) = \gamma \\
    &(\lambda_2-\lambda_1)e^\varepsilon c -\lambda_7 + \lambda_8 +\lambda_9 - \lambda_{10}= 0.
\end{align}
The following choices of $\lambda$ give us the desired upper bound: For simpler notation, define,
\begin{equation}
    \gamma^* = \frac{e^\varepsilon d_c}{1-e^\varepsilon c \mu_X(\mathcal X)}.
\end{equation}
\begin{enumerate}
    \item $\gamma < \gamma^*$: Pick $\lambda_3^\star=\lambda_4^\star=\lambda_5^\star=\lambda_6^\star =\lambda_7^\star=\lambda_8^\star=\lambda_{10}^\star=0$. The three equality constraints in the dual linear program uniquely determine the remaining three dual variables as
\begin{equation}
    \lambda_1^\star(\gamma) = \frac{e^\varepsilon\gamma-1}{(e^\varepsilon-1)(e^\varepsilon d_c+1)},\quad \lambda_2^\star(\gamma) = \frac{e^\varepsilon d_c -\gamma(1-e^\varepsilon c \mu_X(\mathcal X))}{(e^\varepsilon-1)(e^\varepsilon d_c +1)}, \quad \lambda_9^\star(\gamma) = \frac{e^\varepsilon c (\gamma-1)}{e^\varepsilon-1}.
\end{equation}
Clearly, $\lambda_1^\star\geq 0$ and $\lambda_9^\star\geq 0$ for all $\varepsilon,c\geq 0$ and $\gamma \geq 1$. Further, we have $\lambda_2^\star(\gamma) \geq 0$ as long as $\gamma \leq e^\varepsilon d_c/(1-e^\varepsilon c \mu_X(\mathcal X))$. Hence, the choice yields a feasible solution to the dual for this range of values. The dual objective for this choice of dual variables $\lambda$ will hence be an upper bound to the primal value. The dual objective at these values is 
\begin{equation}
\label{eq:optgammadep}
    \lambda_2^\star(\gamma)(e^\varepsilon-1) = \frac{e^\varepsilon d_c - \gamma(1-e^\varepsilon c \mu_X(\mathcal X))}{1+e^\varepsilon d_c}.
\end{equation}

\item $\gamma\geq\gamma^*$: Pick $\lambda_1^\star,\lambda_6^\star,\lambda_9^\star \neq 0$, and set all other dual variables to zero. We obtain the feasible solution,
\begin{equation}
    \lambda_1^\star(\gamma) = \lambda_1^\star(\gamma^*), \quad \lambda_9^\star(\gamma) = \lambda_9^\star(\gamma^*),  \quad \lambda_6^\star(\gamma) = \gamma - \gamma^*.
\end{equation}
Using the fact that $\gamma^*\leq \gamma$ in the regime of interest, we see that $\lambda_6^\star(\gamma)\geq 0$. The dual objective for this choice is zero,
which continuously connects to the solution above for $\gamma \leq\gamma^*$. 
The above choices provide a feasible solution to the dual problem and hence an upper bound to the primal problem whenever $\gamma^*>0$. This is the case as long as $\varepsilon <-\log(c\mu_X(\mathcal X))$.
\end{enumerate}
For larger values of $\varepsilon$, note that we have $\eta_\gamma(\mathsf K)\leq \eta_\text{TV}(\mathsf K)$ whenever $\gamma \geq 1$, so we can use the bound in Theorem \ref{thm:TVcontraction} to find,
\begin{equation}
    \eta_\gamma(\mathsf K) \leq \min\left\{1,\, \frac{e^\varepsilon-1}{e^\varepsilon d_c+1}\right\}.
\end{equation}
Putting all the above cases together, we find that,
\begin{equation}
    \eta_\gamma(\mathsf K) \leq \left(\frac{e^\varepsilon d_c - \gamma(1-e^\varepsilon c \mu_X(\mathcal X))}{1+e^\varepsilon d_c}\right)_+, \quad \text{if }0\leq \varepsilon <-\log(c\mu_X(\mathcal X)),
\end{equation}
\begin{equation}
    \eta_\gamma(\mathsf K) \leq \eta_\text{TV}(\varepsilon), \quad \text{otherwise,}
\end{equation}
which is what we wanted to show.

\section{Proof of Lemma \ref{lem:gammamingammamaxbound}}
\label{app:lemgammamingammamaxbound}
To show the statement, note that for any $\mathsf K$,
\begin{equation}
\label{eq:binetteproof:Gammarelax}
     \sup_{P_X,Q_X\in\mathcal P_c} \sup_{B\in\Sigma_\mathcal Y} \frac{(P_X\mathsf K)(B)}{(Q_X\mathsf K)(B)}\leq \sup_{B\in\Sigma_\mathcal Y}\frac{\sup_{P_X \in\mathcal P_c} (P_X \mathsf K)(B)}{\inf_{Q_X\in\mathcal P_c}(Q_X\mathsf K)(B)}.
\end{equation}
The proof of Lemma \ref{lemma:pml_constraints} shows that,
\begin{equation}
\label{eq:binetteproof:supset}
    \sup_{P_X\in\mathcal P_c} (P_X\mathsf K)(B) = c(\mu_X \mathsf K)(B) + d_c \sup_{x\in\mathcal X}\mathsf K(B|x),
\end{equation}
\begin{equation}
\label{eq:binetteproof:infset}
    \inf_{Q_X\in\mathcal P_c} (Q_X\mathsf K)(B) = c(\mu_X \mathsf K)(B) + d_c \inf_{x\in\mathcal X}\mathsf K(B|x).
\end{equation}
Further, we can use Lemma \ref{lemma:pml_constraints} to get,
\begin{equation}
\label{eq:binetteproof:pmlsupbound}
    \sup_{x\in\mathcal X}\mathsf K(B|x) \leq e^\varepsilon c (\mu_X\mathsf K)(B) + e^\varepsilon d_c \inf_{x\in\mathcal X}\mathsf K(B|x).
\end{equation}
Using \eqref{eq:binetteproof:supset} and \eqref{eq:binetteproof:infset} to rewrite \eqref{eq:binetteproof:Gammarelax} and then bounding according to \eqref{eq:binetteproof:pmlsupbound} yields,
\begin{align}
    \Gamma_{\max}(\varepsilon,c) &\leq \frac{d_ce^\varepsilon\big(c(\mu_X\mathsf K)(B)+d_c \inf_{x\in\mathcal X}\mathsf K(B|x)\big) +c(\mu_X\mathsf K)(B)}{c(\mu_X\mathsf K)(B)+d_c \inf_{x\in\mathcal X}\mathsf K(B|x)} \\[.5em]
    &= e^\varepsilon d_c + \frac{c(\mu_X\mathsf K)(B)}{c(\mu_X\mathsf K)(B)+d_c \inf_{x\in\mathcal X}\mathsf K(B|x)} 
    \leq e^\varepsilon d_c +1,
\end{align}
where the last inequality follows from $\inf_x\mathsf K(B|x)\geq 0$. The bound on $\Gamma_{\min}(\varepsilon,c)$ follows by the same steps.

\section{Proof of Lemma \ref{lem:TVseparatinglemma}}
\label{app:proofTVseparatinglemma}
 Let $f\coloneqq\frac{\mathrm dP_X}{\mathrm d\mu_X}-\frac{\mathrm dQ_X}{\mathrm d\mu_X}$ and $S\coloneqq\{f\geq0\}$. This implies that,
 \begin{equation}
     \int_S f\, d\mu_X=\int_{S^c}(-f)\, d\mu_X=\text{TV}(P_X\|Q_X).
 \end{equation}
We will need the following claim, which we proof below.
\begin{claim}
\label{claim:toplevelsets}
     Let $T\in\Sigma_\mathcal X$ with $\mu_X(T)>0$, let $g\geq0$ be integrable on $T$, and let $t\in[0,\mu_X(T)]$. There is a measurable function $u:T\to[0,1]$ with $\int_T u\, d\mu_X=t$ and
    \begin{equation}
    \label{eq:toplevel}
        \int_T ug\,d\mu_X\ \geq\ \frac{t}{\mu_X(T)}\int_T g\, d\mu_X.
    \end{equation}
\end{claim}
Given this claim, we proceed as follwos: Since $\mu_X(S)+\mu_X(S^c)=\mu_X(\mathcal X)\leq2M$, at most one of $\mu_X(S),\mu_X(S^c)$ exceeds $M$. We therefore need to treat the following three cases.
    \begin{enumerate}
        \item If both are at most $M$, take $a\coloneqq\mathbf 1_S$ (likelihood-ratio-test). Then, $\int a\,d(P_X-Q_X)=\text{TV}(P_X||Q_X)$.
        \item If $\mu_X(S)>M$, apply Claim \ref{claim:toplevelsets} with $T=S$, $g=f$, $t=M$ and let $a\coloneqq u$ on $S$, $a\coloneqq0$ on $S^c$. Then $\int a\,d\mu_X=M$, $\mu_X(\mathcal X)-\int a\,\mathrm d\mu_X\leq M$, and by \eqref{eq:toplevel},
        \begin{equation}
            \int a\,d(P_X-Q_X)=\int_S uf\,d\mu_X\geq\frac{M}{\mu_X(S)}\,\text{TV}(P_X||Q_X)\geq\frac{M}{\mu_X(\mathcal X)}\,\text{TV}(P_X||Q_X).
        \end{equation}
        \item If $\mu_X(S^c)>M$, apply Claim \ref{claim:toplevelsets} with $T=S^c$, $g=-f$, $t=M$, and set $a\coloneqq1-u$ on $S^c$, $a\coloneqq1$ on $S$. The bounds follow as above, since $\int a\, d(P_X-Q_X)=-\int_{S^c}u\,(-f)\, d\mu_X$.
    \end{enumerate}
    In all cases the statement follows, since we assumed $M/\mu_X(\mathcal X)\geq\tfrac12$.

\subsection*{Proof of Claim \ref{claim:toplevelsets}}
To see this, note that $s\mapsto\mu_X(T\cap\{g>s\})$ is non-increasing and right-continuous, so for $\tau\coloneqq\inf\{s\geq0:\mu_X(T\cap\{g>s\})\leq t\}$ we get,
    \begin{equation}
        \mu_X(T\cap\{g>\tau\}) \leq t \leq \mu_X(T\cap\{g\geq\tau\}).
    \end{equation}
    Let, 
    \begin{equation}
        \lambda\coloneqq\frac{t-\mu_X(T\cap\{g>\tau\})}{\mu_X(T\cap\{g=\tau\})}\in[0,1],
    \end{equation}
    if the denominator is positive, and $\lambda\coloneqq0$ otherwise, and set $u\coloneqq\mathbf 1_{\{g>\tau\}}+\lambda\mathbf 1_{\{g=\tau\}}$ on $T$. Then $\int_T u\,\mathrm d\mu_X=t$. With $\lambda_0\coloneqq t/\mu_X(T)\in[0,1]$ and $h\coloneqq u-\lambda_0$ we have $h\geq0$ on $\{g>\tau\}$, $h\leq0$ on $\{g<\tau\}$, and $g-\tau=0$ on $\{g=\tau\}$, so $h\cdot(g-\tau)\geq0$ pointwise on $T$. Since $\int_T h\,\mathrm d\mu_X=0$, we get \eqref{eq:toplevel}, because,
    \begin{equation}
        \int_T h\,g\,\mathrm d\mu_X\geq\tau\int_T h\,\mathrm d\mu_X=0.
    \end{equation}

\section{Proof of Lemma \ref{lem:etaTVsq=Xi}}
\label{app:etaTVsq=Xi}
    Both $\eta_\text{TV}^2(\varepsilon)$ and $\Xi(\varepsilon,c)$ are positive and bounded on $\varepsilon\in (0, \varepsilon^*)$. To see this, note that for any $c$ in the prescribed range, $\varepsilon$ there is a finite value of $\varepsilon<-\log(c\mu_X(\mathcal X))$ such that $\Xi(\varepsilon,c) = e^{2\varepsilon}/(e^\varepsilon d_c +1) <\infty$, so $\Xi(\varepsilon,c)$ will always be finite. For very small $\varepsilon$, we have $\eta_\text{TV}(\varepsilon) = (e^\varepsilon-1)/(e^\varepsilon d_c+1)$ and $\Xi(\varepsilon,c) = (e^\varepsilon-1)^2/(e^\varepsilon d_c(1-e^\varepsilon c\mu_X(\mathcal X)))$. Hence,
    \begin{align}
        \lim_{\varepsilon\to 0^+}\frac{\Xi(\varepsilon,c)}{\eta_\text{TV}^2(\varepsilon)} &= \lim_{\varepsilon\to 0^+}\left\{\frac{(e^\varepsilon-1)^2}{e^\varepsilon d_c(1-e^\varepsilon c\mu_X(\mathcal X))}\cdot \frac{(e^\varepsilon d_c+1)^2}{(e^\varepsilon-1)^2}\right\} \\[.5em]
        &= \lim_{\varepsilon\to0^+} \frac{(e^\varepsilon d_c+1)^2}{e^\varepsilon d_c(1-e^\varepsilon c\mu_X(\mathcal X))} \\[.5em]
        &= \lim_{\varepsilon\to 0^+} \frac{(d_c+1)^2}{d_c^2} <\infty,
    \end{align}
    since $d_c\in(0,1)$ is a strictly positive constant bounded from above by one.
    Define the ratio function,
    \begin{equation}
         R(\varepsilon) \coloneqq \frac{\Xi(\varepsilon,c)}{\eta_\text{TV}^2(\varepsilon)}.
    \end{equation}
    By the above, we see that $R(\varepsilon)$ is positive bounded and continuous on $(0,\varepsilon^*)$. Define the function $R: [0,\varepsilon^*] \to \mathbb R$ by continous extension. This function is a positive continuous function on a compact interval, and hence attains a finite maximum and minimum. Let therefore,
    \begin{equation}
        a(c) \coloneqq \min_{\varepsilon\in[0,\varepsilon^*]}R(\varepsilon) \leq \max_{\varepsilon'\in[0,\varepsilon^*]}R(\varepsilon) \eqqcolon A(c).
    \end{equation}
    Since by we have $0<a(c)\leq A(c) <\infty$, this finishes the proof.

\section{Equivalence of the Hellinger and total variation metrics on discrete spaces}
\label{app:lem:TV^2_H^2}
We show that on $\mathcal P_c$, the total variation distance and the Hellinger distance are equivalent metrics. Consider the following lemma. 

\begin{lemma}
\label{lem:TV^2_H^2}
     For any discrete distributions $P_X, Q_X \in\mathcal P_c$, with slight abuse of notation let $P_X(x),Q_X(x)$ denote the $P_X(\{x\}), Q_X(\{x\})$. We have,
    \begin{equation}
        \text{TV}^2(P_X||Q_X) \geq c \, H^2(P_X||Q_X).
    \end{equation}
\end{lemma}
\begin{proof}
    To see this, observe the following relation. 
    \begin{align}
        H^2(P_X||Q_X) &= \sum_{x\in \mathcal X}\left(\sqrt{P_X(x)}-\sqrt{Q_X(x)}\right)^2 
        = \sum_{x\in \mathcal X} \frac{(P_X(x)-Q_X(x))^2}{(\sqrt{P_X(x)}+\sqrt{Q_X(x)})^2} \\[.5em]
        &\leq \sum_{x\in\mathcal X} \frac{(P_X(dx)-Q_X(dx))^2}{4c} = \frac{||P_X-Q_X||_2^2}{4c}\\[.5em] &\leq \frac{||P_X-Q_X||_1^2}{4c} = \frac{1}{c}\text{TV}^2(P_X||Q_X),
    \end{align}
    where we used the fact that $||\cdot||_2 \leq ||\cdot||_1$ in the second to last inequality.
    \end{proof}
    With the standard inequality $\text{TV}(P_X||Q_X)\leq H(P_X||Q_X)$ \cite[(7.22)]{Polyanskiy_Wu_2025} this implies that, whenever the space $\mathcal X$ is discrete,
    \begin{equation}
        \sqrt{c}H(P_X||Q_X) \leq \text{TV}(P_X||Q_X)\leq H(P_X||Q_X) \quad \text{on } \mathcal P_c(\mathcal X),
    \end{equation}
    or equivalently, $H(P_X||Q_X) \asymp \text{TV}(P_X||Q_X)$ on $\mathcal P_c$. In other words, if a sequence of distributions $\{P^n_X\}_n \subset \mathcal P_c$ converges to a distribution $\bar P_X \in\mathcal P_c$ in terms of total variation at a certain rate, it converges in terms of Hellinger distance at the same rate (up to multiplicative constants).

\section{Proof of Lemma \ref{lem:multidimpacking}}
\label{app:packingproof}
We construct a packing analogous to the one-dimensional packing used in the proof of Proposition \ref{prop:meanestimation}. Note that we will glance over the measure-theoretic details of absolute continuity here, and the reader should take the distributions defined on atoms simply as a shorthand for absolutely continuous distributions with mass uniformely spread over a small interval, analogous to the the constructions in the proof of Proposition \ref{prop:meanestimation}. For each element in $v\in\mathcal \{\pm 1\}^d$, define, $P_v = c\mu_X +  d_c \cdot \bar P_v$, where $\mu_X$ is the Lebesgue measure restricted to the unit ball, and $X\sim\bar P_v$ is constructed in the following way: Let $J\sim \text{Unif}(\{1,\dots,d\})$. Given a realization $J=j$, $X$ is chosen as,
\begin{equation}
    X= \begin{cases}
        e_j &\text{ w.p. }\frac{1+\delta  v_j}{2},\\
        -e_j &\text{ w.p. }\frac{1-\delta v_j}{2}.
    \end{cases}
\end{equation}
$P_v$ is in $\mathcal P_c(\mathcal X)$ for any $v$ by construction. For each $P_v$, we find the mean of $X$ as,
\begin{equation}
   \theta(P_v) = \mathbb E_{X\sim P_v}[X] = c\int_{\mathcal B_d^p}x\, d\mu_X(x) + (1-c\mu_X(\mathcal X))\mathbb E_{X\sim \bar P_v}[X] = 0 + \frac{\delta d_c}{d}v.
\end{equation}
Hence, the squared Euclidean distance between two parameter instances $v\neq v'$ from the packing set is,
\begin{equation}
\label{eq:paramdiff_meanpacking}
    ||\theta(P_v)- \theta(P_{v'})||_2^2 = \frac{\delta^2d_c^2}{d^2}||v-v'||_2^2.
\end{equation}
Since $v,v' \in \{-1,1\}^d$, we have,
\begin{equation}
    ||v-v'||_2^2 = 4\sum_{i=1}^d \mathbf 1\{v\neq v'\} = 4\,d_H(v,v'),
\end{equation}
where $d_H(\cdot,\cdot)$ denotes the Hamming-distance between the two points. 
Plugging this bound into \eqref{eq:paramdiff_meanpacking}, we see that,
\begin{equation}
     ||\theta(P_v)- \theta(P_{v'})||_2 \geq \frac{4\delta^2d_c^2}{d^2}d_H(v,v').
\end{equation}
What remains is to bound the total variation distance between two instances $P_v,P_{v'}$. To this end, notice that we have $\text{TV}(P_v||P_{v'}) = d_c \text{TV}(\bar P_v||\bar P_{v'})$. For each $v\in\mathcal \{\pm 1\}^d$, $\bar P_v$ is a distribution supported on the set $\{\pm e_j\}_{=1}^d$, and since $j$ is sampled uniformly, we have,
\begin{equation}
    \bar P_v(\pm e_j) = \frac{1}{d}\frac{1\pm\delta v_j}{2}.
\end{equation}
Since the densities only differ when $v_j\neq v'_j$, and each differing coordinate contributes $\delta /d$, we get that,
\begin{equation}
    \text{TV}(P_v||P_{v'}) = \frac{d_c}{d}\sum_{j=1}^d \delta\mathbf 1\{v_j\neq v'_j\} = \frac{\delta d_c}{d} d_H(v,v'),
\end{equation}
which is what we wanted to show. \qed

\section{Proof of Proposition \ref{prop:highd-estimator}}
\label{app:highd-estimator-proof}
    The statement follows by appliyng the DJW construction in \cite[(26a)]{duchi2013local} analogous to Proposition~\ref{prop:meanestimation}: given $X=x$, quantize to the vector $Z=z$ with,
    \begin{equation}
        Z = \begin{cases}
            \frac{-x}{||x||_2}, &\text{with probability }\frac{1}{2} - \frac{||x||_2}{2}, \\
            \frac{+x}{||x||_2}, & \text{with probability}\frac{1}{2} + \frac{||x||_2}{2}.
        \end{cases}
    \end{equation}
    Decompose $Z = S U$, where $U = X/||X||_2$ is the direction of $X$ and $S =\pm 1$ is the sign of the quantizer. We privatize $S$ to $\bar S = \mathsf K\circ S$ with,
    \begin{equation}
        \mathsf K = \frac{1}{e^\varepsilon (1-2p_{\min})+1} \begin{bmatrix}
            e^\varepsilon(1-p_{\min}) & 1-e^\varepsilon p_{\min} \\
            1-e^\varepsilon p_{\min} & e^\varepsilon(1-p_{\min}),
        \end{bmatrix}
    \end{equation}
    if $0\leq p_{\min} < e^{-\varepsilon}$, and with identity otherwise, where $p_{\min}$ is to be picked later. Let $\eta = \eta_\text{TV}(\mathsf K) = \min\{1,\,(e^\varepsilon-1)/(e^\varepsilon(1-2p_{\min})+1)\}$. Let the output of the privacy mechanism be used to compute $\bar Z = \bar S U$. For $B>0$, we sample $Y$ as,
    \begin{equation}
        Y \sim \text{Uniform}(y\in\mathbb R^d: \langle \bar z,y\rangle >0, ||y||_2=B).
    \end{equation}
    To achieve $\mathbb E[Y|X=x]=x$ with this estimator, it is shown in \cite{duchi2013local} that we need to pick,
    \begin{equation}
        B= \eta^{-1} \cdot \frac{d}{2}\cdot \frac{\sqrt{\pi}\Gamma(\frac{d-1}{2}+1)}{\Gamma(\frac{d}{2}+1)}.
    \end{equation}
    What remains is to choose $p_{\min}$ such that the whole procedure is $\varepsilon$-PML on $\mathcal P_c$. The conditional density of $Y=y$ given $X=x$ is given by,
    \begin{align}
        &p(y\mid x) = \\&p(y\mid U=u, \bar S=+1)\mathbb P[\bar S=+1\mid X=x] + p(y\mid U=u, \bar S=-1)\mathbb P[\bar S=-1\mid X=x],
    \end{align}
    and given $U=u$, and $\bar S = \bar s$, $Y$ is distributed uniformly on the half-sphere determined by $\text{sgn}\langle y,u\rangle > \bar s$. Let $A$ denote the surface of the sphere $\{y:||y||_2=B\}$. The probability of the signs is,
    \begin{equation}
        \mathbb P[\bar S=\bar s\mid X=x] = \frac{1}{2}+\frac{||x||_2}{2} \eta\,\text{sgn}(\bar s). 
    \end{equation}
    Therefore, the conditional density of $Y=y$ given $X=x$ is,
    \begin{equation}
        p(y\mid x) = \frac{2}{A} \left(\frac{1}{2}+\frac{\eta}{2}||x||_2\,\text{sgn}\langle y,x\rangle \right) \leq \frac{2}{A} \frac{1+\eta}{2},
    \end{equation}
    where the upper bound follows from $||x||_2\leq ||x||_p\leq 1$.
    The marginal $p_Y$ of $Y$ with input distribution $P_X$ with density $p_X$ is given by,
    \begin{equation}
       p_Y(y) = \int_{\mathcal B_d^p} \frac{2}{A} \left(\frac{1}{2}+\frac{\eta}{2}||x||_2\,\text{sgn}\langle y,x\rangle \right)p_X(x) \,d\mu_X = \frac{2}{A} \left(\frac{1}{2}+\frac{\eta}{2}\int_{\mathcal B_d^p}||x||_2\,\text{sgn}\langle y,x\rangle p_X(x) d\mu_X \right).
    \end{equation}
    We aim to lower-bound this term for any $P_X\in\mathcal P_c$. Note that we can decompose $P_X = c \mu_X + (1-c\mathbb V(\mathcal B_d^p))\bar P_X$, where $\bar P_X$ some distribution on $\mathcal B_d^p$ and $\mu_X$ is the Lebesgue measure restricted to the ball. Therefore, we can write,
    \begin{equation}
        \int_{\mathcal B_d^p}||x||_2\,\text{sgn}\langle y,x\rangle p_X(x) d\mu_X= c\int_{\mathcal B_d^p} ||x||_2\text{sgn}\langle y,x\rangle d\mu_X + (1-c\mathbb V(\mathcal B_d^p))\int_{\mathcal B_d^p}||x||_2\text{sgn}\langle y,x\rangle \bar p_X(x) d\mu_X. 
    \end{equation}
    The first term evaluates to zero by the symmetry of the unit ball. For the second term, note again that $0\leq ||x||_2\leq ||x||_p\leq 1$, $-1\leq\text{sgn}(\langle y,x\rangle)\leq1$, and $\bar p(x)$ is an arbitrary density, so,
    \begin{equation}
        \int_{\mathcal B_d^p}||x||_2\text{sgn}\langle y,x\rangle \bar p_X(x) d\mu_X \geq -1.
    \end{equation}
    We get that,
    \begin{equation}
        p_Y(y) \geq \frac{2}{A} \left(\frac{1}{2}-\frac{\eta}{2}(1-c\mathbb V(\mathcal B_d^p))\right).
    \end{equation}
    Hence, the pointwise maximal leakage of the procedure is bounded by,
    \begin{equation}
        e^{\ell(X\to y)} \leq  \frac{1+\eta}{1-\eta(1-c\mathbb V(\mathcal B_d^p))}.
    \end{equation}
    We require this term to be bounded by $e^\varepsilon$. In the case $\eta < 1$, rearranging yields,
    \begin{equation}
        \frac{1+\eta}{1-\eta(1-c\mathbb V(\mathcal B_d^p))} \leq e^\varepsilon \iff \eta \leq \frac{e^\varepsilon-1}{e^\varepsilon(1-c\mathbb V(\mathcal B_d^p))+1}.
    \end{equation}
    If $\eta =1$, then $\mathsf K$ is identity, and this aligns with the condition of identity in the definition of $\mathsf K$ if $p_{\min} = c\mathbb V(\mathcal B_d^p)/2$. 
    Hence, if we choose $p_{\min} = (c\mathbb V(\mathcal B_d^p))/2$ in the construction of $\mathsf K$, then the estimation procedure satisfies $\varepsilon$-PML on $\mathcal P_c$ for any $\varepsilon\geq 0$. With this choice, the procedure achieves, 
    \begin{equation}
        \mathfrak R^{\varepsilon,c}_n(\theta(\mathcal P),||\cdot||_2^2) \lesssim \frac{d}{n}\max\left\{1,\,\left(\frac{e^\varepsilon(1-c\mathbb V(\mathcal B_d^p))+1}{e^\varepsilon-1}\right)^2\right\} = \frac{d}{n\,\eta_\text{TV}^2(\varepsilon)},
    \end{equation}
    which is what we wanted to show.\qed

\section{Optimal High-Dimensional Mean Estimator}
\label{app:proofVertextTransEst}
The construction follows from the following lemma.
\begin{lemma}
\label{lem:couplinglemma}
    Let $\mathcal Z_0$ be a vertex-transitive $\rho$-covering of the unit sphere and define $\mathcal Z \coloneqq  \sec(\rho)\mathcal Z_0\implies \text{conv}(\mathcal Z) \supseteq \mathcal B_d^2$. Then there exists a coupling $\pi(x,y)$ between $X \sim \text{Unif}(\mathcal B_d^2)$  $Z\sim\text{Unif}(\mathcal Z)$ such that  and $\mathbb E[Z\mid X] = X$. 
\end{lemma}
To construct the estimator, let $P_{Z|X=x}$ be the conditional distribution induced by the coupling in Lemma \ref{lem:couplinglemma}. For each $X=x_i$, sample $Z_i\sim P_{Z|X=x_i}$ and compute $\hat \theta(Z^n) = \frac{1}{n}\sum_{i=1}^n Z_i$. This is an unbiased estimator by the definition of $P_{Z|X}$. Further, it achieves $\mathsf{Var}(\hat \theta(Z^n))\leq \sec^2(\rho)/n$. What remains is to show the privacy guarantee. To this end, let $f_X$ denote the density associated with some $P_X\in \mathcal P_c$. We see that,
\begin{align}
    \mathbb P[Z=z] &= \int_{\mathcal B_d^2} \mathbb P[Z=z\mid X=x]dP_X(x)
     \geq c\mathbb V(\mathcal B_d^2)\int_{\mathcal B_d^2}\mathbb P[Z=z|X=x]dU(x) \\[.5em]
    &= c\mathbb V(\mathcal B_d^2)\int_{\mathcal B_d^2}\frac{\pi(x,z)}{dU_X(x)}dU_X(x) = \frac{c\mathbb V(\mathcal B_d^2)}{M},
\end{align}
where the last step follows from the definition of the coupling $\pi(x,z)$, which has uniform marginals. Since $Z$ is discrete with minimum probability $\min_{z}\mathbb P[Z=z] = c\mathbb V(\mathcal B_d^2)/M$, an estimator constructed in this way satisfies $(\log M-\log c\mathbb V(\mathcal B_d^2))$-PML. Note that this only proves the \emph{existence} of such an estimator. For implementations, the mapping resulting from the coupling in Lemma \ref{lem:couplinglemma} needs to be made explicit, or approximated.

\subsection*{Proof of Lemma \ref{lem:couplinglemma}}
According to Strassen's Theorem \cite{strassen1965existence}, the suggested coupling exists if and only if $\mathbb E[u(X)]\leq \mathbb E[u(Z)]$ for all convex functions $u$ \cite{armbruster2016short}. Fix any such convex $u$ and let $G$ denote a transitive group on $\mathcal Z$ (at least one such group exists by the assumptions made in Remark \ref{rem:vertextransEst}). Define,
\begin{equation}
   \tilde u(x) \coloneqq \frac{1}{|G|}\sum_{g\in G}u(gx).
\end{equation}
Since $\mathcal B_d^2$ is invariant to rotation and reflection about the origin, and each $g\in G$ is a bijection, we have that,
\begin{equation}
    \int_{\mathcal B_d^2} u(x) dP_X(x) = \int\tilde u(x)dP_X(x).
\end{equation}
By definition, each $z_k$ is an extreme point of $\text{conv}(\mathcal Z)$. Fix any $z_k\in\mathcal Z$. We have,
\begin{equation}
    \tilde u(z_k) = \frac{1}{|G|}\sum_{g\in G}u(gz_k).
\end{equation}
We have that, $|\{g\mid gz_i = z_k\}| = |G_{z_k}|$ for any $i\in[M]$, where $G_{z_k}$ denotes the stabilizer of $z_k$, and by the Orbit-Stabilizer Theorem \cite[Proposition 6.9.2]{artin2010algebra}, we have,
\begin{equation}
    |G| = |G_{z_k}|\cdot |G z_k| \implies |\{g\mid gz_k=z_i\}| = \frac{|G|}{M}.
\end{equation}
Hence $gz_k=z_i$ for exactly $|G|/M$ elements of $g$, hence we get,
\begin{equation}
    \tilde u(z_k) = \frac{1}{|G|}\sum_{g\in G}u(gz_k) =  \frac{1}{|G|}\sum_{j=1}^M\frac{|G|}{M}u(z_j) = \frac{1}{M}\sum_{j=1}^M u(gz_j) \eqqcolon C = \text{const.}
\end{equation}
Because $\tilde u$ is convex, the maximum value of $\tilde u$ on $\text{conv}(\mathcal Z)$ is attained on an extreme points. This yields,
\begin{equation}
    \mathbb E[u(X)] = \int_{\mathcal B_d^2} \tilde u(x)dP_X(x) \leq \max_{z\in\mathcal Z}\tilde u(z) = C = \frac{1}{M}\sum_{j=1}^Mu(z_j) = \mathbb E[u(Z)]. 
\end{equation}
Since $u$ was picked arbitrarily, Strassen's theorem applies, and we have shown the statement. \qed